\documentclass[sigplan, screen]{acmart}

\usepackage{mathpartir}
\usepackage{float}
\usepackage{natbib}
\usepackage{amsmath}
\usepackage{amsthm}
\usepackage{xcolor}
\usepackage{xspace}
\usepackage{stmaryrd}
\usepackage{fancyvrb}
\usepackage{soul}
\usepackage{scalerel}
\usepackage{thm-restate}
\usepackage{tikz}
\newcommand{\ImageWidth}{8cm}
\usepackage{paralist}
\usetikzlibrary{decorations.pathreplacing,positioning, arrows.meta}

\theoremstyle{definition}
\newtheorem{definition}{Definition}[section]
\newtheorem{lemma}{Lemma}[section]

\AtEndPreamble{
\theoremstyle{acmdefinition}
\newtheorem{remark}[theorem]{Remark}
}

\usepackage{listings}
\lstdefinelanguage{Links}{%
  morekeywords={typename, fun, linfun, op, var, if, this, true, false, else, case, switch, handle,
    handler, shallowhandler, open, do, sig, new, send, receive, spawnAt, spawn,
module, request, accept, try, as, in, otherwise, catch, offer, select, raise,
fork, spawnClient, cancel, query, nested, flat, nonsequenced, current,
valid, from, to, sequenced, where, with, set, insert, update, delete,
between, and, table, using, valid_time, vt_insert, values, database,
forever, for, vtCurrent
},%
  sensitive=t, %
  comment=[l]{\#\ },%
  escapeinside={(*}{*)},%
  morestring=[d]{"},%
  keywordstyle=\color{blue},
  showstringspaces=false
 }

 \usepackage[T1]{fontenc}
\usepackage[scaled=0.85]{beramono}

\newcommand{\header}[1]{
  \begin{flushleft}
    \textbf{#1}
\end{flushleft}}

\newcommand{\headersig}[2]{
  \begin{flushleft}
    \textbf{#1} \hfill \framebox{#2}
\end{flushleft}}

\newcommand{\headerarg}[2]{
  \begin{flushleft}
    \textbf{#1} \hfill {#2}
\end{flushleft}}
\newcommand{\linq}{\ensuremath{\lambda_{\textsf{LINQ}}}\xspace}
\newcommand{\vlinq}{\ensuremath{\lambda_{\textsf{VLINQ}}}\xspace}
\newcommand{\tlinq}{\ensuremath{\lambda_{\textsf{TLINQ}}}\xspace}

\newcommand{\mkwd}[1]{\textsf{#1}}
\newcommand{\calcwd}[1]{\textbf{\textsf{#1}}}
\newcommand{\stringty}{\mkwd{String}}
\newcommand{\intty}{\mkwd{Int}}
\newcommand{\boolty}{\mkwd{Bool}}

\newcommand{\tablety}[1]{\mkwd{Table}(#1)}
\newcommand{\query}[1]{\calcwd{query} \: #1}
\newcommand{\queryzero}{\calcwd{query}}
\newcommand{\join}[1]{\calcwd{join} \: #1}
\newcommand{\joinzero}{\calcwd{join}}

\newcommand{\tablevar}{t}
\newcommand{\recordty}[1]{(#1)}
\newcommand{\recordtylr}[1]{\left(#1\right)}
\newcommand{\recordterm}{\recordty}
\newcommand{\recordtermlr}{\recordtylr}

\newcommand{\app}{\;}
\newcommand{\comparrow}{\leftarrow}
\newcommand{\dbcomparrow}{\Leftarrow}
\newcommand{\forcomp}[3]{\calcwd{for} \; (#1 \comparrow #2) \; #3}
\newcommand{\forcomptwo}[2]{\calcwd{for} \; (#1 \comparrow #2)}

\newcommand{\dbinsert}[2]{\calcwd{insert} \; #1 \; \calcwd{values} \; #2}

\newcommand{\dbinsertseq}[2]{\calcwd{insert} \; \calcwd{sequenced} \; #1 \; \calcwd{values} \; #2}
\newcommand{\dbinsertseqann}[3][(\ell_i : \tya_i)_i]{\calcwd{insert}^{#1} \;
\calcwd{sequenced} \; #2 \; \calcwd{values} \; #3}
\newcommand{\dbinsertann}[3][\recordterm{\ell_i :
    \tya_i}_i]{\calcwd{insert}^{#1}
\; #2 \; \calcwd{values} \; #3}
\newcommand{\dbupdatebetween}[6]{\calcwd{update} \; \calcwd{sequenced} \; (#1 \dbcomparrow #2) \; \calcwd{between} \; #3 \; \calcwd{and} \; #4 \; \calcwd{where} \; #5 \; \calcwd{set} \; #6}

\newcommand{\dbupdate}[4]{\calcwd{update} \; (#1 \dbcomparrow #2) \; \calcwd{where} \; #3 \; \calcwd{set} \; #4}

\newcommand{\dbupdatenonseq}[6]
    {\calcwd{update} \; \calcwd{nonsequenced} \;
     (#1 \dbcomparrow #2) \; \calcwd{where} \; #3 \; \calcwd{set} \; #4
       \; \calcwd{valid} \; \calcwd{from} \; #5 \; \calcwd{to} \; #6}

\newcommand{\dbupdatenonseqtwo}[2]
    {\calcwd{update} \; \calcwd{nonsequenced} \;
     (#1 \dbcomparrow #2)}

\newcommand{\dbupdateann}[5][\recordterm{\ell_i : \tya_i}_i]
    {\calcwd{update}^{#1} \; (#2 \dbcomparrow #3) \; \calcwd{where} \; #4 \; \calcwd{set} \; #5}
\newcommand{\dbupdatetwo}[2]{\calcwd{update} \; (#1 \dbcomparrow #2)}
\newcommand{\dbdelete}[3]{\calcwd{delete} \; (#1 \dbcomparrow #2) \; \calcwd{where} \; #3}

\newcommand{\dbdeletenonseq}[3]{\calcwd{delete} \; \calcwd{nonsequenced}
\; (#1 \dbcomparrow #2) \; \calcwd{where} \; #3}
\newcommand{\dbdeletenonseqann}[4][(\ell_i : \tya_i)_i]
{\calcwd{delete}^{#1} \; \calcwd{nonsequenced}
\; (#2 \dbcomparrow #3) \; \calcwd{where} \; #4}

\newcommand{\dbdeleteann}[4][\recordterm{\ell_i :
\tya_i}_i]{\calcwd{delete}^{#1} \; (#2 \dbcomparrow #3) \; \calcwd{where} \; #4}
\newcommand{\dbdeletetwo}[2]{\calcwd{delete} \; (#1 \dbcomparrow #2)}
\newcommand{\dbdeletebetween}[5]{\calcwd{delete} \; \calcwd{sequenced} \; (#1 \dbcomparrow #2) \; \calcwd{between} \; #3 \; \calcwd{and} \; #4 \; \calcwd{where} \; #5}
\newcommand
    {\dbdeletebetweenann}
    [6][(\ell_i : \tya_i)_i]
    {\calcwd{delete}^{#1} \; \calcwd{sequenced} \; (#2 \dbcomparrow #3) \; \calcwd{between} \; #4 \; \calcwd{and} \; #5 \; \calcwd{where} \; #6}
\newcommand{\doubleplus}{+\kern-1.3ex+\kern0.8ex}

\newcommand{\timevar}{\iota}
\newcommand{\ite}[3]{\calcwd{if} \; #1 \; \calcwd{then} \; #2 \; \calcwd{else} \; #3}
\newcommand{\timety}{\mkwd{Time}}

\newcommand{\ttrue}{\calcwd{true}}
\newcommand{\ffalse}{\calcwd{false}}

\newcommand{\db}{\Delta}

\newcommand{\schema}{\Sigma}

\newcommand{\oseq}[1]{\overrightarrow{#1}}
\newcommand{\seq}[1]{\widetilde{#1}}
\newcommand{\pure}{\emptyset}

\newcommand{\effann}[1]{\: ! \: #1}
\newcommand{\now}{\calcwd{now}\xspace}
\newcommand{\beginningoftime}{\ensuremath{-\infty}}
\newcommand{\forever}{\ensuremath{\infty}}

\newcommand{\bl}{\begin{array}{l}}
\newcommand{\blt}{\begin{array}[t]{l}}
\newcommand{\el}{\end{array}}
\newcommand{\letinone}[1]{\calcwd{let} \: #1  =}
\newcommand{\letintwo}[2]{\calcwd{let} \: #1  = #2 \: \calcwd{in}}

\newcommand{\defeq}{\triangleq}

\newcommand{\iscurrent}[1]{\mkwd{isCurrent}(#1)}
\newcommand{\currentat}[2]{\mkwd{currentAt}(#1, #2)}
\newcommand{\var}[1]{\ensuremath{\mathit{#1}}}
\newcommand{\field}{\var}

\newcommand{\recordwith}[3]{(#1 \: \calcwd{with} \: #2 = #3)}
\newcommand{\recordwithtwo}[2]{(#1 \: \calcwd{with} \: #2)}

\newcommand{\effread}{\calcwd{read}\xspace}
\newcommand{\effwrite}{\calcwd{write}\xspace}

\newcommand{\effs}{E}
\newcommand{\eff}{e}
\newcommand{\effset}[1]{\{#1\}}

\newcommand{\validtimety}[1]{\mkwd{ValidTime}(#1)}
\newcommand{\transtimety}[1]{\mkwd{TransactionTime}(#1)}
\newcommand{\data}[1]{\calcwd{data} \: #1}

\newcommand{\timestart}[1]{\calcwd{start} \: #1}
\newcommand{\timeend}[1]{\calcwd{end} \: #1}
\newcommand{\flattenbag}[1]{\hat{\biguplus}#1}

\newcommand{\project}[2]{#1{.}#2}

\newcommand{\where}[2]{\mkwd{where} \: #1 \: #2}
\newcommand{\whereone}[1]{\mkwd{where} \: (#1)}
\newcommand{\calcwhereone}[1]{\calcwd{where} \: (#1)}

\newcommand{\spacerow}{\vspace{0.6em}}

\newcommand{\etaexp}[2]{\eta(#1, #2)}
\newcommand{\denot}[1]{\hat{#1}}
\newcommand{\opsymb}{\odot}
\newcommand{\langop}[1]{\opsymb \{ #1 \}}
\newcommand{\denotlangop}[1]{\denot{\opsymb} \{ #1 \}}

\newsavebox{\lXparen}
\savebox{\lXparen}{$\llparenthesis$}
\newsavebox{\rXparen}
\savebox{\rXparen}{$\rrparenthesis$}

\newcommand{\hasqtype}[1]{#1 :: \mkwd{QType}}
\newcommand{\hasfqtype}[1]{#1 :: \mkwd{FQType}}

\newcommand{\normfor}[3]{\calcwd{for} \: (#1) \: \mkwd{where} \: #2 \: #3}

\newcommand{\tmeval}{\Downarrow}
\newcommand{\tmevaltwo}[3][\db, \timevar]{#2 \tmeval_{#1} #3}
\newcommand{\ttmevaltwo}[3][\db, \timevar]{#2 \tmeval^{\mkwd{T}}_{#1} #3}
\newcommand{\vtmevaltwo}[3][\db, \timevar]{#2 \tmeval^{\mkwd{V}}_{#1} #3}

\newcommand{\puretmeval}[1][\timevar]{\tmeval_{#1}}
\newcommand{\puretmevaltwo}[3][\timevar]{#2 \tmeval^{\star}_{#1} #3}
\newcommand{\readtmevaltwo}[3][\db, \timevar]{#2 \tmeval^{\star}_{#1} #3}

\newcommand{\tya}{A}
\newcommand{\tyb}{B}
\newcommand{\tma}{M}
\newcommand{\tmb}{N}
\newcommand{\tmc}{L}
\newcommand{\vala}{V}
\newcommand{\valb}{W}
\newcommand{\midspace}{\, \mid \,}
\newcommand{\basety}{C}
\newcommand{\tyfun}[3]{#1 \to^{#3} #2}
\newcommand{\efflet}[3]{\calcwd{let} \; #1 = #2 \; \calcwd{in} \; #3}
\newcommand{\effletone}[1]{\calcwd{let} \; #1 =}

\newcommand{\tblvar}{\tablevar}
\newcommand{\const}{c}
\newcommand{\fun}[2]{\lambda #1.#2}

\newcommand{\retpair}[2]{(#1, #2)}

\newcommand{\extendenv}[3]{#1[#2 \mapsto #3]}

\newcommand{\ttrans}[1]{\llbracket #1 \rrbracket}
\newcommand{\tvaltrans}{\ttrans}
\newcommand{\ttmtrans}{\ttrans}
\newcommand{\vtrans}[1]{\llparenthesis #1 \rrparenthesis}

\newcommand{\tseq}[4]{#1 \vdash #2 {:} #3 \,!\, #4}
\newcommand{\tyenv}{\Gamma}
\newcommand{\dbrow}[3]{#1 ^{[#2, #3)}}
\newcommand{\fielddata}[1]{\mathit{data} = #1 }
\newcommand{\fieldstart}[1]{\mathit{start} = #1 }
\newcommand{\fieldend}[1]{\mathit{end} = #1}
\newcommand{\tyfielddata}[1]{\mathit{data} {:} #1 }
\newcommand{\tyfieldstart}[1]{\mathit{start} {:} #1 }
\newcommand{\tyfieldend}[1]{\mathit{end} {:} #1 }

\newcommand{\totheleft}[1]{\begin{flushleft}#1\end{flushleft}}
\newenvironment{proofcase}[1]
  {\totheleft{\textbf{Case } {#1}}}
  {}
\newenvironment{subcase}[1]
  {\totheleft{\textbf{Subcase } {#1}}}
  {}
\newenvironment{subsubcase}[1]
  {\totheleft{\textbf{Subsubcase } {#1}}}
  {}

\newcommand{\subst}[3]{#1 \{ #2 / #3 \}}
\newcommand{\set}[1]{\{#1\}}
\newcommand{\bag}[1]{\lbag #1 \rbag}
\newcommand{\baglr}[1]{\scaleleftright[1.5ex]{\Lbag}{ #1 }{\Rbag}}
\newcommand{\bagty}[1]{\mkwd{Bag}(#1)}
\newcommand{\get}[1]{\calcwd{get} \; #1}
\newcommand{\getph}[1]{\calcwd{get}\phantom{\scriptsize{\mathsf{v}}} \; #1 }
\newcommand{\getv}[1]{\calcwd{get}_{\mathsf{v}} \; #1}
\newcommand{\bagunion}{\uplus}
\newcommand{\denotbagunion}{\mathop{\denot{\bagunion}}}
\newcommand{\baguniontwo}[2]{#1 \bagunion #2}
\newcommand{\emptybag}{\bag{~}}
\newcommand{\getann}[2]{\calcwd{get}^{#1} \: #2}

\newcommand{\restrict}[3]{\mkwd{restrict}(#1, #2, #3)}
\newcommand{\recordplus}{\oplus}
\newcommand{\recordplustwo}[2]{#1 \recordplus #2}

\newcommand{\without}{\backslash}

\newcommand{\at}[2]{\mkwd{at}(#1, #2)}
\newcommand{\current}[1]{\mkwd{current}(#1)}
\newcommand{\lift}[2]{\mkwd{lift}(#1, #2)}
\newcommand{\greatest}[1]{\calcwd{greatest}(#1)}
\newcommand{\greatestzero}{\calcwd{greatest}\xspace}
\newcommand{\leastzero}{\calcwd{least}\xspace}
\newcommand{\least}[1]{\calcwd{least}(#1)}
\newcommand{\greatesttwo}[2]{\calcwd{greatest}(#1, #2)}
\newcommand{\leasttwo}[2]{\calcwd{least}(#1, #2)}
\newcommand{\qqquad}{\qquad \quad}

\newcommand{\querytm}{Q}
\newcommand{\querycomp}{K}
\newcommand{\querygen}{G}
\newcommand{\querybasetm}{P}
\newcommand{\queryrecordtm}{R}
\newcommand{\querynormtm}{S}

\newcommand{\normtransl}{\|}
\newcommand{\normtransr}{\|}
\newcommand{\normtrans}[1]{\normtransl #1 \normtransr}

\newcommand{\secref}[1]{\S\ref{#1}}
\newcommand{\secrefp}[1]{(\secref{#1})}
\newcommand{\flatten}[1]{{\downarrow}{#1}}

\newcommand{\maxts}[1]{\mkwd{max}(#1)}
\newcommand{\wf}[1]{\mkwd{wf}(#1)}
 
\renewcommand\footnotetextcopyrightpermission[1]{}

\AtBeginDocument{%
  }

\begin{document}

\title{Language-Integrated Query for Temporal Data}
\subtitle{(Extended version)}
\author{Simon Fowler}
\orcid{0000-0001-5143-5475}
\affiliation{
    \institution{University of Glasgow}
    \country{UK}
}
\email{simon.fowler@glasgow.ac.uk}

\author{Vashti Galpin}
\orcid{0000-0001-8914-1122}
\affiliation{
    \institution{University of Edinburgh}
    \country{UK}
}
\email{vashti.galpin@ed.ac.uk}

\author{James Cheney}
\orcid{0000-0002-1307-9286}
\affiliation{
    \institution{University of Edinburgh}
    \country{UK}
}
\email{jcheney@inf.ed.ac.uk}

\begin{abstract}
  Modern applications often manage time-varying data.  Despite decades
  of research on temporal databases, which culminated in the addition of
  temporal data operations into the SQL:2011 standard, temporal
  data query and manipulation operations are unavailable in most
  mainstream database management systems, leaving developers with
  the unenviable task of implementing such functionality from scratch.

  In this paper, we extend \emph{language-integrated query} to support
  writing temporal queries and updates in a uniform host language,
  with the language performing the required rewriting to
  emulate temporal capabilities automatically on any standard relational
  database.
    We introduce two core languages, $\tlinq$ and $\vlinq$, for manipulating
    transaction time and valid time data respectively, and formalise
    existing implementation strategies by giving provably correct
    semantics-preserving translations into a non-temporal core language,
    $\linq$.
    We show how existing work on query normalisation supports a surprisingly
    simple implementation strategy for \emph{sequenced joins}.  We implement our
    approach in the Links programming language, and describe a non-trivial case
    study based on curating COVID-19 statistics.
\end{abstract}



\maketitle

\section{Introduction}\label{sec:intro}

Most interesting programs access or query data stored persistently,
often in a database.  Relational database management systems (RDBMSs)
are the most popular option and provide a standard
domain-specific language, SQL, for querying and modifying the data.
Ideally, programmers can express the desired queries or updates
declaratively in SQL and leave the database to decide how to answer
queries or perform updates efficiently and safely (e.g. in the
presence of concurrent accesses), but there are many pitfalls arising
from interfacing with SQL from a general-purpose language, leading to
the well-known \emph{impedance mismatch} problem~\cite{copeland-maier:1984}.
These difficulties range from run-time failures due to the
generation of queries as unchecked SQL strings at runtime, to serious
security vulnerabilities like SQL injection attacks~\cite{sql-injection}.

Among the most successful approaches to overcome the above challenges,
and the approach we build upon in this paper, is
\emph{language-integrated query}, exemplified by Microsoft's popular
LINQ for .NET~\cite{meijer:sigmod,Syme06} and in a number of other
language designs such as Links~\cite{CooperLWY06,lindley12tldi} and
libraries such as Quill~\cite{quill}.  Within this design space we
focus on a family of techniques derived from foundational work by
Buneman et al.~\cite{BNTW95} on the \emph{nested relational calculus}
(NRC), a core query language with monadic collection types; work by
Wong~\cite{wong96jcss} on rewriting NRC expressions to normal forms
that can be translated to SQL, which forms the basis of the approach
taken in Links~\cite{Cooper09,lindley12tldi} and has been adapted to
F\#~\cite{cheney13icfp}.

Many interesting database applications involve data that changes over
time.  Perhaps inevitably, \emph{temporal data
  management}~\cite{jensen99tkde} has a long history.
Temporal databases provide powerful features for querying and
modifying data that changes over time, and are particularly suitable
for modeling time-varying phenomena, such as enterprise data or
evolving scientific knowledge, and supporting auditing and transparency
about how the current state of the data was achieved, for example in
financial or scientific settings.

To illustrate how temporal databases work and why they are useful,
consider the following toy example: a temporal to-do list.  A temporal
database can be conceptualised abstractly as a \emph{function} mapping
each possible time instant (e.g. times of day) to a
database state~\cite{jensen94is}.  For efficiency in the common case where most of the
data is unchanged most of the time, temporal databases are often
represented by augmenting each row with an \emph{interval timestamp} indicating the
time period when the row is present.  For technical reasons,
\emph{closed-open} intervals $[start,end)$ representing the times
$\var{start} \leq t < \var{end}$ are typically used~\cite{Snodgrass99:book}.

In our temporal to-do list, the table at each time instant has fields
``task'', a string field, and ``done'', a Boolean field.  Additional
fields ``start'' and ``end'' record the endpoints of the time interval
during which each row is to be considered part of the table.  An
end time of $\forever\xspace$ (``forever'') reflects that there is no
(currently known) end time and in the absence of other changes, the
row is considered present from the start time onwards.
For example, the table:

{\footnotesize
    \begin{center}
\begin{tabular}{| l | c || c | c |}
    \hline
    task & done & start & end \\ \hline
    Go shopping & \ttrue & 11:00 & $\infty$ \\
    \hline
    Cook dinner & \ffalse & 11:00 & 17:30 \\
    \hline
    Walk the dog & \ffalse & 11:00 & $\infty$ \\
    \hline
     Cook dinner & \ttrue & 17:30 & $\infty$ \\
    \hline
    Watch TV & \ffalse & 11:00 & 19:00 \\
    \hline
\end{tabular}
\end{center}
}
\noindent represents a temporal table where all four tasks were added at 11:00,
with ``Go shopping'' being complete and the others incomplete; at
17:30 ``Cook dinner'' was marked ``done'' from then onwards, and at
19:00 ``Watch TV'' was removed from the table without being completed.
Technically,
note that this example interprets the time annotations as
\emph{transaction time}, that is, the times indicate when certain data
was in the database; there is another dimension, \emph{valid time},
and we will discuss both dimensions in
greater detail later on.

The problems of querying and updating temporal databases have been
well-studied, leading to the landmark language design
TSQL2~\cite{tsql2} extending SQL.  However, despite decades of effort,
only limited elements of TSQL2 were eventually incorporated into the
SQL:2011 standard~\cite{kulkarni12sigmodrecord} and these features
have not yet been widely adopted.  Directly implementing temporal
queries in SQL is possible, but painful: a TSQL2-style query or
update operation may grow by a factor of 10 or more when translated to
plain SQL, which leaves plenty of scope for error, and thus these powerful
capabilities remain outside the grasp of non-experts.
In this paper we take first steps towards reconciling temporal data management
with language-integrated query based on query normalisation.  We propose
supporting temporal capabilities by translation to ordinary language-integrated
query and hypothesise that this approach can make temporal data management
safer, easier to use and more accessible to non-experts than the current state
of affairs.  As an initial test of this hypothesis we present a high-level
design, a working implementation, and detail our experience with a nontrivial
case study.

Although both language-integrated query and temporal databases are now
well-studied topics, we believe that their combination has never been considered
before.  Doing so has a number of potential benefits, including making
the power of well-studied language designs such as TSQL2 more
accessible to non-expert programmers, and providing a high-level abstraction that can be implemented efficiently in different ways.
Our interest is particularly
motivated by the needs of scientific database development, where data
versioning for accountability and research integrity are very
important needs that are not well-supported by conventional database
systems~\cite{buneman18sigmodrecord}.  Temporal data management has the potential to become a
``killer app'' for language-integrated query, and this paper takes a
first but significant step towards this goal.

The overarching contribution of this paper is the first extension of
language-integrated query to support transaction time and valid time
temporal data.

Concretely, we make three main contributions:

\begin{compactenum}
\item Based on $\linq$~\secrefp{sec:linq}, a formalism based on the
  Nested Relational Calculus (NRC)~\cite{RothKS88}, we introduce typed
  $\lambda$-calculi to model queries and modifications on transaction
  time~\secrefp{sec:tlinq} and valid time~\secrefp{sec:vlinq}
  databases.  We give
  semantics-preserving translations to $\linq$ for both.
\item We show how existing work on query normalisation allows a
  surprisingly straightforward implementation strategy for \emph{sequenced
  joins}~\secrefp{sec:temporal-joins}.
\item We implement our constructs in the Links functional web
  programming language, and describe a case study based on curating
  COVID-19 data~\secrefp{sec:implementation}.
\end{compactenum}

\sloppypar
Although the concepts behind translating temporal queries and updates into
non-temporal core languages are well known~\cite{Snodgrass99:book}, our core
calculi $\tlinq$ and $\vlinq$ are novel and aid us in showing (to the best of
our knowledge) the first correctness results for the translations.

We relegate several details and all proofs to the appendices.

\section{Background: Language-Integrated Query}\label{sec:linq}

\begin{figure}[tb]
{\footnotesize
\[
    \begin{array}{lrcl}
        \text{Types} & \tya, \tyb & ::= &
            \basety \midspace \tyfun{\tya}{\tyb}{\effs} \midspace \bagty{\tya}
            \midspace \recordty{\seq{\ell : \tya}} \midspace \tablety{\tya} \\
        \text{Base types} \!\!\!\!\!\! & \basety & ::= & \stringty \midspace \intty \midspace
        \boolty \midspace \timety \\
        \text{Effects} & \eff & ::= & \effread \midspace \effwrite \\
        \text{Effect sets} \!\!\!\!\!\!& \effs \\
        \text{Terms} & \tmc, \tma, \tmb & ::= & x \midspace \const \midspace \tblvar \\
                     & & \midspace & \fun{x}{\tma} \midspace \tma \app \tmb \midspace
                     \langop{\oseq{\tma}} \\
                     & & \midspace & \ite{\tmc}{\tma}{\tmb} \\
                     & & \midspace & \bag{~} \midspace \bag{\tma} \midspace \baguniontwo{\tma}{\tmb} \midspace \forcomp{x}{\tma}{\tmb} \\
                     & & \midspace & \recordterm{\seq{\ell = \tma}} \midspace
                     \project{\tma}{\ell} \midspace \now \\
                     & & \midspace & \query{\tma} \midspace \get{\tma} \midspace
                                     \dbinsert{\tma}{\tmb} \\
                     & & \midspace & \dbupdate{x}{\tmc}{\tma}{(\seq{\ell =
                     \tmb})} \\
                     & & \midspace & \dbdelete{x}{\tma}{\tmb}
    \end{array}
\]
}
\vspace{-0.75em}
\caption{Syntax of $\linq$}\label{fig:syntax}
\vspace{-1.5em}
\end{figure}

We begin by introducing a basic $\lambda$-calculus, called $\linq$, to model
language-integrated query in non-temporal databases. Our calculus is based
heavily on the Nested Relational Calculus~\cite{RothKS88}, with support for database
modifications heavily inspired by the calculus of~\citet{FehrenbachC18}.
The calculus uses a
type-and-effect system to ensure database accesses can occur in `safe' places,
i.e., that we do not attempt to perform a modification operation in the middle
of a query.
Effects include $\effread$ (denoting a read from a database) and $\effwrite$
(denoting an update to the database). %
Types $\tya, \tyb$ include base types $\basety$, effect-annotated function types
$\tyfun{\tya}{\tyb}{\effs}$, unordered collection types $\bagty{\tya}$, record
types $\recordty{\ell_i : \tya_i}_i$ denoting a record with labels $\ell_i$ and
types $\tya_i$, and handles $\tablety{\tya}$ for tables containing records of
type $\tya$. A record is a \emph{base record} if it contains only
fields of base type. We assume that the base types includes at least $\boolty$
and the $\timety$ type, which denotes (abstract) timestamps.

Basic terms include variables $x$, constants $c$, table handles $\tblvar$, functions
$\fun{x}{\tma}$, application $\tma \app \tmb$, n-ary operations
$\langop{\oseq{\tma}}$, and conditionals $\ite{\tmc}{\tma}{\tmb}$.
We assume that the set of operations contains the usual unary and
binary relation operators, as well as the $n$-ary operations
$\greatestzero$ and $\leastzero$ on timestamps which
return their largest and smallest arguments respectively. We assume that the set
of constants contains timestamps $\timevar$ of type $\timety$, and two distinguished timestamps
$\beginningoftime$ and $\forever$ of
type $\timety$, which denote the minimum and maximum timestamps respectively.
The calculus also includes the empty multiset constructor $\bag{~}$;
the singleton multiset constructor $\bag{\tma}$; multiset union $\tma
\bagunion \tmb$; and comprehensions $\forcomp{x}{\tma}{\tmb}$.
We write $\bag{\tma_1, \ldots, \tma_n}$ as sugar for $\bag{\tma_1} \bagunion
\ldots \bagunion \bag{\tma_n}$.
We also have
records $\recordterm{\ell_i = \tma_i}_i$ and projection $\project{\tma}{\ell}$.
Term $\now$ retrieves the current time.

We write $\efflet{x}{\tma}{\tmb}$ as the usual syntactic sugar for
$(\lambda x . \tmb) \app \tma$, and $\tma; \tmb$ as sugar for
$(\lambda x .  \tmb) \tma$ for some fresh $x$.
We also define $\where{\tma}{\tmb}$ as sugar for $\ite{\tma}{\tmb}{\bag{~}}$.
We denote unordered collections with a tilde (e.g., $\seq{M}$), and ordered
sequences with an arrow (e.g., $\oseq{M}$).

The $\get{\tma}$ term retrieves the
contents of a table into a bag; $\dbinsert{\tma}{\tmb}$ inserts values
 $\tmb$ into table $\tma$;
$\dbupdate{x}{\tmc}{\tma}{(\ell_i = \tmb_i)_i}$ updates table $\tmc$,
updating the fields $\ell_i$ to $\tmb_i$ of each record $x$ satisfying predicate
$\tma$. The $\dbdelete{x}{\tma}{\tmb}$ term removes those records $x$ in
table $\tma$ satisfying predicate $\tmb$.

\begin{figure*}
{\footnotesize
  ~\header{Term typing} \hfill \framebox{$\tseq{\tyenv}{\tma}{\tya}{\effs}$}
  \vspace{-1em}
  \begin{mathpar}
    \inferrule
    [T-Query]
    { \tseq{\tyenv}{\tma}{\bagty{\tya}}{\effs}\!\!\! \\
      \hasqtype{\tya}\!\!\! \\
      \effs {\subseteq} \effset{\effread}
    }
    { \tseq{\tyenv}{\query{\tma}}{\bagty{\tya}}{\effs} }

    \inferrule
    [T-Now]
    { }
    { \tseq{\tyenv}{\now}{\timety}{\pure} }

    \inferrule
    [T-Get]
    { \tseq{\tyenv}{\tma}{\tablety{\tya}}{\effs}}
    { \tseq
        {\tyenv\!}
        {\!\get{\tma}}
        {\bagty{\tya}}
        {\effset{\effread} \cup \effs}
    }

    \inferrule
    [T-Insert]
    { \tseq{\tyenv}{\tma}{\tablety{\tya}}{\effs} \\
      \tseq{\tyenv}{\tmb}{\bagty{\tya}}{\pure}
    }
    { \tseq{\tyenv}{\dbinsert{\tma}{\tmb}}{\recordty{}}{\effset{\effwrite} \cup \effs} }

    \inferrule
    [T-Update]
    { \tseq{\tyenv}{\tmc}{\tablety{\tya}}{\effs} \\\\
      \tya = \recordty{\ell_i : \tyb_i}_{i \in I} \\
      \tseq{\tyenv, x : \tya}{\tma}{\boolty}{\pure} \\
      (j \in I \wedge \tseq{\tyenv, x : \tya}{\tmb_j}{\tyb_j}{\pure})_{j \in J}
    }
    { \tseq
        {\tyenv}
        {\dbupdate{x}{\tmc}{\tma}{(\ell_j = \tmb_j)_{j \in J}}}
        {\recordty{}}
        {\effset{\effwrite} \cup \effs }
    }

    \inferrule
    [T-Delete]
    { \tseq{\tyenv}{\tma}{\tablety{\tya}}{\effs} \\
      \tseq{\tyenv, x : \tya}{\tmb}{\boolty}{\pure}
    }
    { \tseq
        {\tyenv}
        {\dbdelete{x}{\tma}{\tmb}}
        {\recordty{}}
        {\effset{\effwrite} \cup \effs }
    }
  \end{mathpar}

~\header{Query typing} \hfill \framebox{\hasqtype{A}}
  \vspace{-2em}
  \begin{mathpar}
    \inferrule
    { }
    { \hasqtype{\basety} }

    \inferrule
    { (\hasqtype{A_i})_i }
    { \hasqtype{(\ell_i : A_i)_i} }

    \inferrule
    { \hasqtype{A} }
    { \hasqtype{\bagty{A}} }
  \end{mathpar}
}

    \caption{Typing rules for $\linq$ (selected)}
    \label{fig:linq-typing}
\end{figure*}

\subsection{Typing rules}
Figure~\ref{fig:linq-typing} shows the typing rules for $\linq$;
the typing judgement has the shape $\tseq{\tyenv}{\tma}{\tya}{\effs}$, which can
be read, ``Under type environment $\tyenv$, term $\tma$ has type $\tya$ and
produces effects $\effs$''.
The rules are implicitly parameterised by a database schema $\schema$ mapping table
names to types of the form $\bagty{\recordty{\ell_i: \basety_i}_i}$.
Many rules are similar to those of the simply-typed $\lambda$-calculus
extended with monadic collection operations~\cite{BNTW95} and a set-based effect
system~\cite{LucassenG88}, and such standard rules are relegated to
Appendix~\ref{app:full-definitions}.

Rule \textsc{T-Query} states that a term $\query{\tma}$ is well-typed if $\tma$
is of a \emph{query type}: either a base type, a record type whose fields are
query types, or a bag whose elements are query types. The term $\tma$ must also
only have $\effread$ effects. These restrictions allow efficient compilation to
SQL~\cite{Cooper09,cheney14sigmod}.

Rule \textsc{T-Get} states that $\get{\tma}$ has type $\bagty{\tya}$
if $\tma$ has type $\tablety{\tya}$, and produces the $\effread$
effect.  Rule \textsc{T-Table} states that a table reference follows
the type of the table in the schema. \textsc{T-Insert} types a
database insert $\dbinsert{\tma}{\tmb}$, requiring $\tma$ to be a
table reference of type $\tablety{\tya}$, and the inserted values
$\tmb$ to be a bag of type $\bagty{\tya}$.  \textsc{T-Update} ensures
the predicate and update terms are typable under an environment
extended with the row type, and ensures that all updated values match
the type expected by the row.  Rule \textsc{T-Delete} is similar. All
subterms used as predicates or used to calculate updated terms must be
pure (that is, side-effect free), and all modifications have the
$\effwrite$ effect.

\subsection{Semantics}

\begin{figure*}
{\footnotesize
    \header{Syntax of values, operations on values, and value typing}
    \begin{minipage}{0.4\textwidth}
    \[
        \begin{array}{lrcl}
            \text{Values} & \vala, \valb & ::= & \fun{x}{\tma} \midspace c \midspace
            \tablevar \midspace \recordterm{\ell_i = \vala_i}_i \midspace \bag{\seq{V}}
        \end{array}
    \]
    \[
        \begin{array}{rcl}
            \bag{\seq{\vala}} \denotbagunion \bag{\seq{\valb}} & \defeq &
            \bag{\seq{\vala} \cdot \seq{\valb}} \spacerow \\
            \recordwithtwo{\recordterm{\ell_i = \vala_i}_{i \in I}}{(\ell_j =
        \valb_j)} & \defeq & \recordterm{\ell_i = \vala_i}_{i \in {I \without
    J}} \recordplus \recordterm{\ell_j = \valb_j}_{j \in J}
        \end{array}
    \]
    \end{minipage}
    \qquad
    \begin{minipage}{0.4\textwidth}
        \begin{mathpar}
            \inferrule
            [T-BagV]
            {
                (\tseq{\tyenv}{\vala_i}{\tya}{\pure})_{i}
            }
            { \tseq{\tyenv}{\bag{\seq{\vala}\!}}{\bagty{\tya}}{\pure} }
        \end{mathpar}
    \end{minipage}

    \headersig{Big-step reduction rules}{$\tmevaltwo{\tma}{\retpair{\vala}{\db'}} $}
    \begin{mathpar}
        \inferrule
        [E-Now]
        { }
        { \tmevaltwo{\now}{\retpair{\timevar}{\db}} }

        \inferrule
        [E-Query]
        { \readtmevaltwo{\tma}{\vala}}
        { \tmevaltwo{\query{\tma}}{\retpair{\vala}{\db}} }

        \inferrule
        [E-Get]
        {
            \tmevaltwo{\tma}{\retpair{\tablevar}{\db'}}
        }
        { \tmevaltwo{\get{\tma}}{\retpair{\db'(\tblvar)}{\db'}} }

        \inferrule
        [E-Insert]
        {
            \tmevaltwo{\tma}{\retpair{\tblvar}{\db'}} \\
            \puretmevaltwo{\tmb}{\vala}
        }
        {
            \tmevaltwo
                {\dbinsert{\tma}{\tmb}}
                {\retpair{()}{\extendenv{\db'}{\tblvar}{\db'(\tblvar) \denotbagunion \vala}} }
        }

        \inferrule
        [E-Update]
        {
            \tmevaltwo{\tmc}{\retpair{\tblvar}{\db_1}} \\
                    \db_2 =
                        {\extendenv
                            {\db_1}
                            {\tblvar}
                            { \bag{\mkwd{upd}(\var{v}) \mid \var{v} \in \db_1(\tblvar)}}
                        }
                    \\\\
                    {
                    \mkwd{upd}(\var{v}) =
                    {
                        \begin{cases}
                            \recordwithtwo
                                {\var{v}}
                                {\seq{\ell = \valb}}
                                & \text{ if }
                                \puretmevaltwo{\tma \{ \var{v} / x \}}{\ttrue}
                                \text{ and }
                                (\puretmevaltwo{\subst{\tmb_i}{\var{v}}{x}}{\valb_i})_i \\
                            \var{v} & \text{ if } \puretmevaltwo{\tma \{ \var{v} / x \}}{\ffalse}
                        \end{cases}
                    }
                    }
        }
        {
            \tmevaltwo
                {\dbupdate{x}{\tmc}{\tma}{(\seq{\ell = \tmb})}}
                {\retpair{()}{\db_2}}
        }

        \inferrule
        [E-Delete]
        {
            \tmevaltwo{\tma}{\retpair{\tblvar}{\db_1}}
            \\
            \db_2 =
                {\extendenv
                    {\db_1}
                    {\tblvar}
                    { \bag{\var{v} \in \db(\tblvar) \mid
                        \puretmevaltwo{\tmb \{\var{v} / x \}}{\ffalse}} }
                }
        }
        {
            \tmevaltwo
                {\dbdelete{x}{\tma}{\tmb}}
                {\retpair
                    {()}
                    {\db_2}
                }
        }
    \end{mathpar}
}
    \caption{Semantics of $\linq$ (selected)}
    \label{fig:linq-semantics}
\end{figure*}

Figure~\ref{fig:linq-semantics} shows the syntax and typing rules of values, and
the big-step semantics of $\linq$.  Most rules are standard, and presented in
Appendix~\ref{app:full-definitions}.
Values $\vala, \valb$ include functions, constants, tables, fully-evaluated
records, and fully-evaluated bags $\bag{\seq{\vala}}$. Unlike the
unary bag constructor $\bag{\tma}$, fully-evaluated bags contain an
unordered
\emph{collection} of values.
All values are pure.  We write $\recordplus$
for record extension, e.g.,
$\recordterm{\ell_1 = \tma} \recordplus \recordterm{\ell_2 = \tmb} =
\recordterm{\ell_1 = \tma, \ell_2 = \tmb}$.

Since evaluation is effectful (as database operations can
update the database), the evaluation judgement has the shape
$\tmevaltwo{\tma}{\retpair{\vala}{\db'}}$; this can be read ``term $\tma$ with
current database $\db$ at time $\timevar$ evaluates to value $\vala$ with
updated database $\db'$''. A database is a mapping from table names to bags of
base records. To avoid additional complexity, we assume
evaluation is atomic and does not update the time; one could straightforwardly
update the semantics with a $\calcwd{tick}$ operation without affecting any
results.

We use two further evaluation relations for terms which do not write to
the database: $\puretmevaltwo{\tma}{\vala}$ states that a pure term $\tma$
(i.e., a term typable under an empty effect set) evaluates to $\vala$.
Similarly, $\readtmevaltwo{\tma}{\vala}$ states that a term $\tma$ which only
performs the $\calcwd{read}$ effect evaluates to $\vala$. We omit the
definitions, which are similar to the evaluation relation but do not propagate
database changes (since no changes can occur).

Rule \textsc{E-Now} returns the current timestamp.  Rule \textsc{E-Query}
evaluates the body of a query using the read-only evaluation relation.
Rule \textsc{E-Get} evaluates its subject to a table reference, and then returns
the contents of a table. Rule \textsc{E-Insert} does similar, evaluating the
values to insert, and then appending them to the contents of the table.
Rule \textsc{E-Update} iterates over a table, updating a record if the
predicate matches, and leaving it unmodified if not.
Finally, \textsc{E-Delete} deletes those rows satisfying the deletion
predicate.

$\linq$ enjoys a standard type soundness property.

\begin{proposition}[Type soundness]\label{prop:linq:soundness}
    If $\tseq{\cdot}{\tma}{\tya}{\effs}$ then there exists some $\vala$ and
    $\db'$
    such that $\tmevaltwo{\tma}{\retpair{\vala}{\db'}}$ and
    $\tseq{\cdot}{\vala}{\tya}{\pure}$.
\end{proposition}

More importantly, the type-and-effect system ensures that query and
update expressions in $\linq$ can be translated to SQL
equivalent,
even in the presence of higher-order functions and
nested query results~\cite{Cooper09,lindley12tldi,cheney13icfp,cheney14sigmod}.  This alternative
implementation is equivalent to the semantics given here but usually
much more efficient since the database query optimiser can takes
advantage of any available integrity constraints or statistics about
the data.

\section{Transaction Time}\label{sec:tlinq}

The first dimension of time we investigate is \emph{transaction
time}~\cite{SnodgrassA85:taxonomy}, which records how the \emph{state of the
database} changes over time.
The key idea behind transaction time databases is that update operations are
\emph{non-destructive}, so we can always view a database as it stood at a
particular point in time.

Let us illustrate with the to-do list example from the
introduction. The original table is on the left.  The table after
making some changes
is shown on the right.

\begin{minipage}{0.5\columnwidth}
{\footnotesize
\begin{tabular}{| c | c |}
    \hline
    task & done \\ \hline
    Go shopping & \ttrue \\
    \hline
    Cook dinner & \ffalse \\
    \hline
    Walk the dog & \ffalse \\
    \hline
    Watch TV & \ffalse \\
    \hline
\end{tabular}
}
\end{minipage}
\hfill
\begin{minipage}{0.45\columnwidth}
{\footnotesize
\begin{tabular}{| c | c |}
    \hline
    task & done \\ \hline
    Go shopping & \ttrue \\
    \hline
    Cook dinner & \ttrue \\
    \hline
    Walk the dog & \ffalse \\
    \hline
\end{tabular}
}
\end{minipage}

However, since updates and deletions in $\linq$ are destructive, we have lost
the original data. Instead, let us see how this could be handled by a
transaction-time database:

{\footnotesize
    \begin{center}
\begin{tabular}{| c | c || c | c |}
    \hline
    task & done & start & end \\ \hline
    Go shopping & \ttrue & 11:00 & $\infty$ \\
    \hline
    Cook dinner & \ffalse & 11:00 & $\infty$ \\
    \hline
    Walk the dog & \ffalse & 11:00 & $\infty$ \\
    \hline
    Watch TV & \ffalse & 11:00 & $\infty$ \\
    \hline
\end{tabular}
\end{center}
}

There are several methods by which we can maintain the temporal information in
the database: for example we could maintain a tracking log which records each
entry, or we could use various \emph{temporal partitioning} strategies~\cite{Snodgrass99:book}.
In this paper, we use a \emph{period-stamped} representation, where each record in the
database is augmented with fields delimiting the interval when the record was present in the database.

The time period is a closed-open representation, meaning that each row
is in the database from (and including) the time stated in the
\emph{start} column, and is in the database up to (but not including)
the time stated in the \emph{end} column. We also assume that
$start < end$ always holds.

Here, our database states that all four tasks were entered into the
database at 11:00. However, if we then decide to check off ``Cook
dinner'' at 17:30 and delete ``Watch TV'' at 19:00, we  obtain the
table shown in the introduction.

Since timestamps are either $\forever$ or only refer to the past; users do
not modify period stamps directly; and the information in the database
grows monotonically, we can reconstruct the state of the database at any given
time.

\subsection{Calculus}
\begin{figure*}[t]
    {\footnotesize
\header{Additional Syntax for $\tlinq$}
\begin{minipage}{\textwidth}
\[
\begin{array}{lrclclrcl}
    \text{Types} & \tya, \tyb & ::= & \cdots \midspace \transtimety{\tya} &\qquad & \text{Timestamped rows} & D & ::= & \dbrow{\vala_1}{\vala_2}{\vala_3} \\
    \text{Terms} & \tmc, \tma, \tmb & ::= & \cdots \midspace \data{\tma} \midspace
    \timestart{\tma} \midspace \timeend{\tma} &&
    \text{Values} & \vala, \valb & ::= & \cdots \midspace D
\end{array}
\]
\end{minipage}

\headerarg
    {Modified Typing Rules for $\tlinq$}
    {\framebox{$\tseq{\tyenv}{\vala}{\tya}{\effs}$}\qquad\framebox{$\tseq{\tyenv}{\tma}{\tya}{\effs}$}}
\begin{mathpar}
    \inferrule
    [T-Row]
    {
        \tseq{\tyenv}{\vala_1}{\tya}{\pure} \\\\
        \tseq{\tyenv}{\vala_2}{\timety}{\pure} \\
        \tseq{\tyenv}{\vala_3}{\timety}{\pure}
    }
    { \tseq
        {\tyenv}
        {\dbrow{\vala_1}{\vala_2}{\vala_3}}
        {\transtimety{\tya}}
        {\pure}
    }

    \inferrule
    [T-Data]
    { \tseq{\tyenv}{\tma}{\transtimety{\tya}}{\effs} }
    { \tseq{\tyenv}{\data{\tma}}{\tya}{\effs} }

    \inferrule
    [T-Start]
    { \tseq{\tyenv}{\tma}{\transtimety{\tya}}{\effs} }
    { \tseq{\tyenv}{\timestart{\tma}}{\timety}{\effs} }

    \inferrule
    [T-End]
    { \tseq{\tyenv}{\tma}{\transtimety{\tya}}{\effs} }
    { \tseq{\tyenv}{\timeend{\tma}}{\timety}{\effs} }

    \inferrule
    [T-Get]
    { \tseq{\tyenv}{\tma}{\tablety{\tya}}{\effs} }
    { \tseq
        {\tyenv}
        {\get{\tma}}
        {\bagty{\transtimety{\tya}}}
        {\effset{\effread} \cup \effs}
    }
\end{mathpar}

\headersig{Semantics for $\tlinq$ database
operations}{$\ttmevaltwo{\tma}{\retpair{\vala}{\db'}}$}
    \begin{mathpar}
        \inferrule
        [ET-Data]
        { \ttmevaltwo
            {\tma}
            {\retpair{\dbrow{\vala_1}{\vala_2}{\vala_3}}{\db'}}
        }
        { \ttmevaltwo
            {\data{\tma}}
            {\retpair{\vala_1}{\db'}}
        }

        \inferrule
        [ET-Start]
        { \ttmevaltwo
            {\tma}
            {\retpair{\dbrow{\vala_1}{\vala_2}{\vala_3}}{\db'}}
        }
        { \ttmevaltwo
            {\timestart{\tma}}
            {\retpair{\vala_2}{\db'}}
        }

        \inferrule
        [ET-End]
        { \ttmevaltwo
            {\tma}
            {\retpair{\dbrow{\vala_1}{\vala_2}{\vala_3}}{\db'}}
        }
        { \ttmevaltwo
            {\timeend{\tma}}
            {\retpair{\vala_3}{\db'}}
        }

        \inferrule
        [ET-Insert]
        {
            \ttmevaltwo{\tma}{\retpair{\tblvar}{\db_1}} \\
            \puretmevaltwo{\tmb}{\bag{\seq{\vala}}} \\\\
            \mathit{vs} = \bag{\dbrow{\var{v}}{\timevar}{\forever} \mid \var{v} \in \seq{\vala}} \\
            \db_2 = \extendenv{\db'}{\tblvar}{\db_1(\tblvar) \denotbagunion \mathit{vs}}
        }
        {
          \ttmevaltwo{\dbinsert{\tma}{\tmb}}{\retpair{()}{\db_2}}
        }

        \inferrule
        [ET-Delete]
        {
            \ttmevaltwo{\tma}{\retpair{\tblvar}{\db_1}} \\
            \db_2 =
                {\extendenv
                    {\db_1}
                    {\tblvar}
                    { \bag{\mkwd{del}(\var{v}) \midspace \var{v} \in \db_1(\tblvar)} }
                } \\\\
                {\bl
                    \mkwd{del}(\dbrow{\var{data}}{\var{start}}{\var{end}})  = \\
                    \quad
                    {
                        \begin{cases}
                            \dbrow{\var{data}}{\var{start}}{\timevar} &
                             \text{if }
                                 \var{end} = \forever \text{ and }
                                 \puretmevaltwo{\subst{\tmb}{\var{data}}{x}}{\ttrue}
                                \\
                            \dbrow{\var{data}}{\var{start}}{\var{end}} & \text{otherwise}
                        \end{cases}
                    }
                 \el
                }
        }
        {
            \ttmevaltwo
                {\dbdelete{x}{\tma}{\tmb}}
                {\retpair
                    {()}
                    {\db_2}}
        }

        \inferrule
        [ET-Update]
        {
           \ttmevaltwo{\tmc}{\retpair{\tblvar}{\db_1}} \\
           \db_2 =
            \extendenv
                {\db_1}
                {\tblvar}
                {\flattenbag{\bag{\mkwd{upd}(\var{v}) \mid \var{v} \in \db_1(\tblvar)}}} \\\\
                {\bl
            \mkwd{upd}(\dbrow{\mathit{data}}{\mathit{start}}{\mathit{end}}) = \\
            \quad
                {
                \begin{cases}
                    {
                       \bag{\dbrow
                           {\var{data}}
                           {\var{start}}
                           {\timevar},
                           \dbrow
                           {
                               \recordwithtwo
                               {\mathit{data}}
                               {\seq{\ell = \vala}}
                           }
                           {\timevar}
                           {\forever}
                       }
                     }
                     \\
                           \quad {
                               \text{if }
                                 \puretmevaltwo{\subst{\tma}{\var{data}}{x}}{\ttrue},
                                 (\puretmevaltwo{\subst{\tmb_i}{\var{data}}{x}}{\vala_i})_i,
                                  \text{and } \mathit{end} = \forever
                           }
                  \\
                  {
                      \bag{\dbrow{\mathit{data}}{\mathit{start}}{\mathit{end}}}
                       \quad { \text{otherwise} }
                  }
               \end{cases}
             }
            \el     
            }
        }
        { \ttmevaltwo
            {\dbupdate{x}{\tmc}{\tma}{(\seq{\ell = \tmb})}}
            {\retpair{()}{\db_2}}
        }
      \end{mathpar}
}
\caption{Syntax, typing rules, and semantics of $\tlinq$}
\label{fig:tlinq}
\end{figure*}

$\tlinq$ extends $\linq$ with native support for transaction time
operations; instead of performing destructive updates, we adjust the end
timestamp of affected rows and, if necessary, insert updated rows.
$\tlinq$ database entries are therefore of the form
$\dbrow{\vala_1}{\vala_2}{\vala_3}$, where $\vala_1$ is the record data and
$\vala_2$ and $\vala_3$ are the start and end timestamps.

Figure~\ref{fig:tlinq} shows the syntax, typing rules, and semantics of
$\tlinq$; for brevity, we show the main differences to $\linq$.
Period-stamped database rows
are represented as triples $\dbrow{\var{data}}{\var{start}}{\var{end}}$ with type
$\transtimety{\tya}$, where $\var{data}$ has type $\tya$ (the type of each
record), and both $\var{start}$ and $\var{end}$ have type $\timety$.
A row is currently present in the
database if its end value is $\forever$. We introduce three accessors:
$\data{\tma}$ extracts the data record from a transaction-time row;
$\timestart{\tma}$ extracts the start time; and $\timeend{\tma}$ extracts the
end time. The $\calcwd{get}$ construct has an updated type to show that it
returns a bag of $\transtimety{\tya}$ values, rather than the records directly.
The typing rules for the other constructs remain the same as in $\linq$.

The accessor rules \textsc{ET-Data}, \textsc{ET-Start}, and \textsc{ET-End}
project the expected component of the transaction-time row.
Rule \textsc{ET-Insert}
period-stamps each record to begin at the current time, and sets the end time to
be $\forever$.
Rule \textsc{ET-Delete} records deletions for current rows satisfying
the deletion predicate.  Instead of being removed from the database,
the end times of the affected rows are set to the current timestamp.
Finally, rule \textsc{ET-Update} performs updates for current rows satisfying the update predicate. Instead of changing a record directly, the $\mkwd{upd}$
definition generates \emph{two} records: the previous
record, closed off at the current timestamp, and the new record with updated
values, starting from the current timestamp and with end field $\forever$.
Returning to our running example, define
$\mkwd{at}(\var{tbl}, \var{time})$ to return all records in
$\var{tbl}$ starting before $\var{time}$ and ending after
$\var{time}$.  We can then query the database as it stood at 18:00:

\begin{minipage}{0.6\columnwidth}
    {\footnotesize
\[
    \bl
    \at{\var{t}}{\var{time}} \defeq \\
    \quad \forcomptwo{x}{\get{\var{t}}} \\
    \qquad \whereone{\timestart{x} \le \var{time} \wedge \var{time} < \timeend{x}} \\
    \qquad \bag{\data{x}} \smallskip\\
    \at{\var{tbl}}{18{:}00}
    \el
\]%
}
\end{minipage}
\hfill
\begin{minipage}{0.3\columnwidth}
{\footnotesize
    \begin{tabular}{| c | c |}
        \hline
        task & done \\ \hline
        Go shopping & \ttrue \\
        \hline
        Cook dinner & \ttrue \\
        \hline
        Walk the dog & \ffalse \\
        \hline
        Watch TV & \ffalse \\
        \hline
    \end{tabular}\\
}
\end{minipage}

\noindent
Let $\mkwd{current}(\var{t}) = \mkwd{at}(\var{t},\forever)$
return the current snapshot of $t$.
We can then query the current snapshot of the database:

{\footnotesize
\begin{minipage}{0.3\columnwidth}
\[
  \bl
  \current{\var{tbl}} =
    \el
\]
\end{minipage}%
\begin{minipage}{0.3\columnwidth}
\begin{tabular}{| c | c |}
    \hline
    task & done \\ \hline
    Go shopping & \ttrue \\
    \hline
    Cook dinner & \ttrue \\
    \hline
    Walk the dog & \ffalse \\
    \hline
\end{tabular}
\end{minipage}
}
\vspace{-1em}

\subsection{Translation}
We can implement the native transaction-time semantics for $\tlinq$
database operations by translation to $\linq$. Our translation adapts
the SQL implementations of temporal operations
by~\citet{Snodgrass99:book} to a language-integrated query
setting. We prove correctness relative to the semantics.

$\tlinq$ has a native representation of period-stamped data, whereas
$\linq$ requires table types to be flat. Consequently, the
translations require knowledge of the types of each record. We
therefore annotate each $\tlinq$ database term with the type of table
on which it operates (this can be achieved through a standard
type-directed translation pass).

\begin{figure}
{\footnotesize
\header{Auxiliary Definitions}
    \[
    \begin{array}{rcl}
        \etaexp{x}{\seq{\ell}} & \defeq & \recordterm{\ell_i = \project{x}{\ell_i}}_i \\
        \restrict{x}{\seq{\ell}}{\tma} & \defeq & (\fun{x}{\tma}) \app \etaexp{\seq{\ell}}{x}\\
        \iscurrent{\tma} & \defeq & \project{\tma}{\var{end}} = \forever
    \end{array}
\]

\header{Translation on database terms}
\[
\bl
      \ttmtrans{\data{\tma}}  =  \project{\tvaltrans{\tma}}{\field{data}} \spacerow \\
      \ttmtrans{\timestart{\tma}}  =  \project{\tvaltrans{\tma}}{\field{start}} \spacerow \\
      \ttmtrans{\timeend{\tma}}  =  \project{\tvaltrans{\tma}}{\field{end}} \spacerow \\
  \ttmtrans{\getann{(\ell_i: A_i)_i}{\tma}} = \\
  \quad
  \queryzero \\
  \quad
    {
      \blt
        \forcomp{x}{\get{\tvaltrans{\tma}}} \\
        \quad \bag{\recordterm{
                \fielddata{\etaexp{x}{\seq{\ell}}},
                \fieldstart{\project{x}{\field{start}}},
                \fieldend{\project{x}{\field{end}}}}}
      \el
    } \spacerow \\
  \ttmtrans{\dbinsertann{\tma}{\tmb}} = \\
  \quad
    {
    \blt
      \effletone{\var{rows}} \\
      \quad \forcomptwo{x}{\ttmtrans{\tmb}} \\
      \qquad
        \bag{
            \recordplustwo
                {\etaexp{x}{\seq{\ell}}}
                {\recordterm{
                    \fieldstart{\now},
                    \fieldend{\forever}
                }}
        } \\
        \calcwd{in} \\
      \dbinsert{\ttmtrans{\tma}}{\var{rows}}
    \el
    } \spacerow \\
  \ttmtrans{\dbdeleteann{x}{\tma}{\tmb}} = \\
    \quad {
      \blt
      \dbupdatetwo{x}{\tvaltrans{\tma}} \\
      \quad \calcwhereone{\restrict{x}{\seq{\ell}}{\ttrans{\tmb}} \wedge \iscurrent{x}} \\
      \quad \calcwd{set} \: (\fieldend{\now})
      \el
    } \spacerow \\
  \ttmtrans{
    \dbupdateann
        [\recordterm{\ell_i : \tya_i}_{i \in I}]
        {x}
        {\tmc}
        {\tma}
        {(\ell = \tmb_j)_{j \in J}}} = \\
  \quad
  {
    \blt
      \letintwo{\var{tbl}}{\ttrans{\tmc}} \\
      \letinone{\var{affected}}\\
      \quad \queryzero \\
      \qquad \forcomptwo{x}{\get{\var{tbl}}} \\
      \qqquad
        \whereone
            {(\restrict{x}{\set{\ell_i}_{i \in I}}{\ttmtrans{\tma}} \wedge \iscurrent{x})} \\
      \qqquad \baglr{
              {\bl
                        \recordterm{\ell_i = \project{x}{\ell_i}}_{i \in I
                        \without J} \, \recordplus \\
                        \recordterm
                            {\ell_j =
                                \restrict
                                    {x}
                                    {\set{\ell_i}_{i \in I}}
                                    {\ttrans{\tmb_j}}}_{j \in J} \, \recordplus \\
                  \recordterm{
                      \fieldstart{\now},
                      \fieldend{\forever}
                  }
                  \el
              }
        }
      \\
      \calcwd{in}
      \\
      \dbupdatetwo{x}{\var{tbl}} \\
      \quad \calcwd{where} \: (\restrict{x}{\seq{\ell}}{\ttmtrans{\tma}} \wedge \iscurrent{x}) \\
      \quad \calcwd{set} \: \recordterm{\fieldend{\now}}; \\
      \dbinsert{\var{tbl}}{\var{affected}}
    \el
  }
  \el
  \]
}
\caption{Translation from $\tlinq$ into $\linq$}
\label{fig:tlinq-translation}
\end{figure}

The (omitted) translation of $\tlinq$ types into $\linq$ types is straightforward,
save for $\transtimety{\tya}$ which is translated into a record
$\recordty{\tyfielddata{\ttrans{\tya}}, \tyfieldstart{\timety},
\tyfieldend{\timety}}$.
The same is true for the basic $\lambda$-calculus terms.
Timestamped rows
$\dbrow{\vala_{\var{data}}}{\vala_{\var{start}}}{\vala_{\var{end}}}$ are
translated to fit the above record type; specifically,
$\recordterm{\fielddata{\ttrans{\vala_{\var{data}}}},
    \fieldstart{\ttrans{\vala_{\var{start}}}},
\fieldend{\ttrans{\vala_{\var{end}}}} }$.

\begin{remark}
We have chosen to represent a $\tlinq$ period-stamped record as a nested record
in $\linq$, but we could equally adopt a flat representation.
Since we build on the Nested Relational Calculus, we take advantage of the
ability to return nested results; previous work on query
shredding~\cite{cheney14sigmod} allows us to flatten nested results in a later
translation pass.
A nested representation is more convenient for our implementation and makes the
translation simpler and more compositional, so we mirror this choice in the
formalism.
\end{remark}

We define the \emph{flattening} of a
$\tlinq$ row and database as:%
\[
    \begin{array}{rcl}
    \flatten{(\dbrow{\recordterm{\ell_i = \vala_i}_i}{\valb_1}{\valb_2})}
    & \defeq &
\recordterm{\ell_i = \vala_i}_i \recordplus
    \recordterm{\fieldstart{\valb_1}, \fieldend{\valb_2}} \\
    \flatten{\db} & \defeq & [ t \mapsto
\bag{\seq{\flatten{D}}} \mid t \mapsto \bag{\seq{D}} \in \db]
    \end{array}
\]%

Figure~\ref{fig:tlinq-translation} shows the translation of $\tlinq$ terms into
$\linq$. The translation makes use of three auxiliary definitions. Eta-expansion
$\etaexp{x}{\seq{\ell}}$ eta-expands a variable of record type with respect to a
sequence of labels, and $\restrict{x}{\seq{\ell}}{\tma}$ applies an eta-expanded
record to a $\tma$ under a $\lambda$-binder for $x$ (required since
the $\tlinq$ predicates expect just the data from the record, and not the
period-stamping fields).
Finally, $\mkwd{isCurrent}(\tma)$ tests whether the
$\var{end}$ field of $\tma$ is $\forever$.

Since timestamped rows are translated to records, the temporal accessor functions
are translated to record projection. The $\calcwd{get}$ function is translated
to retrieve the flattened $\linq$ representation of the table and pack it via
$\eta$-expansion into a bag of nested records. The translation of
$\calcwd{insert}$ extends the provided values with a $\var{start}$
field referring to the current timestamp and an $\var{end}$ field set to
$\forever$, before inserting them into the database.

A $\calcwd{delete}$ is translated as a $\linq$ $\calcwd{update}$ operation,
which sets the $\var{end}$ record of each affected row to the current timestamp.
An $\calcwd{update}$ is translated in three steps: querying the database to
obtain the affected current records, with updated values and timestamps;
updating the database to close off the existing affected rows;
and materialising the insertion.

\subsection{Metatheory}

We restrict our attention to \emph{well formed} rows and databases, where the
$\var{start}$ timestamp is less than the $\var{end}$ timestamp.

\begin{definition}[Well formed rows and databases]
    A database $\db$ is well formed, written $\wf{\db}$, if for each
    timestamped row
    $\dbrow{\vala_{\var{data}}}{\vala_{\var{start}}}{\vala_{\var{end}}}$ in $\db$
    we have that $\vala_{\var{start}} < \vala_{\var{end}}$.
\end{definition}
\begin{definition}[Maximum timestamp]
    The \emph{maximum timestamp} of a collection of records
    $\seq{D}$
    is defined as the maximum timestamp in the set
    $\{ \vala_{\var{end}} \mid
    \dbrow{\vala_{\var{data}}}{\vala_{\var{start}}}{\vala_{\var{end}}} \in
    \seq{D},
\vala_{\var{end}} \ne \forever \}$,
    or $\beginningoftime$ if the set is empty.

    The maximum timestamp of a database
    $\db$, written $\maxts{\db}$, is the maximum timestamp of all its constituent tables.
\end{definition}

Again, $\tlinq$ enjoys type soundness.

\begin{proposition}[Type soundness ($\tlinq$)]
    If $\tseq{\cdot}{\tma}{\tya}{\effs}$, then given a
    $\wf{\db}$ and $\timevar$ such that
    $\maxts{\db} \le \timevar$,
    then there exists some $\vala$ and
    well formed $\db'$ such that $\ttmevaltwo{\tma}{\retpair{\vala}{\db'}}$.
\end{proposition}

We can now show that the translation is correct:

\begin{restatable}{thm}{ttranscorrect}
    If $\tseq{\cdot}{\tma}{\tya}{\effs}$ and
    $\ttmevaltwo{\tma}{\retpair{\vala}{\db'}}$ where
    $\wf{\db}$ and $\maxts{\db} \le \timevar$,
    then
        $\tmevaltwo
            [\flatten{\db}, \timevar]
            {\ttrans{\tma}}
            {\retpair
                {\ttrans{\vala}}
                {\flatten{\db'}}}$.
\end{restatable}

\newcommand{\vlinqsyntax}{
\begin{figure*}[t]
{\footnotesize
~\header{Syntax}
\begin{minipage}{\textwidth}
\[
    \begin{array}{lrcl}
      \text{Types} & \tya, \tyb & ::= & \validtimety{\tya} \\
      \text{Terms} & \tmc, \tma, \tmb & ::= & \cdots \midspace
      \dbrow{\tmc}{\tma}{\tmb}
                       \midspace \data{\tma} \midspace \timestart{\tma} \midspace \timeend{\tma}
                   \mid \dbinsertseq{\tma}{\tmb} \\
                   & & \midspace &
                   \dbupdatebetween{x}{\tmc}{\tma_1}{\tma_2}{\tma_3}{(\seq{\ell = \tmb})} \\
                   & & \midspace & \dbupdatenonseq{x}{\tmc}{\tma}{(\seq{\ell = \tmb})}{\tmb'_1}{\tmb'_2} \\
                   & & \midspace &
                   \dbdeletebetween{x}{\tmc}{\tma_1}{\tma_2}{\tmb} \\
                   & & \midspace & \dbdeletenonseq{x}{\tma}{\tmb} \\
      \text{Values} & \vala, \valb & ::= & \cdots \midspace \dbrow{\vala_1}{\vala_2}{\vala_3}
    \end{array}
  \]
\end{minipage}

~\headersig{Typing rules}{$\Gamma \vdash M : A \effann{\effs}$}
\vspace{-1em}
\begin{mathpar}
  \inferrule
  [TV-Get]
  { \tseq{\tyenv}{\tma}{\tablety{\tya}}{\effs} }
  { \tseq
      {\tyenv}
      {\get{\tma}}
      {\bagty{\validtimety{\tya}}}
      {\effset{\effread} \cup \effs}
  }

  \inferrule
  [TV-SeqInsert]
  {
      \tseq{\tyenv}{\tma}{\tablety{\tya}}{\effs} \\
      \tseq{\tyenv}{\tmb}{\bagty{\validtimety{\tya}}}{\pure}
  }
  { \tseq{\tyenv}{\dbinsertseq{\tma}{\tmb}}{\recordty{}}{\effset{\effwrite} \cup \effs} }

  \inferrule
  [TV-SeqUpdate]
  { \tseq{\tyenv}{\tmc}{\tablety{\tya}}{\effs} \\
      \tya = \recordty{\ell_i : \tyb_i}_{i \in I} \\
    \tseq{\tyenv}{\tma_1}{\timety}{\pure} \\
    \tseq{\tyenv}{\tma_2}{\timety}{\pure} \\
    \tseq{\tyenv, x : \tya}{\tma_3}{\boolty}{\pure} \\
    (j \in I \wedge \tseq{\tyenv, x : \tya}{\tmb_j}{\tyb_j}{\pure})_{j \in J}
  }
  { \tseq
      {\tyenv}
      {\dbupdatebetween{x}{\tmc}{\tma_1}{\tma_2}{\tma_3}{(\ell_j = \tmb_j)_{j \in J}}}
      {()}
      {\effset{\effwrite} \cup \effs }
  }

  \inferrule
  [TV-NonseqUpdate]
  {
      \tseq{\tyenv}{\tmc}{\tablety{\tya}}{\effs} \\
      \tya = \recordty{\ell_i : \tyb_i}_{i \in I} \\
      \tseq{\tyenv, x : \validtimety{\tya}}{\tma}{\boolty}{\pure} \\
      (j \in I \wedge \tseq{\tyenv, x : \validtimety{\tya}}{\tmb_j}{\tyb_j}{\pure})_{j \in J} \\
      \tseq{\tyenv, x : \validtimety{\tya}}{\tmb'_1}{\timety}{\pure} \\
      \tseq{\tyenv, x : \validtimety{\tya}}{\tmb'_2}{\timety}{\pure}
  }
  {
      \tseq
        {\tyenv}
        { \dbupdatenonseq{x}{\tmc}{\tma}{(\ell_j {=} \tmb_j)_{j \in J}}{\tmb'_1}{\tmb'_2} }
        { () }
        { \effset{\effwrite} \cup \effs }
  }

     \inferrule
     [TV-SeqDelete]
     {
       \tseq{\tyenv}{\tmc}{\tablety{\tya}}{\effs_1}\!\! \\
       \tseq{\tyenv}{\tma_1}{\timety}{\effs_2}\!\! \\
       \tseq{\tyenv}{\tma_2}{\timety}{\effs_3}\!\! \\
       \tseq{\tyenv, x {:} \tya}{\tmb}{\boolty}{\pure}
     }
     {
         \tseq
           {\tyenv}
           {\dbdeletebetween{x}{\tmc}{\tma_1}{\tma_2}{\tmb}}
           {()}
           {\effset{\effwrite} \cup \effs_1 \cup \effs_2 \cup \effs_3}
     }

\inferrule
  [TV-NonseqDelete]
  {
      \tseq{\tyenv}{\tma}{\tablety{\tya}}{}{\effs} \\
      \tseq{\tyenv, x : \validtimety{\tya}}{\tmb}{\boolty}{\pure}
  }
  { \tseq
        {\tyenv}
        {\dbdeletenonseq{x}{\tma}{\tmb}}
        {()}
        {\effset{\effwrite} \cup \effs}
  }
\end{mathpar}
}
\caption{Syntax and typing rules for \vlinq}
\label{fig:vlinq-syntax}
\end{figure*}
}

\newcommand{\vlinqreduction}{
\begin{figure*}[t]
    {\footnotesize
~\headersig{Reduction rules}{$\vtmevaltwo{\tma}{\retpair{\vala}{\db'}}$}
\begin{mathpar}
  \inferrule
  [EV-Row]
  {
      \vtmevaltwo{\tma_1}{\retpair{\vala_1}{\db_1}} \\\\
      \vtmevaltwo[\timevar, \db_1]{\tma_2}{\retpair{\vala_2}{\db_2}} \\
      \vtmevaltwo[\timevar, \db_2]{\tma_3}{\retpair{\vala_3}{\db_3}}
  }
  { \vtmevaltwo
      {\dbrow{\tma_1}{\tma_2}{\tma_3}}
      {\retpair{\dbrow{\vala_1}{\vala_2}{\vala_3}}{\db_3}}
  }

  \inferrule
  [EV-SeqInsert]
  {
      \vtmevaltwo{\tma}{\retpair{\tblvar}{\db_1}} \\
      \puretmevaltwo{\tmb}{\bag{\seq{\vala}}} \\
      \forall \dbrow{\var{data}}{\var{start}}{\var{end}} \in \seq{\vala}. \var{start}
      < \var{end} \\\\
      \db_2 = \extendenv{\db_1}{\tblvar}{\db_1(\tblvar) \denotbagunion
      \bag{\seq{\vala}}}
  }
  {
    \vtmevaltwo{\dbinsertseq{\tma}{\tmb}}{\retpair{()}{\db_2}}
  }

  \inferrule
  [EV-SeqUpdate]
  {
    \vtmevaltwo{\tmc}{\retpair{\tablevar}{\db_1}} \\
      \puretmevaltwo{\tma_1}{\vala_{\var{start}}} \\
      \puretmevaltwo{\tma_2}{\vala_{\var{end}}} \\
      \vala_{\var{start}} < \vala_{\var{end}}
      \\
    \db_2 = \db_1[ \tablevar \mapsto
      \flattenbag{\bag{\mkwd{upd}(d) \midspace d \in \db_1(\tablevar)}} ]
    \\
      \mkwd{upd}(\dbrow{\var{v}}{\var{start}}{\var{end}}) =
      {
        \begin{cases}
          \bag{\dbrow{W}{\var{start}}{\var{end}}} &
          \text{ if } \puretmevaltwo{\subst{\tma_3}{\var{v}}{x}}{\ttrue} \text{
          and } \vala_{\var{start}} \le \var{start} \text{ and }
          \vala_{\var{end}} \ge \var{end} \quad \textbf{(Case 1)}\\
          \bag{\dbrow{W}{\var{start}}{\vala_{\var{end}}}, \dbrow{\var{v}}{\vala_{\var{end}}}{\var{end}}} &
          \text{ if } \puretmevaltwo{\subst{\tma_3}{\var{v}}{x}}{\ttrue} \text{
          and } \vala_{\var{start}} \le \var{start} \text{ and }
          \vala_{\var{end}} < \var{end} \quad \textbf{(Case 2)}\\
          \bag{
              \dbrow{\var{v}}{\var{start}}{\vala_{\var{start}}},
              \dbrow{W}{\vala_{\var{start}}}{\vala_{\var{end}}},
              \dbrow{\var{v}}{\vala_{\var{end}}}{\var{end}}} &
          \text{ if } \puretmevaltwo{\subst{\tma_3}{\var{v}}{x}}{\ttrue} \text{
          and } \vala_{\var{start}} > \var{start} \text{ and }
          \vala_{\var{end}} < \var{end} \quad \textbf{(Case 3)}\\
          \bag{\dbrow{\var{v}}{\var{start}}{\vala_{\var{start}}}, \dbrow{W}{\vala_{\var{start}}}{\var{end}}} &
          \text{ if } \puretmevaltwo{\subst{\tma_3}{\var{v}}{x}}{\ttrue} \text{
          and } \vala_{\var{start}} > \var{start} \text{ and }
          \vala_{\var{end}} \ge \var{end} \quad \textbf{(Case 4)}\\
          \bag{\dbrow{\var{v}}{\var{start}}{\var{end}}} & \text{ otherwise }
          \hfill \textbf{(Case 5)}\\
        \end{cases}
      }
      \\
      \text{ where for all cases, } W =
                              \recordwithtwo
                                {\var{v}}
                                {\seq{\ell = \valb'}} \text{ given } (\puretmevaltwo{\tmb_i}{\valb'_i})_i
  }
  {
      \tmevaltwo
          {\dbupdatebetween{x}{\tmc}{\tma_1}{\tma_2}{\tma_3}{(\ell_i = \tmb_i)_i}}
          { \retpair{\recordterm{}}{\db_2}}
  }

    \inferrule
    [EV-NonseqUpdate]
    {
      \tmevaltwo{\tmc}{\retpair{\tblvar}{\db_1}} \\
              \db_2 =
                  {\extendenv
                      {\db_1}
                      {\tblvar}
                      { \bag{\mkwd{upd}(d) \mid d \in \db_1(\tblvar)}}
                  }
              \\\\
              {
              \mkwd{upd}(D = \dbrow{\var{v}}{\var{start}}{\var{end}}) =
              {
                  \begin{cases}
                      \dbrow
                          {
                              \recordwithtwo
                                  {\var{v}}
                                  {\seq{\ell = \valb}}
                          }
                          { \valb_{\var{start}} }
                          { \valb_{\var{end}} }
                      & \text{ if }
                      {\blt
                              \puretmevaltwo{\subst{\tma}{D}{x}}{\ttrue}
                              \text{ and }
                              (\puretmevaltwo{\subst{\tmb_i}{D}{x}}{\valb_i})_i
                              \text{ and } \\
                              \puretmevaltwo{\subst{\tmb'_1}{D}{x}}{\valb_{\var{start}}}
                              \text{ and }
                              \puretmevaltwo{\subst{\tmb'_2}{D}{x}}{\valb_{\var{end}}}
                              \text{ and }  \valb_{\var{start}} < \valb_{\var{end}}
                       \el}
                              \\
                      D & \text{ if } \puretmevaltwo{\tma \{ D / x \}}{\ffalse}
                  \end{cases}
              }
              }
  }
  {
      \tmevaltwo
          { \dbupdatenonseq{x}{\tmc}{\tma}{(\ell_i = \tmb_i)_{i}}{\tmb'_1}{\tmb'_2} }
          {\retpair{()}{\db_2}}
  } 
  \end{mathpar}
}
\caption{Reduction rules for \vlinq (selected)}
\label{fig:vlinq-reduction}
\end{figure*}
}

\section{Valid Time}\label{sec:vlinq}
\vlinqsyntax

The other dimension of time we will look at is \emph{valid time}, which tracks
when something is true in the \emph{domain being modelled}. Each timestamp
therefore defines the \emph{period of validity} (PV) of each record.

Unlike in a transaction time database, the database does not
necessarily grow monotonically since we can apply destructive updates
and deletions. Furthermore, whereas in a transaction time database
timestamps can only refer to the past (or $\forever$), in a valid time
database we may state that a row is valid until some specific point in
the future (for example, the end of a fixed-term employment
contract). A further difference from transaction time databases is
that users \emph{can} modify timestamps directly, and can also apply
updates and deletions over a time period.
Let us illustrate with the `employees' table of an HR database:\\

\vspace{-2mm}
\begin{center}
\footnotesize{
\begin{tabular}{|c|c|c||c|c|}
    \hline
    name & position & salary & start & end \\
    \hline
    Alice & Lecturer & 40000 & 2010 & 2018 \\
    \hline
    Alice & Senior Lecturer & 50000 & 2018 & $\infty$ \\
    \hline
    Bob & PhD Student & 15000 & 2019 & 2023 \\
    \hline
    Charles & PhD Student & 15000 & 2018 & 2022 \\
    \hline
\end{tabular}\\
}
\end{center}
The first modification is to hire Dolores as a professor, on an open-ended
contract. As this is an insertion operation on the database at the current
moment in time, it is known as a \emph{current insertion}.
We can write the following query:

\vspace{-2mm}
{\footnotesize
\[
    \bl
    \calcwd{insert} \:
        {\var{employees}} \: \calcwd{values} \\
        \quad
        {\recordterm{\var{name} = \text{``Dolores''},
            \var{position} = \text{``Professor''}, \var{salary} = 70000}}
    \el
\]
}
Next, we want to record that Alice has resigned. We can write the following
\emph{current deletion} query:

\vspace{-2mm}
{\footnotesize
\[
    \dbdelete{x}{\var{employees}}{\project{x}{\var{name}} = \text{``Alice''}}
\]
}
The resulting table state shows that Dolores is a Professor from the current
time onwards, and that the `end' field of Alice's current row
is updated to the current year:

{\footnotesize
    \begin{center}
\begin{tabular}{|c|c|c||c|c|}
    \hline
    name & position & salary & start & end \\
    \hline
    Alice & Lecturer & 40000 & 2010 & 2018 \\
    \hline
    Alice & Senior Lecturer & 50000 & 2018 & 2022 \\
    \hline
     Dolores & Professor & 70000 & 2022 & $\infty$ \\
     \hline
  $\cdots$ & $\cdots$ & $\cdots$ & $\cdots$ & $\cdots$ \\
    \hline
\end{tabular}
\end{center}
}
A powerful feature of valid-time databases is the ability to perform
\emph{sequenced modifications}, which apply an
update or deletion over a particular \emph{period of applicability} (PA).
In fact, current modifications are a special case of sequenced modifications
applied from $\now$ until $\forever$.
Suppose that Dolores has agreed to act as Head of School between 2023 and
2028. We can record this using a \emph{sequenced update} query:

\vspace{-2mm}
{\footnotesize
\[
    \bl
    \calcwd{update} \: \calcwd{sequenced} \: (x \dbcomparrow \var{employees}) \\
    \quad \calcwd{between} \: 2023 \: \calcwd{and} \: 2025 \:
          \calcwd{where} \: (\project{x}{\var{name}} = \text{``Dolores''}) \\
    \quad \calcwd{set} \: (\var{position} = \text{``Head of School''})
    \el
\]
}
with the resulting table being:

{\footnotesize
    \begin{center}
\begin{tabular}{|c|c|c||c|c|}
    \hline
    name & position & salary & start & end \\
    \hline
    Dolores & Professor & 70000 & 2022 & 2023 \\
    \hline
    Dolores & Head of School & 70000 & 2023 & 2028 \\
    \hline
  Dolores & Professor & 70000 & 2028 & $\infty$ \\
  \hline
  $\cdots$ & $\cdots$ & $\cdots$ & $\cdots$ & $\cdots$ \\
    \hline
\end{tabular}\\
\end{center}
}
Since the period of applicability of the sequenced update
is entirely contained within the period of validity of Dolores's row, we end up
with three rows: the unchanged record before and after the PA, and the updated
record during the PA.
We also allow a \emph{sequenced deletion}, and a \emph{sequenced insertion},
where each record's period of validity is given explicitly.

Additionally, suppose that all PhD students are to be given a 1-year extension
due to the disruption caused by the pandemic; in this
case we want to change the period of validity directly. This is known as a
\emph{nonsequenced update}.
We cannot express this
modification using either current or sequenced modifications since we must
calculate the each row's new end date from its previous end date.
We can write the modification as follows, noting
that we can both read from, and write to, the period of validity directly:

\vspace{-2mm}
{\footnotesize
\[
\bl
\dbupdatenonseqtwo{x}{\var{employees}} \\
\quad \calcwd{where} \: (\project{(\data{x})}{\var{position}} = \text{``PhD student''}) \\
\quad \calcwd{set} \: () \: \calcwd{valid} \: \calcwd{from} \: (\timestart{x}) \:
\calcwd{to} \: (\timeend{x} + 1)
\el
\]
}
The resulting table shows that the `end' field of Bob's and
Charles' records are updated to 2024 and 2023 respectively:\\

\vspace{-2mm}
{\footnotesize
\begin{center}
\begin{tabular}{|c|c|c||c|c|}
    \hline
    name & position & salary & start & end \\
    \hline
    Bob & PhD Student & 15000 & 2019 & 2024 \\
    \hline
    Charles & PhD Student & 15000 & 2018 & 2023 \\
    \hline
  $\cdots$ & $\cdots$ & $\cdots$ & $\cdots$ & $\cdots$ \\
    \hline
\end{tabular}\\
\end{center}
}
\vlinqreduction
\subsection{Calculus}
The $\vlinq$ calculus gives a direct semantics to valid time operations. Like
$\tlinq$, $\vlinq$ has a native notion of a period-stamped database row, with
accessors for the data and each timestamp; the typing rules, reduction rules,
and translations are straightforward adaptations of those in $\tlinq$.

Figure~\ref{fig:vlinq-syntax} shows how the syntax and typing rules for $\vlinq$
differ from those of $\linq$. Unlike in $\tlinq$, we can use the term
$\dbrow{\tma_1}{\tma_2}{\tma_3}$ to construct a valid-time row.
Ordinary insert, delete and update operations may be applied to valid time tables (the
straightforward rules are omitted).
Sequenced insertions are described by the term
$\dbinsertseq{\tma}{\tmb}$
where \textsc{TV-SeqInsert} ensures that $\tmb$ is a bag of timestamped records.
Sequenced updates are described by:

\vspace{-2mm}
{\footnotesize
\[
\dbupdatebetween{x}{\tmc}{\tma_1}{\tma_2}{\tma_3}{(\seq{\ell = \tmb})}
\]
}%
Terms $\tma_1$ and $\tma_2$ must be of type $\timety$,
referring to the period of applicability of the sequenced update.
Nonsequenced updates are described by the term:

\vspace{-2mm}
{\footnotesize
\[
    \bl
\calcwd{update} \: \calcwd{nonsequenced} \: (x \dbcomparrow \tmc) \\
\quad \calcwd{where} \: \tma \:
      \calcwd{set} \: (\seq{\ell = \tmb}) \:
    \calcwd{valid} \: \calcwd{from} \: \tmb'_1 \: \calcwd{to} \: \tmb'_2
\el
\]
}
with \textsc{TV-NonseqUpdate} stating that the
database row (including period information) is bound as $x$ in
the predicate $\tma$, update terms $\tmb_j$, and new time periods
$\tmb'_1$ and $\tmb'_2$.

Finally, the term:

\vspace{-2mm}
{\footnotesize
\[
    \dbdeletebetween{x}{\tmc}{\tma_1}{\tma_2}{\tmb}
\]
}
describes a sequenced
deletion which removes the portion of each record satisfying $\tmb$ between times $\tma_1$ and
$\tma_2$.

Since current insertions, updates, and deletions are special cases of
sequenced operations, we need not consider them explicitly; for completeness,
direct semantics can be found in Appendix~\ref{sec:vlinq-current-updates}.
Instead, we show macro translations to the sequenced constructs.
Current insertions can be implemented by desugaring to sequenced insertions, annotating each row with $[\now,\forever)$:

\vspace{-2mm}
{\footnotesize
\[
  \dbinsert{\tma}{\tmb} \leadsto
  \bl
     \letintwo{\var{rows}}{\forcomptwo{x}{\tmb}{\bag{\dbrow{x}{\now}{\forever}}}} \\
     \dbinsertseq{\tma}{\var{rows}}
  \el
\]
}%
Current updates and deletions can be implemented as sequenced updates and
deletions where the period of applicability spans from $\now$ until $\forever$:

\vspace{-2mm}
{\footnotesize
\[
  \bl
  \dbupdate{x}{\tmc}{\tma}{\recordterm{\ell_i = \tmb_i}_i} \leadsto\\
    \quad %
    {
        \bl
        \calcwd{update} \: \calcwd{sequenced} \: (x \dbcomparrow {\tmc}) \\
        \quad \calcwd{between} \: \now \: \calcwd{and} \: \forever \:
        \calcwd{where} \: \tma \:
        \calcwd{set} \: \recordterm{\ell_i = \tmb_i}_i
        \el
      }
      \spacerow \\
    \dbdelete{x}{\tma}{\tmb}\leadsto \\
    \quad \dbdeletebetween{x}{{\tma}}{\now}{\forever}{\tmb}
    \el
\]
}
\sloppypar Fig.~\ref{fig:vlinq-reduction} shows selected reduction
rules for sequenced operations.  We show the rules
for sequenced inserts and updates; the rules for other cases employ
similar ideas and are included in Appendix~\ref{app:full-definitions}.
Nonsequenced updates and deletes are similar to their analogues in
$\linq$ but allow access to, and modification of, row timestamps.
For sequenced insertions, \textsc{EV-SeqInsert} checks that the
period of validity for each row is correct (i.e., that the $\var{start}$ field
is less than the $\var{end}$ field) and appends the provided bag to the table.
Sequenced updates and deletions must account for the various ways that the
period of applicability can overlap the period of validity. There are five main
cases, corresponding to the five ways two closed-open intervals can overlap (or
fail to do so).  Appendix~\ref{app:timelines} summarizes the five cases and
describes how deletes and updates are handled in each case.

\subsection{Translation}
Figure~\ref{fig:vlinq-translation} illustrates the translation from $\vlinq$
into $\linq$. We discuss the translations for sequenced inserts and both
sequenced and nonsequenced updates; the other modifications are similar and
included in Appendix~\ref{app:full-definitions}. As before, we require
annotations on each of the database update terms.

Nonsequenced updates and deletions can be updated directly by their
corresponding $\linq$ operation; we use an auxiliary definition, $\mkwd{lift}$,
which lifts the flat representation into the nested representations expected by
the predicate and update fields. Sequenced inserts flatten the contents of the
provided bag and map directly to an $\calcwd{insert}$.

\begin{figure}
{\footnotesize
    \[
      \bl
   \vtrans{\dbinsertseqann{\tma}{\tmb}} = \\
\quad \letintwo{\var{tbl}}{\vtrans{\tma}} \\
\quad \letinone{\var{rows}} \\
\qquad \forcomptwo{x}{\vtrans{\tmb}} \\
\qqquad \bag{
            \recordplustwo
                {\etaexp{\project{x}{\var{data}}}{\seq{\ell}}}
                {\recordterm{
                    \fieldstart{\project{x}{\var{start}}},
                    \fieldend{\project{x}{\var{end}}}
                }}
        } \\
\quad \calcwd{in} \\
\quad \dbinsert{\var{tbl}}{\var{rows}}
\spacerow\\
     \scaleleftright[0.45em]
         {\llparenthesis}
         {
             \bl
             \calcwd{update}^{(\ell_i : \tya_i)_{i \in I}} \: \calcwd{sequenced}
             \: (x \dbcomparrow \tmc) \\
             \quad \calcwd{between} \: \tma_1 \: \calcwd{and} \: \tma_2 \:
                   \calcwd{where} \: \tma_3 \\
             \quad \calcwd{set} \: (\ell_j = \tmb_j)_{j \in J}
             \el
         }
         {\rrparenthesis} = \vspace{0.5em}\\
     \qquad
       {
         \blt
         \letintwo{\var{tbl}}{\vtrans{L}} \\
         \letintwo{\var{aStart}}{\vtrans{M_1}} \\
         \letintwo{\var{aEnd}}{\vtrans{M_2}} \\
         \letintwo{\var{lRows}}{
             \mkwd{startRows}(\var{tbl}, \mkwd{pred}, \var{aStart})} \\
         \letintwo{\var{rRows}}{\mkwd{endRows}(\var{tbl}, \mkwd{pred},
             \var{aEnd})} \\
         \dbupdatetwo{x}{\var{tbl}} \\
         \quad \calcwhereone{\mkwd{pred} \wedge
         (\project{x}{\field{start}} < \var{aEnd}) \wedge
         (\project{x}{\field{end}} > \var{aStart}}) \\
         \quad \calcwd{set} \: \left(
             {\bl
             (\ell_j = \restrict{x}{\{\ell_i\}_{i \in I}}{\vtrans{\tmb_j}})_{j
             \in J},\\
             \var{start} = \greatesttwo{\project{x}{\var{start}}}{\var{aStart}},
             \\
             \var{end} = \leasttwo{\project{x}{\var{end}}}{\var{aEnd}} \\
             \el
             }
         \right); \\
        \dbinsert{\var{tbl}}{\var{lRows}}; \\
        \dbinsert{\var{tbl}}{\var{rRows}} \\
         \el
       } \\
       \text{where } \\
       \quad \mkwd{pred} \defeq \restrict{x}{\{\ell_i\}_{i \in
       I}}{\vtrans{M_3}} \\
        \quad \mkwd{startRows}(\var{tbl}, \var{pred}, \var{aStart}) \defeq
      \queryzero \\
        \qquad \forcomptwo{x}{\get{\var{tbl}}} \\
        \qqquad \whereone{\mkwd{pred} \wedge (\project{x}{\field{start}} <
            \var{aStart})
          \wedge (\project{x}{\field{end}} > \var{aStart}})
          \\
            \qqquad
            \bag{
                \etaexp
                    {x}
                    {\set{\ell_i}_{i \in I}}
                \recordplus
                \recordterm{
                    \fieldstart{\project{x}{\var{start}}},
                    \fieldend{\var{aStart}}
                    }
            } \\
            \quad \mkwd{endRows}(\var{tbl}, \var{pred}, \var{aEnd})
            \defeq \text{(symmetric)}
            \spacerow \\
       \scaleleftright[0.45em]
         {\llparenthesis}
         {
            \bl
            \calcwd{update}^{{(\ell_i : \tya_i)}_{i \in I}}\:\calcwd{nonsequenced} \: (x \dbcomparrow \tmc) \:
                  \calcwd{where} \: \tma \\
             \quad \calcwd{set} \: (\ell_j = \tmb_j)_{j \in J} \: \calcwd{valid} \:
            \calcwd{from} \: \tmb'_1 \: \calcwd{to} \: \tmb'_2
            \el
         }
         { \rrparenthesis }
    = \vspace{0.5em}
     \\
    \quad
    {
        \blt
        \calcwd{update}^{{(\ell_i:\tya_i)_{i \in I}}} \: (x \dbcomparrow
        \vtrans{\tmc}) \\
        \quad \calcwd{where}\:(\lift{x}{\vtrans{\tma}}) \\
        \quad \calcwd{set}\:
        {
            \left(
                \bl
                    (\recordterm{\ell_j = \lift{x}{\vtrans{\tmb_j}}}_{j \in J},
                    \\
                        \fieldstart{\lift{x}{\vtrans{\tmb'_1}}}, \\
                        \fieldend{\lift{x}{\vtrans{\tmb'_2}}})
                \el
            \right)
        }
        \el
    }
        \\
        \text{where } \lift{x}{f} \defeq \\
        \quad
        (\fun{x}{f}) \app
            \recordterm{
                \fielddata{\etaexp{x}{\{\ell_i\}_{i \in I}}},
                \fieldstart{\project{x}{\var{start}}},
                \fieldend{\project{x}{\var{end}}}}
      \el
    \]
}
\caption{Translation from $\vlinq$ into $\linq$ (selected cases)}
\label{fig:vlinq-translation}
\end{figure}

The remaining sequenced operations are the most complex to translate.
Since a sequenced modification may partition a row, the
\mkwd{startRows} and \mkwd{endRows} functions calculate the records
which must be inserted before and after the period of applicability.
To translate a sequenced update, we calculate the rows to insert,
perform an $\calcwd{update}$ to set the new values and
set the new period of applicability to the overlap between the PA and PV using
the $\calcwd{greatest}$ and $\calcwd{least}$ functions, and finally materialise
the insertions.
Sequenced deletions (shown in Appendix~\ref{app:full-definitions}) are similar
but delete the rows that overlap the PA instead of updating them.

\subsection{Metatheory}

Evaluation preserves typing and well-formedness.

\begin{proposition}[Preservation ($\vlinq$)]
    If $\tseq{\cdot}{\tma}{\tya}{\effs}$ and
    $\vtmevaltwo{\tma}{\retpair{\vala}{\db'}}$
    for some $\wf{\db}$, then
    $\tseq{\cdot}{\vala}{\tya}{\pure}$ and $\wf{\db'}$.
\end{proposition}

Unlike $\linq$ and $\tlinq$, evaluation in $\vlinq$ is \emph{partial}
in order to reflect the need for dynamic checks that start times
precede end times.
In practice, our implementation evaluates temporal updates as single transactions and raises an
exception (aborting the transaction) when a well-formedness check fails, but our formalisation
assumes  updates preserve well-formedness in order to avoid clutter.

Our translation from $\vlinq$ into $\tlinq$ satisfies the following correctness property:

\begin{restatable}{thm}{vtranscorrect}
    If $\tseq{\cdot}{\tma}{\tya}{\effs}$ and
    $\vtmevaltwo{\tma}{\retpair{\vala}{\db'}}$ for some
    $\wf{\db}$,
    then
        $\tmevaltwo
            [\flatten{\db}, \timevar]
            {\vtrans{\tma}}
            {\retpair
                {\vtrans{\vala}}
                {\flatten{\db'}}}$
\end{restatable}

\section{Sequenced Joins}\label{sec:temporal-joins}

Queries that join multiple tables are straightforward to encode using
language integrated query. Keeping with our employee database, say we
wish to separate out the salary into a separate table. The
non-temporal employee database might look as follows:

{\footnotesize
    \begin{minipage}{0.6\columnwidth}
        \header{employees}
\begin{tabular}{|c|c|c|}
    \hline
    name & position & band \\
    \hline
    Alice & Senior Lecturer & A08 \\
    \hline
    Bob & PhD Student & B01 \\
    \hline
    Charles & PhD Student & B01 \\
    \hline
    Dolores & Professor & A10 \\
    \hline
\end{tabular}
\end{minipage}
\hfill
\begin{minipage}{0.3\columnwidth}
\header{salaries}
\begin{tabular}{|c|c|}
    \hline
    band & salary \\
    \hline
    A08 & 40000 \\
    \hline
    A09 & 50000 \\
    \hline
    A10 & 70000 \\
    \hline
    B01 & 15000 \\
    \hline
\end{tabular}
\end{minipage}
}
We can get the salary for each employee
as follows:

\noindent
\begin{minipage}{0.65\columnwidth}
    {\footnotesize
\[
\bl
    \queryzero \\
    \quad \forcomptwo{e}{\get{\var{employees}}} \\
    \qquad \forcomptwo{s}{\get{\var{salaries}}} \\
    \qqquad \whereone{\project{e}{\var{band}} = \project{s}{\var{band}}} \\
    \qqquad \bag{
      \recordterm{
        \var{name} = \project{e}{\var{name}},
        \var{salary} = \project{s}{\var{salary}}
      }
    }
\el
\]
}
\end{minipage}
\hfill
\begin{minipage}{0.3\columnwidth}
    {\footnotesize
\begin{tabular}{|c|c|}
    \hline
    name & salary \\
    \hline
    Alice & 40000 \\
    \hline
    Bob & 15000 \\
    \hline
    Charles & 15000 \\
    \hline
    Dolores & 70000 \\
    \hline
\end{tabular}
    }
\end{minipage}

\noindent
Joining a temporal table with a non-temporal table is also easily expressible.
Consider a version of our previous temporal employees table from just after when
Dolores joined:

\begin{center}
{\footnotesize
\begin{tabular}{|c|c|c||c|c|}
    \hline
    name & position & band & start & end \\
    \hline
    Alice & Lecturer & A08 & 2010 & 2018 \\
    \hline
    Alice & Senior Lecturer & A09 & 2018 & $\infty$ \\
    \hline
    Bob & PhD Student & B01 & 2019 & 2023 \\
    \hline
    Charles & PhD Student & B01 & 2018 & 2022 \\
    \hline
    Dolores & Professor & A10 & 2022 & $\infty$ \\
    \hline
\end{tabular}\\
}
\end{center}
We can join this table with the non-temporal salaries table as follows;
for clarity, we denote valid-time $\calcwd{get}$ as
$\calcwd{get}_{\mathsf{V}}$:
{\footnotesize
\vspace{-0.25em}
\[
\bl
    \queryzero \\
    \quad \forcomptwo{e}{\getv{\var{employees}}} \\
    \qquad \forcomptwo{s}{\get{\var{salaries}}} \\
    \qqquad \whereone{\project{(\data{e})}{\var{band}} = \project{s}{\var{band}}} \\
    \qqquad \bag{
        \dbrow
        {
            \recordterm{\var{name} = \project{e}{\var{name}},
            \var{salary} = \project{s}{\var{salary}}}
        }
        {\timestart{e}}
        {\timeend{e}}
    }
\el
\]
\vspace{-0.25em}
}%
giving us the corresponding table in Section~\ref{sec:vlinq}.

Things get more interesting when \emph{both} tables are temporal. Salaries are
not static over time; bands go up with inflation, for example. Suppose we now
have two temporal tables. Consider the above table along with a temporal
salaries table showing a pay increase in 2015:

\begin{center}
{\footnotesize
\begin{tabular}{|c|c||c|c|}
    \hline
    band & salary & start & end \\
    \hline
    A08 & 38000 & 2000 & 2015 \\
    \hline
   A09 & 48000 & 2000 & 2015 \\
   \hline
    A08 & 40000 & 2015 & $\infty$ \\
    \hline
   A09 & 50000 & 2015 & $\infty$ \\
   \hline
  $\cdots$ &$\cdots$ & $\cdots$ & $\cdots$\\
    \hline
\end{tabular}
}
\end{center}
What does it mean to join two \emph{temporal} tables?
In essence, we want to record \emph{all} configurations of a particular joined
record, creating new records with shorter periods of validity whenever data from
either underlying table changes.
Concretely, joining the above two temporal tables would give:

\begin{center}
{\footnotesize
\begin{tabular}{|c|c||c|c|c|}
    \hline
    name & salary & start & end \\
    \hline
    Alice & 38000 & 2010 & 2015 \\
    \hline
    Alice & 40000 & 2015 & 2018 \\
    \hline
    Alice & 50000 & 2018 & $\infty$ \\
  \hline
  $\cdots$ &$\cdots$ & $\cdots$ & $\cdots$\\
    \hline
\end{tabular}\\
}
\end{center}%
Now there are records for Alice for\emph{three} different periods:

\begin{compactitem}
    \item The first when Alice was on salary band A08,
        conferring a salary of £38000.
    \item The second when band A08 increased to
        £40000.
    \item The third when Alice was promoted to band A09.
\end{compactitem}

Such joins are called \emph{sequenced} because they (conceptually) evaluate the
join on the whole sequence of states encoded by each table. Manually writing the
sequenced joins in SQL is error-prone. We instead introduce a construct,
$\calcwd{join}$, which allows us to write the following:

\vspace{-2mm}
{\footnotesize
\[
    \bl
    \joinzero \\
    \quad \forcomptwo{e}{\getv{\var{employees}}} \\
    \qquad \forcomptwo{s}{\getv{\var{salaries}}} \\
    \qqquad \whereone{\project{(\data{e})}{\var{band}} = \project{(\data{s})}{\var{band}}} \\
    \qqquad \bag{
        \recordterm{
                \var{name} = \project{(\data{e})}{\var{name}},
                \var{salary} = \project{(\data{s})}{\var{salary}}
        }
    }
    \el
\]
}%
Note that we \emph{do not} need to calculate the period of validity for each resulting
row; this is computed automatically.

\begin{figure}
    {\footnotesize
    \headerarg
        {Typing rules}
        {$\framebox{\hasfqtype{\tya}}~\framebox{$\vphantom{\hasfqtype{\tya}}\tseq{\tyenv}{\tma}{\tya}{\effs}$}$}
\begin{mathpar}
  \inferrule
  { }
  { \hasfqtype{\basety} }

  \inferrule
  { }
  { \hasfqtype{\recordterm{\ell_i : \basety_i}_i} }

  \inferrule
  { \tseq{\tyenv}{\tma}{\bagty{\tya}}{\effs} \\
    \hasfqtype{\tya} \\
    \effs \subseteq \effset{\effread}
  }
  { \tseq
        {\tyenv}
        {\join{\tma}}
        {\bagty{\validtimety{\tya}}}
        {\effs}
  }
\end{mathpar}

    \header{Normal forms}
    \[
      \begin{array}{lrcl}
        \text{Queries} & \querytm & ::= & \querycomp_1 \bagunion \cdots \bagunion \querycomp_n \\
        \text{Comprehensions} & \querycomp & ::= &
        \normfor{\seq{\querygen}}{\querybasetm}{\bag{\querynormtm}} \\
        \text{Generators} & \querygen & ::= & x \comparrow \get{\tablevar}
            \midspace x \comparrow \getv{\tablevar} \\
        \text{Normalised terms} & \querynormtm & ::= & \querybasetm \midspace
        \queryrecordtm \\ %
        \text{Base terms} & \querybasetm & ::= & c \midspace \project{x}{\ell}
        \midspace \langop{\oseq{\querybasetm}} \\
        \text{Record terms} & \queryrecordtm & ::= & \recordterm{\ell_i =
        \querybasetm_i}_i \\
      \end{array}
    \]

    \header{Translation on normal forms}
\vspace{0.5em}
\begin{mathpar}
  \normtrans{\join{\querytm}} = \query{\normtrans{\querytm}}

  \normtrans{\querycomp_1 \bagunion \cdots \bagunion \querycomp_n} =
    \normtrans{\querycomp_1} \bagunion \cdots \bagunion \normtrans{\querycomp_n}

  \normtrans{\querybasetm} = \querybasetm

  \normtrans{\queryrecordtm} = \queryrecordtm
\end{mathpar}

\[
  \begin{array}{l}
    \scaleleftright[0.35em]
    {\normtransl}
        {
        \bl
        \calcwd{for} \:
        {
            \left(
                \bl
                    x_1 \comparrow \getv{t_1}, \, \ldots, \, x_m \comparrow \getv{t_m}, \\
                    y_{1} \comparrow \getph{t'_{1}}, \,\ldots,\, y_n \comparrow
                    \getph{t'_n}
                \el
            \right)
        }
        \\
        \quad \whereone{\querybasetm} \\
        \quad \bag{\queryrecordtm}
        \el
        }
    {\normtransr}
    = \vspace{1em} \\
    \qqquad
    {
      \blt
        \calcwd{for} \:
        {
            \left(
                \bl
                    x_1 \comparrow \getv{t_1}, \, \ldots, \, x_m \comparrow \getv{t_m}, \\
                    y_{1} \comparrow \getph{t'_{1}}, \,\ldots,\, y_n \comparrow
                    \getph{t'_n}
                \el
            \right)
        } \\
          \quad \whereone{\querybasetm \wedge \var{joinStart} < \var{joinEnd}} \\
          \quad \bag{\dbrow{\queryrecordtm}{\var{joinStart}}{\var{joinEnd}}}
      \el
    }
    \spacerow \\
     \text{where} \:
          {
            \blt
              \var{joinStart} \defeq
              \calcwd{greatest}(\project{x_1}{\field{start}}, \ldots,
              \project{x_m}{\field{start}}) \\
              \var{joinEnd} \defeq \calcwd{least}(\project{x_1}{\field{end}},
              \ldots, \project{x_m}{\field{end}})
            \el
          }
  \end{array}
\]
}
    \caption{Sequenced joins}
    \label{fig:temporal-joins}
    \vspace*{-0.2cm}
\end{figure}

Figure~\ref{fig:temporal-joins} shows how sequenced joins can be implemented;
we show the constructs for valid time, but the same technique can be used
for transaction time. The
typing rule requires that the result of a $\calcwd{join}$ query is \emph{flat}
(nested sequenced queries are conceptually nontrivial).
As mentioned earlier, queries can be rewritten to normal forms for conversion to SQL, as shown in
Figure~\ref{fig:temporal-joins}.
The structure of these
normal forms allows sequenced joins to be implemented through a simple rewrite:
the $\calcwd{greatest}$ and $\calcwd{least}$ functions are used to calculate
the intersections of the periods of validity for each combination of records
from each generator, with the modified predicate ensuring that the periods of
overlap make sense. The calculated overlapping periods of validity
are then returned in the resulting row.

\section{Implementation and Case Study}\label{sec:implementation}

The Links programming language~\cite{CooperLWY06} is a statically-typed
functional web programming language which allows client, server, and
database code to be written in a uniform language.
We have extended Links with support for the constructs described in
Sections~\ref{sec:tlinq} and~\ref{sec:vlinq}, as well as support for temporal
joins as described in Section~\ref{sec:temporal-joins}.  In this section, we
describe a case study based on curating COVID-19 data.

Our translations from $\vlinq$ into $\linq$ are
trivially realisable in SQL. Queries can be compiled using known techniques
(e.g.,~\cite{Cooper09}).
The \mkwd{startRows} and \mkwd{endRows} functions can be compiled using an SQL
\verb+WITH+ statement, and there is a direct correspondence between $\linq$
modification operations and their SQL equivalents.
Each translated temporal modification is executed as an SQL transaction, with
primary key and referential integrity constraint checking deferred until the end
of the transaction.

\paragraph*{Case study.}

We have used the temporal features of Links in two
prototypes based on curated scientific databases: curation
of publicly available COVID-19 data, and storage and
curation of XML documents~\cite{GalpinSH22} using the Dynamic Dewey
labelling algorithm~\cite{XuLWB09}.

We concentrate on the first application\footnote{Links code is
available at \url{https://github.com/vcgalpin/links-covid-curation} \cite{Galpin22}}; a previous prototype which used
a preliminary version of the language design has previously
been presented as a short demo paper~\cite{GalpinC21}.
In 2020,
the Scottish Government began releasing various data about the
COVID-19 pandemic \cite{NRS:21}, which included weekly data
of fatalities in Scotland in various categories (such
as `Sex').  Each weekly release was a CSV
file, with a row for each subcategory (e.g., `Sex' has the subcategories
`Male' and `Female') and a column for each week for
which data was available. Each release included an additional week column with
the latest data (see Figure~\ref{data}). Importantly, each release
could include revisions to data for previous weeks.

\begin{figure*}
\includegraphics[trim=10 580 30 75,clip,width=18.0cm]{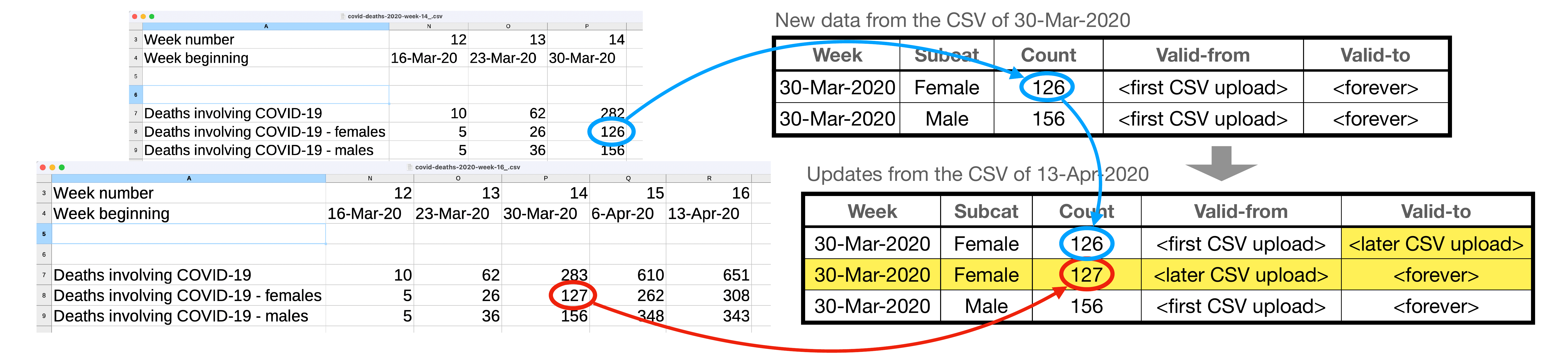}
\caption{Example of data uploads, sequenced insertion and sequenced update}
\label{data}
\end{figure*}

Information about the changes to the data over time is often desired
to understand its provenance and assess its
trustworthiness~\cite{buneman18sigmodrecord}.  From a provenance
point of view, this data is interesting because a column for an
earlier week may contain updated data. We developed a web application
for the querying of the data ("How do the Male and Female
subcategories compare in terms of the change in fatalities from
last week to this week?") as well as querying the changes in the
data ("How do the Male and Female subcategories compare in terms
of number of updates to existing values?").

Considering the non-temporal data, an entry in a database table
would be a row consisting of the key fields \lstinline{subcat} and
\lstinline{weekdate} and a value field giving the corresponding
\lstinline{count}. In the case of the temporal data, the key fields
are insufficient to uniquely identify the value of the \lstinline{count}
because it may have different values over time. Thus the time validity
fields are necessary to provide a key for the value.

The prototype uses a valid time table for fatality data to capture
the notion that a count value, either brand new or an update,
becomes valid as soon as the CSV is uploaded into the
interface\footnote{Other possibilities for
the start time of validity are the date of the release of the CSV
file or the start of the new week.}
(this can be a different time from when the new value is
accepted and written to the database).
In Links it is possible to specify the names of the period stamping fields,
which have the built-in type \lstinline{DateTime}.
This table is defined using the following Links code;
we have omitted some details in the code snippets for brevity.
\begin{lstlisting}
var covid_data =
  table "covid_data"
    with (subcat: Int, weekdate: String, count: Int)
    using valid_time(valid_from, valid_to)
    from database "covid_curation";
\end{lstlisting}
The prototype's upload workflow is as follows: the user uploads a
new CSV file, and the count values for the new week are added
to the database. For counts that pertain to earlier
weeks and that now have different values, the user is shown these counts
and can accept them, reject them or
move them to a pending list for a later decision.
In terms of implementation, brand new count values are
added to the table with a \emph{sequenced insert}, using the upload
time as the start time. The Links code for this and other examples can
be found in Appendix~\ref{appendix:code}.
The process is more complex for updated count values, because the interface
shows the user previous value, to support decision
making. This requires a conditional join over the current state
of the \lstinline{covid_data} table and the count values from the CSV
file.
If a modification is accepted, it is
added using a \emph{sequenced update}.
Figure~\ref{data} illustrates how the table changes as a result of
a single update.

The prototype also provides functionality to query data, both as current
data, and as data with information about changes. The current data
is obtained using a \emph{current query}. The result is a list of
weeks and counts grouped by subcategory. This is repeated for each
category.
\begin{minipage}{\columnwidth}
\begin{lstlisting}
fun getCurrentData (category) {
  query nested {
    for (x <-- subcategory)
      where (x.cat == category)
      [(subcat_name = x.subcat_name, cat = x.cat,
        results =
          for (y <- vtCurrent(covid_data))
            for (z <-- week)
              where (y.subcat == x.subcat &&
                     y.weekdate == z.weekdate &&
                     z.all_zero == false)
              [(count = y.count, weekdate = y.weekdate)])]}
}
\end{lstlisting}
\end{minipage}
Instead of an explicit \calcwd{get} construct, Links uses
the `double arrow' comprehension \lstinline+<--+
to represent a nontemporal database query, with \lstinline+<-t-+ and
\lstinline+<-v-+ supporting transaction time and valid time queries
respectively. The `single arrow' comprehension \lstinline+<-+ denotes a
list comprehension. Finally, \lstinline+vtCurrent+ is a standard library
function which performs a valid time query to obtain the values valid at the
current time.

For update provenance queries of individual counts,
a self join is computed over the \lstinline{subcategory} and
\lstinline{week} fields of the valid time table to provide a nested
result table where each count is associated with a list
of count values and their associated start and end time
information. This is a nonsequenced query because the time period
information is explicitly added to the result table.
The user can specify the subcategory and week they are
interested in, and obtain details of modifications.
The interface also supports update provenance by week and by category.
This is illustrated in Figure~\ref{screenshot}.

\begin{figure}
\includegraphics[trim=150 0 250 0,clip,width=7.5cm]{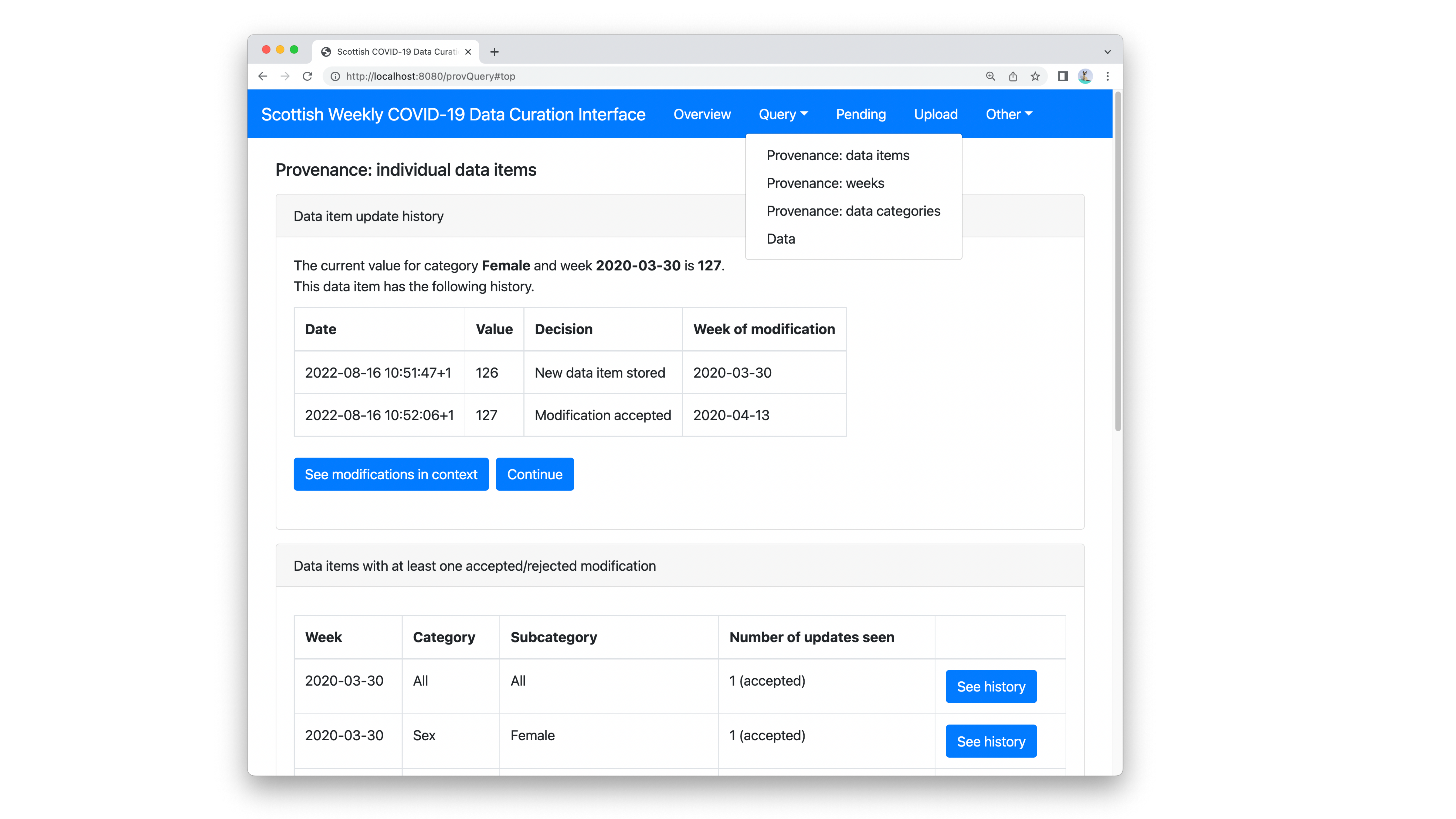}
\caption{Interface screenshot: history of a count}
\label{screenshot}
\vspace*{-0.5cm}
\end{figure}

\section{Discussion}\label{sec:discussion}

\paragraph*{Efficiency.}
Our main focus has been on \emph{portability}: by distilling a language design
and formalising and implementing a translation from temporal calculi to
non-temporal calculi, we allow temporal functionality to be used on a mainstream
DBMS. 
The cost of portability is that our translations will inevitably not perform as
well as a native implementation. Although we do not make any specific claims
about efficiency, we have no reason to believe that the performance is any
different to hand-translated SQL.

In particular, as discussed in~\secref{sec:implementation}, all translated
queries can be run directly on the database and do not require in-memory
processing.
Previous work on language-integrated query (e.g., ~\cite{CooperLWY06,
cheney14sigmod}) shows how the ``nested loop'' style of query is translated into
efficient SQL, and our translation of queries happens prior to normalisation.
Further optimisation is subsequently performed by the DBMS, and anecdotally we
have not observed any performance issues in any applications we have written
using the temporal extensions to Links.

\paragraph*{Supporting existing native implementations.}
Some (mainly proprietary) DBMSs, for example Teradata~\cite{Al-KatebGCBCP13}, have native
support for some temporal features inspired by TSQL2 or SQL:2011.  Although we
are yet to explore this, an advantage of the LINQ approach is that temporal SQL
syntax could be generated for backends which support temporal operations
directly, while maintaining functionality in mainstream backends without native
temporal support.

The translation between our temporal modification constructs and SQL:2011 (and
similarly, TSQL2) would be fairly direct. Consider a sequenced update in
$\vlinq$:

\vspace{-1em}
{\footnotesize
\[
    \bl
    \calcwd{update} \: \calcwd{sequenced} \: (x \dbcomparrow \var{employees}) \\
    \quad \calcwd{between} \: 2022{-}01{-}1 \: \calcwd{and} \: 2022{-}10{-}29  \:
          \calcwd{where} \: (\project{x}{\var{salary}} < 30000) \\
    \quad \calcwd{set} \: (\var{salary} = \project{x}{\var{salary}} + 1000)
    \el
\]
}

This could be implemented in SQL:2011 as follows:

\begin{verbatim}
UPDATE employees
  FOR PORTION OF EPeriod
    FROM DATE '2022-01-01' TO DATE '2022-10-29'
  WHERE salary < 30000
  SET salary = salary + 1000
\end{verbatim}

Moving between the two DBMSs would not require any changes to application source
code. However, not all operations can be as straightforwardly translated: in
particular, SQL:2011 does not natively support sequenced joins.

\section{Related and Future Work}\label{sec:related}

Most of the focus of effort on
language-integrated query has been (perhaps unsurprisingly) on queries
rather than updates, beginning with the foundational work on nested relational
calculus by Buneman et al.~\cite{BNTW95} and on rewriting queries for
translation to SQL by Wong~\cite{wong96jcss}.  Lindley and
Cheney~\cite{lindley12tldi} presented a calculus including both query and update
capabilities and our type and effect system for tracking database read and write
access is loosely based on theirs.  More recently a number of contributions
extending the formal foundations of language-integrated query have appeared,
including to handle higher-order functions~\cite{Cooper09}, nested query
results~\cite{cheney14sigmod}, sorting/ordering~\cite{kiselyov17aplas}, grouping
and aggregation~\cite{okura20flops,okura20gpce}, and
deduplication~\cite{ricciotti21esop}.  Our core calculus $\linq$ only
incorporates the first two of these,  and developing a core calculus that
handles more features, as well as translating temporal queries involving them,
is an obvious future direction.
To the best of our knowledge no previous work on language-integrated query has
considered temporal data specifically.

The ubiquity and importance of time in applications of databases was
appreciated from an early stage~\cite{clifford83tods} and led to a
significant community effort to standardize temporal extensions to SQL
based on the TSQL2 language design in the 1990s and early
2000s~\cite{tsql2}.  This effort ultimately resulted in
standardisation of a relatively limited subset of the original
proposal in SQL:2011~\cite{kulkarni12sigmodrecord}.  Since then
temporal database research has progressed steadily, including recent
contributions showing how to implement temporal data management as a
layer on top of a standard RDBMS~\cite{dignos16tods}, and establishing
connections between temporal querying and data
provenance and annotation models~\cite{dignos19vldb}.

Snodgrass~\cite{Snodgrass99:book} describes how to implement
TSQL2-style updates and queries by translation to SQL, but we are not
aware of previous detailed formal proofs of correctness of
translations for transaction time and valid time updates.  Although
timestamping rows with time intervals is among the most popular ways
to represent temporal databases as flat relational tables, it is not
the only possibility.  Jensen et al.~\cite{jensen94is} proposed a
\emph{bitemporal conceptual data model} that captures the abstract
meaning of a temporal table and used it to compare different
representation strategies.

There can be multiple representations of the same abstract temporal data, leading to consideration of the problem of coalescing or normalizing the intervals to save space and avoid ambiguity.  Nonsequenced updates can be used to perform modifications that have different effects on representations of the same conceptual table.  We have not considered coalescing or other common issues such as how to handle operations such as deduplication, grouping and aggregation (including emptiness testing), or integrity constraints in a  temporal setting.  Some of these issues appear orthogonal to the high-level language design and could be incorporated ``under the hood'' into the implementation or even performed directly on the database.

One important future application is to retrofit temporal aspects to
expert-curated databases, an example being the Guide to Pharmacology Database
(GtoPdb) that summarises pharmacological targets and
interactions~\cite{ArmAFDx:20}.  Links has been used to implement a workalike
version of GtoPdb~\cite{FowlFHSC20} and we hope to build on this to provide a
fully versioned implementation of GtoPdb.  An important requirement here is to
minimize changes to the existing system.
Finally we mention two immediate next steps. 
First, we plan to investigate \emph{bitemporal}
databases~\cite{SnodgrassA85:taxonomy} allow transaction and valid
time to be used together, allowing us to write queries such as ``when
was it recorded that Bob's contract length was
extended?''; bitemporal databases are considerably more difficult to formalise
and reason about, so we aim to investigate how bitemporal support can be added
in a compositional manner.
Second, at present, the result of a sequenced join must be a flat
record; further work is required to understand the semantics and
implementation techniques for joins that produce nested results.

\section{Conclusions}\label{sec:conclusion}
In spite of decades of work on temporal databases and even an extension to the SQL standard,
mainstream support for temporal data remains limited, requiring developers to
implement temporal functionality from scratch.
In this paper, we have shown how to extend language-integrated query
to support transaction time and valid time data, making
temporal data management accessible without explicit DBMS support. We
have formalised our constructs and translational implementation
strategies based on those proposed
by~\citet{Snodgrass99:book}, and proved that the translations are
semantics-preserving.  We have implemented our approach in the Links
programming language and assessed its value through a case study.  Our
work is a first but significant step towards fully supporting temporal
data management at the language level.

\section*{Acknowledgements}
We thank the anonymous reviewers for their helpful comments.  This work is
partially funded by EPSRC Grant EP/T014628/1 (STARDUST), ERC Consolidator Grant
682315 (Skye), and a UK Government ISCF Metrology Fellowship.

\newpage

\bibliographystyle{ACM-Reference-Format}
\bibliography{tlinq-arxiv}

\clearpage
\appendix
\onecolumn
\section{Full definitions}
\label{app:full-definitions}

In this appendix full definitions are provided for judgments and definitions where only selected rules could be included in the main body of the paper.

\begin{itemize}
\item Figure~\ref{fig:linq-typing-full} presents the full definition of the typechecking judgment for $\linq$, complementing the partial definition in Figure~\ref{fig:linq-typing}.
\item Figure~\ref{fig:linq-semantics-full} presents the full definition of the evaluation relation for $\linq$, including the rules omitted from Figure~\ref{fig:linq-semantics}.
\item Figure~\ref{fig:vlinq-syntax-full} shows the full typing rules for $\vlinq$, complementing the selected rules in Figure~\ref{fig:vlinq-syntax}.
  \item Figure~\ref{fig:vlinq-reduction-full} shows the full evaluation rules for $\vlinq$, complementing the selected rules in Figure~\ref{fig:vlinq-reduction}.
  \item Figure~\ref{fig:vlinq-translation-full} presents the full translation rules for sequenced and nonsequenced operations in $\vlinq$ to $\linq$, completing the rules shown in Figure~\ref{fig:vlinq-translation}.  The rules for translating current updates are special cases of those for sequenced updates and shown in Appendix~\ref{sec:vlinq-current-updates}.
\end{itemize}

\begin{figure*}[b]
{\footnotesize
  ~\header{Term typing} \hfill \framebox{$\tseq{\tyenv}{\tma}{\tya}{\effs}$}
  \vspace{-1em}
  \begin{mathpar}
    \inferrule
    [T-Const]
    { c \text{ has base type } \basety }
    { \tseq{\tyenv}{c}{\basety}{\pure} }

    \inferrule
    [T-Var]
    { x : A \in \Gamma }
    { \tseq{\tyenv}{x}{\tya}{\pure} }

    \inferrule
    [T-Abs]
    { \tseq
        {\tyenv, x : \tya}
        {\tma}
        {\tyb}
        {\effs}
    }
    { \tseq
        {\Gamma}
        {\fun{x}{\tma} }
        {\tyfun{\tya}{\tyb}{\effs}}
        {\pure}
    }

    \inferrule
    [T-App]
    { \tseq{\tyenv}{\tma}{\tyfun{\tya}{\tyb}{\effs_1}}{\effs_2} \\
      \tseq{\tyenv}{\tmb}{\tya}{\effs_3} \\
    }
    { \tseq{\tyenv}{\tma \app \tmb}{\tyb}{\effs_1 \cup \effs_2 \cup \effs_3} }

    \inferrule
    [T-If]
    {
        \tseq{\tyenv}{\tmc}{\boolty}{\effs_1} \\\\
        \tseq{\tyenv}{\tma}{\tya}{\effs_2} \\
        \tseq{\tyenv}{\tmb}{\tya}{\effs_3}
    }
    {
        \tseq
            { \tyenv }
            { \ite{\tmc}{\tma}{\tmb} }
            { \tya }
            {\effs_1 \cup \effs_2 \cup \effs_3}
    }

    \inferrule
    [T-Op]
    { \opsymb : \tyfun{\basety_1 \times \cdots \times \basety_n}{\basety}{\effs}
        \\
        (\tseq{\tyenv}{\tma_i}{\basety_i}{\effs_i})_{i \in 1..n}
    }
    { \tseq{\tyenv}{\langop{\oseq{M}}}{\basety}{\effs \cup \bigcup_{i \in 1..n} \effs_i} }

    \inferrule
    [T-Record]
    { (\tseq{\tyenv}{\tma_i}{\tya_i}{\effs_i})_{i} }
    { \tseq{\tyenv}{\recordterm{\ell_i = \tma_i}_i}{\recordty{\ell_i : A_i}_i }{\bigcup_{i} \effs_i} }

    \inferrule
    [T-Project]
    { \tseq{\tyenv}{\tma}{\recordty{\ell_i : \tya_i}_{i \in I}}{\effs} \\ j \in I }
    { \tseq{\tyenv}{\project{\tma}{\ell_j}}{\tya_j}{\effs} }

    \inferrule
    [T-EmptyBag]
    { }
    { \tseq{\tyenv}{\bag{~}}{\bagty{\tya}}{\pure} }

    \inferrule
    [T-Bag]
    { \tseq{\tyenv}{\tma}{\tya}{\effs} }
    { \tseq{\tyenv}{\bag{\tma}}{\bagty{\tya}}{\effs} }

    \inferrule
    [T-BagUnion]
    { \tseq{\tyenv}{\tma}{\bagty{\tya}}{\effs_1} \\\\
      \tseq{\tyenv}{\tmb}{\bagty{\tya}}{\effs_2} }
    { \tseq{\tyenv}{\baguniontwo{\tma}{\tmb}}{\bagty{\tya}}{\effs_1 \cup \effs_2} }

    \inferrule
    [T-For]
    { \tseq{\tyenv}{\tma}{\bagty{\tya}}{\effs_1} \\\\
      \tseq{\tyenv, x : \tya}{\tmb}{\bagty{\tyb}}{\effs_2}
    }
    { \tseq{\tyenv}{\forcomp{x}{M}{N}}{\bagty{\tyb}}{\effs_1 \cup \effs_2} }

    \inferrule
    [T-Query]
    { \tseq{\tyenv}{\tma}{\bagty{\tya}}{\effs} \\\\
      \hasqtype{\tya} \\
      \effs \subseteq \effset{\effread}
    }
    { \tseq{\tyenv}{\query{\tma}}{\bagty{\tya}}{\effs} }

    \inferrule
    [T-Table]
    { }
    { \tseq{\tyenv}{\tblvar}{\tablety{\schema(t)}}{\pure}}

    \inferrule
    [T-Now]
    { }
    { \tseq{\tyenv}{\now}{\timety}{\pure} }

    \inferrule
    [T-Get]
    { \tseq{\tyenv}{\tma}{\tablety{\tya}}{\effs}}
    { \tseq
        {\tyenv}
        {\get{\tma}}
        {\bagty{\tya}}
        {\effset{\effread} \cup \effs}
    }

    \inferrule
    [T-Insert]
    { \tseq{\tyenv}{\tma}{\tablety{\tya}}{\effs} \\
      \tseq{\tyenv}{\tmb}{\bagty{\tya}}{\pure}
    }
    { \tseq{\tyenv}{\dbinsert{\tma}{\tmb}}{\recordty{}}{\effset{\effwrite} \cup \effs} }

    \inferrule
    [T-Update]
    { \tseq{\tyenv}{\tmc}{\tablety{\tya}}{\effs} \\\\
      \tya = \recordty{\ell_i : \tyb_i}_{i \in I} \\
      \tseq{\tyenv, x : \tya}{\tma}{\boolty}{\pure} \\
      (j \in I \wedge \tseq{\tyenv, x : \tya}{\tmb_j}{\tyb_j}{\pure})_{j \in J}
    }
    { \tseq
        {\tyenv}
        {\dbupdate{x}{\tmc}{\tma}{(\ell_j = \tmb_j)_{j \in J}}}
        {\recordty{}}
        {\effset{\effwrite} \cup \effs }
    }

    \inferrule
    [T-Delete]
    { \tseq{\tyenv}{\tma}{\tablety{\tya}}{\effs} \\
      \tseq{\tyenv, x : \tya}{\tmb}{\boolty}{\pure}
    }
    { \tseq
        {\tyenv}
        {\dbdelete{x}{\tma}{\tmb}}
        {\recordty{}}
        {\effset{\effwrite} \cup \effs }
    }
  \end{mathpar}

}

    \caption{Typing rules for $\linq$}
    \label{fig:linq-typing-full}
{\footnotesize

    \headersig{Big-step reduction rules}{$\tmevaltwo{\tma}{\retpair{\vala}{\db'}} $}
    \begin{mathpar}
        \inferrule
        [E-Val]
        { }
        { \tmevaltwo{\vala}{\retpair{\vala}{\db}}}

        \inferrule
        [E-App]
        {
        \tmevaltwo{\tma}{\retpair{\fun{x}{\tmc}}{\db_1}} \\\\
        \tmevaltwo[\db_1, \timevar]{\tmb}{\retpair{\vala}{\db_2}} \\
        \tmevaltwo[\db_2, \timevar]{\subst{\tmc}{\vala}{x}}{\retpair{\valb}{\db_3}}  }
        { \tmevaltwo{ \tma \app \tmb}{\retpair{\valb}{\db_3}} }

        \inferrule
        [E-Op]
        {
            \tmevaltwo{\tma_1}{\retpair{\vala_1}{\db_1}} \\\\
            \ldots \\\\
            \tmevaltwo[\db_{n - 1}, \timevar]{\tma_n}{\retpair{\vala_n}{\db_n}}
        }
        {
          \tmevaltwo{\langop{\oseq{\tma}}}{\retpair{\denotlangop{\oseq{\vala}}}{\db_n} }
        }

        \inferrule
        [E-IfT]
        {
            \tmevaltwo{\tmc}{\retpair{\ttrue}{\db_1}} \\
            \tmevaltwo[\db_1, \timevar]{\tma}{\retpair{\vala}{\db_2}} \\
        }
        {
            \tmevaltwo{\ite{\tmc}{\tma}{\tmb}}{\retpair{\vala}{\db_2}}
        }

        \inferrule
        [E-IfF]
        {
            \tmevaltwo{\tmc}{\retpair{\ffalse}{\db_1}} \\
            \tmevaltwo[\db_1, \timevar]{\tmb}{\retpair{\vala}{\db_2}} \\
        }
        {
            \tmevaltwo{\ite{\tmc}{\tma}{\tmb}}{\retpair{\vala}{\db_2}}
        }

        \inferrule
        [E-Bag]
        {
            \tmevaltwo{\tma}{\retpair{\vala}{\db'}}
        }
        {
          \tmevaltwo{\bag{\tma}}{\retpair{\bag{\vala}}{\db'} }
        }

        \inferrule
        [E-BagUnion]
        {
            \tmevaltwo{\tma}{\retpair{\vala}{\db_1}} \\
            \tmevaltwo[\db_1, \timevar]{\tmb}{\retpair{\valb}{\db_2}} \\
        }
        {
            \tmevaltwo
                {\baguniontwo{\tma}{\tmb}}
                {\retpair
                    {\vala \denotbagunion \valb}
                    {\db_2}
                }
        }

        \inferrule
        [E-Record]
        {
            \tmevaltwo{\tma_1}{\retpair{\vala_1}{\db_1}} \\\\
            \ldots \\\\
            \tmevaltwo[\db_{n - 1}, \timevar]{\tma_n}{\retpair{\vala_n}{\db_n}}
        }
        {
          \tmevaltwo{\recordterm{\seq{\ell = \tma}}}{\retpair{\recordterm{\seq{\ell = \vala}}}{\db_n} }
        }

        \inferrule
        [E-Project]
        {
            \tmevaltwo{\tma}{\retpair{\recordterm{\ell_i = \vala_i}_{i \in I}}{\db'}} \\
            j \in I
        }
        { \tmevaltwo{\project{\tma}{\ell_j}}{\retpair{\vala_j}{\db'}} }

        \inferrule
        [E-Now]
        { }
        { \tmevaltwo{\now}{\retpair{\timevar}{\db}} }

        \inferrule
        [E-Query]
        { \readtmevaltwo{\tma}{\vala}}
        { \tmevaltwo{\query{\tma}}{\retpair{\vala}{\db}} }

        \inferrule
        [E-ForEmpty]
        {
            \tmevaltwo{\tma}{\retpair{\emptybag}{\db'}}
        }
        {
            \tmevaltwo{\forcomp{x}{\tma}{\tmb}}{\retpair{\emptybag}{\db'}}
        }

        \inferrule
        [E-For]
        {
            \tmevaltwo
                {\tma}
                {\retpair{\bag{\seq{\vala_1} \cdot \vala \cdot \seq{\vala_2}}}{\db_1}}
            \\\\
            \tmevaltwo[\db_1,
            \timevar]{\subst{\tmb}{\vala}{x}}{\retpair{\valb_1}{\db_2}}
            \\
            \tmevaltwo
                [\db_2, \timevar]
                {\forcomp
                    {x}
                    {\bag{\seq{\vala_1} \cdot \seq{\vala_2}}}
                    {\tma}
                }
                {\retpair{\valb_2}{\db_3}}
        }
        {
            \tmevaltwo{
                \forcomp
                    {x}
                    {\tma}
                    {\tmb}
            }
            {\retpair
                {\valb_1 \denotbagunion \valb_2}
                {\db_3}
            }
        }

        \inferrule
        [E-Get]
        {
            \tmevaltwo{\tma}{\retpair{\tablevar}{\db'}}
        }
        { \tmevaltwo{\get{\tma}}{\retpair{\db'(\tblvar)}{\db'}} }

        \inferrule
        [E-Insert]
        {
            \tmevaltwo{\tma}{\retpair{\tblvar}{\db'}} \\
            \puretmevaltwo{\tmb}{\vala}
        }
        {
            \tmevaltwo
                {\dbinsert{\tma}{\tmb}}
                {\retpair{()}{\extendenv{\db'}{\tblvar}{\db'(\tblvar) \denotbagunion \vala}} }
        }

        \inferrule
        [E-Update]
        {
            \tmevaltwo{\tmc}{\retpair{\tblvar}{\db_1}} \\
                    \db_2 =
                        {\extendenv
                            {\db_1}
                            {\tblvar}
                            { \bag{\mkwd{upd}(\var{v}) \mid \var{v} \in \db_1(\tblvar)}}
                        }
                    \\\\
                    {
                    \mkwd{upd}(\var{v}) =
                    {
                        \begin{cases}
                            \recordwithtwo
                                {\var{v}}
                                {\seq{\ell = \valb}}
                                & \text{ if }
                                \puretmevaltwo{\tma \{ \var{v} / x \}}{\ttrue}
                                \text{ and }
                                (\puretmevaltwo{\subst{\tmb_i}{\var{v}}{x}}{\valb_i})_i \\
                            v & \text{ if } \puretmevaltwo{\tma \{ \var{v} / x \}}{\ffalse}
                        \end{cases}
                    }
                    }
        }
        {
            \tmevaltwo
                {\dbupdate{x}{\tmc}{\tma}{(\seq{\ell = \tmb})}}
                {\retpair{()}{\db_2}}
        }

        \inferrule
        [E-Delete]
        {
            \tmevaltwo{\tma}{\retpair{\tblvar}{\db_1}}
            \\
            \db_2 =
                {\extendenv
                    {\db_1}
                    {\tblvar}
                    { \bag{\var{v} \in \db(\tblvar) \mid
                        \puretmevaltwo{\tmb \{\var{v} / x \}}{\ffalse}} }
                }
        }
        {
            \tmevaltwo
                {\dbdelete{x}{\tma}{\tmb}}
                {\retpair
                    {()}
                    {\db_2}
                }
        }
    \end{mathpar}
}
    \caption{Semantics of $\linq$}
    \label{fig:linq-semantics-full}
\end{figure*}

\begin{figure*}
{\footnotesize

~\header{Typing rules} \hfill \framebox{$\Gamma \vdash M : A \effann{\effs}$}
\vspace{-2em}
\begin{mathpar}
  \inferrule
  [TV-Get]
  { \tseq{\tyenv}{\tma}{\tablety{\tya}}{\effs} }
  { \tseq
      {\tyenv}
      {\get{\tma}}
      {\bagty{\validtimety{\tya}}}
      {\effset{\effread} \cup \effs}
  }

  \inferrule
  [TV-SeqInsert]
  {
      \tseq{\tyenv}{\tma}{\tablety{\tya}}{\effs} \\
      \tseq{\tyenv}{\tma}{\bagty{\validtimety{\tya}}}{\pure}
  }
  { \tseq{\tyenv}{\dbinsertseq{\tma}{\tmb}}{\recordty{}}{\effset{\effwrite} \cup \effs} }

  \inferrule
  [TV-NonseqUpdate]
  {
      \tseq{\tyenv}{\tmc}{\tablety{\tya}}{\effs} \\
      \tya = \recordty{\ell_i : \tyb_i}_{i \in I} \\\\
      \tseq{\tyenv, x : \validtimety{\tya}}{\tma}{\boolty}{\pure} \\
      (j \in I \wedge \tseq{\tyenv, x : \validtimety{\tya}}{\tmb_j}{\tyb_j}{\pure})_{j \in J} \\\\
      \tseq{\tyenv, x : \validtimety{\tya}}{\tmb'_1}{\timety}{\pure} \\
      \tseq{\tyenv, x : \validtimety{\tya}}{\tmb'_2}{\timety}{\pure}
  }
  {
      \tseq
        {\tyenv}
        { \dbupdatenonseq{x}{\tmc}{\tma}{(\ell_j {=} \tmb_j)_{j \in J}}{\tmb'_1}{\tmb'_2} }
        { () }
        { \effset{\effwrite} \cup \effs }
  }

  \inferrule
  [TV-SeqUpdate]
  { \tseq{\tyenv}{\tmc}{\tablety{\tya}}{\effs} \\
      \tya = \recordty{\ell_i : \tyb_i}_{i \in I} \\
    \tseq{\tyenv}{\tma_1}{\timety}{\pure} \\
    \tseq{\tyenv}{\tma_2}{\timety}{\pure} \\\\
    \tseq{\tyenv, x : \tya}{\tma_3}{\boolty}{\pure} \\
    (j \in I \wedge \tseq{\tyenv, x : \tya}{\tmb_j}{\tyb_j}{\pure})_{j \in J}
  }
  { \tseq
      {\tyenv}
      {\dbupdatebetween{x}{\tmc}{\tma_1}{\tma_2}{\tma_3}{(\ell_j = \tmb_j)_{j \in J}}}
      {()}
      {\effset{\effwrite} \cup \effs}
  }

  \inferrule
  [TV-NonseqDelete]
  {
      \tseq{\tyenv}{\tma}{\tablety{\tya}}{}{\effs} \\
      \tseq{\tyenv, x : \validtimety{\tya}}{\tmb}{\boolty}{\pure}
  }
  { \tseq
        {\tyenv}
        {\dbdeletenonseq{x}{\tma}{\tmb}}
        {()}
        {\effset{\effwrite} \cup \effs}
  }

  \inferrule
  [TV-SeqDelete]
  {
    \tseq{\tyenv}{\tmc}{\tablety{\tya}}{\effs}\!\! \\
    \tseq{\tyenv}{\tma_1}{\timety}{\pure}\!\! \\
    \tseq{\tyenv}{\tma_2}{\timety}{\pure}\!\! \\
    \tseq{\tyenv, x {:} \tya}{\tmb}{\boolty}{\pure}
  }
  {
      \tseq
        {\tyenv}
        {\dbdeletebetween{x}{\tmc}{\tma_1}{\tma_2}{\tmb}}
        {()}
        {\effset{\effwrite} \cup \effs}
  }
\end{mathpar}
}
\caption{Typing rules for \vlinq}
\label{fig:vlinq-syntax-full}
\end{figure*}

\begin{figure*}[t]
    {\footnotesize
~\headersig{Reduction rules}{$\vtmevaltwo{\tma}{\retpair{\vala}{\db'}}$}
\begin{mathpar}
  \inferrule
  [EV-Row]
  {
      \vtmevaltwo{\tma_1}{\retpair{\vala_1}{\db_1}} \\\\
      \vtmevaltwo[\timevar, \db_1]{\tma_2}{\retpair{\vala_2}{\db_2}} \\
      \vtmevaltwo[\timevar, \db_2]{\tma_3}{\retpair{\vala_3}{\db_3}}
  }
  { \vtmevaltwo
      {\dbrow{\tma_1}{\tma_2}{\tma_3}}
      {\retpair{\dbrow{\vala_1}{\vala_2}{\vala_3}}{\db_3}}
  }

  \inferrule
  [EV-SeqInsert]
  {
      \vtmevaltwo{\tma}{\retpair{\tblvar}{\db_1}} \\
      \puretmevaltwo{\tmb}{\bag{\seq{\vala}}} \\\\
      \forall \dbrow{\var{data}}{\var{start}}{\var{end}} \in \seq{\vala}. \var{start}
      < \var{end} \\\\
      \db_2 = \extendenv{\db_1}{\tblvar}{\db_1(\tblvar) \denotbagunion
      \bag{\seq{\vala}}}
  }
  {
    \vtmevaltwo{\dbinsertseq{\tma}{\tmb}}{\retpair{()}{\db_2}}
  }

  \inferrule
  [EV-NonseqDelete]
    {
        \tmevaltwo{\tma}{\retpair{\tblvar}{\db_1}}
        \\\\
        \db_2 =
            {\extendenv
                {\db_1}
                {\tblvar}
                { \bag{d \in \db(\tblvar) \mid
                    \puretmevaltwo{\tmb \{d / x \}}{\ffalse}} }
            }
    }
    {
        \tmevaltwo
            { \dbdeletenonseq{x}{\tma}{\tmb} }
            {\retpair{()}{\db_2}}
    }

    \inferrule
    [EV-SeqDelete]
    {
        \tmevaltwo{\tmc}{\retpair{\tblvar}{\db_1}} \\
        \puretmevaltwo{\tma_1}{\vala_{\var{start}}} \\
        \puretmevaltwo{\tma_2}{\vala_{\var{end}}} \\
        \vala_{\var{start}} < \vala_{\var{end}} \\
      \db_2 = \db_1[ \tablevar \mapsto
        \flattenbag{\bag{\mkwd{del}(d) \midspace d \in \db(\tablevar)}} ] \\\\
        {
                \mkwd{del}(\dbrow{\var{v}}{\var{start}}{\var{end}}) \defeq
                {
                  \begin{cases}
                    \bag{~} &
                    \text{ if } \puretmevaltwo{\subst{\tmb}{\var{v}}{x}}{\ttrue}
                    \text{ and } \vala_{\var{start}} \le \var{start} \text{ and
                    } \vala_{\var{end}} \ge \var{end} \quad \textbf{(Case 1)} \\
                    \bag{\dbrow{\var{v}}{\vala_{\var{end}}}{\var{end}}} &
                    \text{ if } \puretmevaltwo{\subst{\tmb}{\var{v}}{x}}{\ttrue}
                    \text{ and } \vala_{\var{start}} \le \var{start} \text{ and
                    } \vala_{\var{end}} < \var{end} \quad \textbf{(Case 2)}\\
                    \bag{\dbrow{\var{v}}{\var{start}}{\vala_{\var{start}}}, \dbrow{\var{v}}{\vala_{\var{end}}}{\var{end}}} &
                    \text{ if } \puretmevaltwo{\subst{\tmb}{\var{v}}{x}}{\ttrue}
                    \text{ and } \vala_{\var{start}} > \var{start} \text{ and }
                    \vala_{\var{end}} < \var{end} \quad \textbf{(Case 3)} \\
                    \bag{\dbrow{\var{v}}{\var{start}}{\vala_{\var{start}}}} &
                    \text{ if } \puretmevaltwo{\subst{\tmb}{\var{v}}{x}}{\ttrue}
                    \text{ and } \vala_{\var{start}} > \var{start} \text{ and }
                    \vala_{\var{end}} \ge \var{end} \quad \textbf{(Case 4)}\\
                    \bag{\dbrow{\var{v}}{\var{start}}{\var{end}}} & \text{ otherwise }
                    \hfill \textbf{(Case 5)}
                  \end{cases}
                }
        }
    }
    {
        \vtmevaltwo
            {\dbdeletebetween{x}{\tmc}{\tma_1}{\tma_2}{\tmb}}
            { \retpair{\recordterm{}}{\db_2}}
    }

    \inferrule
    [EV-NonseqUpdate]
    {
      \tmevaltwo{\tmc}{\retpair{\tblvar}{\db_1}} \\
              \db_2 =
                  {\extendenv
                      {\db_1}
                      {\tblvar}
                      { \bag{\mkwd{upd}(d) \mid d \in \db_1(\tblvar)}}
                  }
              \\\\
              {
              \mkwd{upd}(D = \dbrow{\var{v}}{\var{start}}{\var{end}}) =
              {
                  \begin{cases}
                      \dbrow
                          {
                              \recordwithtwo
                                  {\var{v}}
                                  {\seq{\ell = \valb}}
                          }
                          { \valb_{\var{start}} }
                          { \valb_{\var{end}} }
                      & \text{ if }
                      {\blt
                              \puretmevaltwo{\subst{\tma}{D}{x}}{\ttrue}
                              \text{ and }
                              (\puretmevaltwo{\subst{\tmb_i}{D}{x}}{\valb_i})_i
                              \text{ and } \\
                              \puretmevaltwo{\subst{\tmb'_1}{D}{x}}{\valb_{\var{start}}}
                              \text{ and }
                              \puretmevaltwo{\subst{\tmb'_2}{D}{x}}{\valb_{\var{end}}}
                              \text{ and }  \valb_{\var{start}} < \valb_{\var{end}}
                       \el}
                              \\
                      D & \text{ if } \puretmevaltwo{\tma \{ D / x \}}{\ffalse}
                  \end{cases}
              }
              }
  }
  {
      \tmevaltwo
          { \dbupdatenonseq{x}{\tmc}{\tma}{(\ell_i = \tmb_i)_{i}}{\tmb'_1}{\tmb'_2} }
          {\retpair{()}{\db_2}}
  }

    \inferrule
    [EV-SeqUpdate]
    {
      \vtmevaltwo{\tmc}{\retpair{\tablevar}{\db_1}} \\
        \puretmevaltwo{\tma_1}{\vala_{\var{start}}} \\
        \puretmevaltwo{\tma_2}{\vala_{\var{end}}} \\
        \vala_{\var{start}} < \vala_{\var{end}}
        \\
      \db_2 = \db_1[ \tablevar \mapsto
        \flattenbag{\bag{\mkwd{upd}(d) \midspace d \in \db_1(\tablevar)}} ]
      \\
        \mkwd{upd}(\dbrow{\var{v}}{\var{start}}{\var{end}}) =
        {
          \begin{cases}
            \bag{\dbrow{W}{\var{start}}{\var{end}}} &
            \text{ if } \puretmevaltwo{\subst{\tma_3}{\var{v}}{x}}{\ttrue} \text{
            and } \vala_{\var{start}} \le \var{start} \text{ and }
            \vala_{\var{end}} \ge \var{end} \quad \textbf{(Case 1)}\\
            \bag{\dbrow{W}{\var{start}}{\vala_{\var{end}}}, \dbrow{\var{v}}{\vala_{\var{end}}}{\var{end}}} &
            \text{ if } \puretmevaltwo{\subst{\tma_3}{\var{v}}{x}}{\ttrue} \text{
            and } \vala_{\var{start}} \le \var{start} \text{ and }
            \vala_{\var{end}} < \var{end} \quad \textbf{(Case 2)}\\
            \bag{
                \dbrow{\var{v}}{\var{start}}{\vala_{\var{start}}},
                \dbrow{W}{\vala_{\var{start}}}{\vala_{\var{end}}},
                \dbrow{\var{v}}{\vala_{\var{end}}}{\var{end}}} &
            \text{ if } \puretmevaltwo{\subst{\tma_3}{\var{v}}{x}}{\ttrue} \text{
            and } \vala_{\var{start}} > \var{start} \text{ and }
            \vala_{\var{end}} < \var{end} \quad \textbf{(Case 3)}\\
            \bag{\dbrow{\var{v}}{\var{start}}{\vala_{\var{start}}}, \dbrow{W}{\vala_{\var{start}}}{\var{end}}} &
            \text{ if } \puretmevaltwo{\subst{\tma_3}{\var{v}}{x}}{\ttrue} \text{
            and } \vala_{\var{start}} > \var{start} \text{ and }
            \vala_{\var{end}} \ge \var{end} \quad \textbf{(Case 4)}\\
            \bag{\dbrow{\var{v}}{\var{start}}{\var{end}}} & \text{ otherwise }
            \hfill \textbf{(Case 5)}\\
          \end{cases}
        }
        \\
        \text{ where for all cases, } W =
                                \recordwithtwo
                                  {\var{v}}
                                  {\seq{\ell = \valb}} \text{ given } (\puretmevaltwo{\tmb_i}{\valb'_i})_i
    }
    {
        \tmevaltwo
            {\dbupdatebetween{x}{\tmc}{\tma_1}{\tma_2}{\tma_3}{(\ell_i = \tmb_i)_i}}
            { \retpair{\recordterm{}}{\db_2}}
    }
  \end{mathpar}
}
\caption{Reduction rules for \vlinq}
\label{fig:vlinq-reduction-full}
\end{figure*}

\begin{figure}
    {\footnotesize
  \[
    \blt
    \vtrans{\dbdeletenonseqann{x}{\tma}{\tmb}} = \\
    \quad
    {
    \blt
        \dbdelete{x}{\vtrans{\tma}}{\lift{x}{\vtrans{\tmb}}}
    \el
    } \\
    \text{where } \lift{x}{f} \defeq \\
        \quad
        (\fun{x}{f}) \app
            \recordterm{
                \fielddata{\etaexp{x}{\{\ell_i\}_{i}}},
                \fieldstart{\project{x}{\var{start}}},
                \fieldend{\project{x}{\var{end}}}}
    \spacerow \\    %
        \llparenthesis
            \calcwd{update}^{{(\ell_i : \tya_i)}_{i \in I}}\:\calcwd{nonsequenced} \: (x \dbcomparrow \tmc) \:
                  \calcwd{where} \: \tma \:
             \calcwd{set} \: (\ell_j = \tmb_j)_{j \in J} \: \calcwd{valid} \:
            \calcwd{from} \: \tmb'_1 \: \calcwd{to} \: \tmb'_2
        \rrparenthesis
    = \vspace{0.5em}
     \\
    \quad
    {
        \blt
        \calcwd{update}^{{(\ell_i:\tya_i)_{i \in I}}} \: (x \dbcomparrow
        \vtrans{\tmc}) \\
        \quad \calcwd{where}\:(\lift{x}{\vtrans{\tma}}) \\
        \quad \calcwd{set}\:
        {
                    (\recordterm{\ell_j = \lift{x}{\vtrans{\tmb_j}}}_{j \in J},
                        \fieldstart{\lift{x}{\vtrans{\tmb'_1}}},
                        \fieldend{\lift{x}{\vtrans{\tmb'_2}}})
        }
        \\
        \text{where } \lift{x}{f} \defeq \\
        \quad
        (\fun{x}{f}) \app
            \recordterm{
                \fielddata{\etaexp{x}{\{\ell_i\}_{i \in I}}},
                \fieldstart{\project{x}{\var{start}}},
                \fieldend{\project{x}{\var{end}}}}
        \el
    }
    \spacerow\\
   \vtrans{\dbinsertseqann{\tma}{\tmb}} = \\
\quad \letintwo{\var{tbl}}{\vtrans{\tma}} \\
\quad \letinone{\var{rows}} \\
\qquad \forcomptwo{x}{\vtrans{\tmb}} \\
\qqquad \bag{
            \recordplustwo
                {\etaexp{x}{\seq{\ell}}}
                {\recordterm{
                    \fieldstart{\project{x}{\var{start}}},
                    \fieldend{\project{x}{\var{end}}}
                }}
        } \\
\quad \calcwd{in} \\
\quad \dbinsert{\var{tbl}}{\var{rows}}
\spacerow\\
\vtrans{\dbdeletebetweenann{x}{L}{M_1}{M_2}{N}} = \\
\qquad
      {
        \blt
        \letintwo{\var{tbl}}{\vtrans{L}} \\
        \letintwo{\var{aStart}}{\vtrans{M_1}} \\
        \letintwo{\var{aEnd}}{\vtrans{M_2}} \\
        \letintwo{\var{lRows}}{\mkwd{startRows}(\var{tbl}, \mkwd{pred},
        \var{aStart}, \seq{\ell})} \\
        \letintwo{\var{rRows}}{\mkwd{endRows}(\var{tbl}, \mkwd{pred},
        \var{aEnd}, \seq{\ell})}
        \\
        \dbdeletetwo{x}{\var{tbl}} \\
        \quad \calcwhereone{\mkwd{pred} \wedge
        (\project{x}{\field{start}} < \var{aEnd}) \wedge
        (\project{x}{\field{end}} > \var{aStart}}) \\
       \dbinsert{\var{tbl}}{\var{lRows}}; \\
       \dbinsert{\var{tbl}}{\var{rRows}} \\
        \text{where } \mkwd{pred} \defeq \restrict{x}{\seq{\ell}}{\vtrans{\tmb}}
        \el
      }
      \spacerow \\
        \llparenthesis
            \calcwd{update}^{(\ell_i : \tya_i)_{i \in I}} \: \calcwd{sequenced}
            \: (x \dbcomparrow \tmc) \:
            \: \calcwd{between} \: \tma_1 \: \calcwd{and} \: \tma_2 \:
                  \calcwd{where} \: \tma_3 \:
            \: \calcwd{set} \: (\ell_j = \tmb_j)_{j \in J}
        \rrparenthesis = \vspace{0.5em}\\
    \qquad
      {
        \blt
        \letintwo{\var{tbl}}{\vtrans{L}} \\
        \letintwo{\var{aStart}}{\vtrans{M_1}} \\
        \letintwo{\var{aEnd}}{\vtrans{M_2}} \\
        \letintwo{\var{lRows}}{
            \mkwd{startRows}(\var{tbl}, \mkwd{pred}, \var{aStart}, \set{\ell_i}_{i \in I})} \\
        \letintwo{\var{rRows}}{\mkwd{endRows}(\var{tbl}, \mkwd{pred}, \var{aEnd}, \set{\ell_i}_{i \in I})} \\
        \dbupdatetwo{x}{\var{tbl}} \\
        \quad \calcwhereone{\mkwd{pred} \wedge
        (\project{x}{\field{start}} < \var{aEnd}) \wedge
        (\project{x}{\field{end}} > \var{aStart}}) \\
      \quad \calcwd{set} \:
      (
            (\ell_j = \restrict{x}{\{\ell_i\}_{i \in I}}{\vtrans{\tmb_j}})_{j
            \in J},
            \var{start} = \greatesttwo{\project{x}{\var{start}}}{\var{aStart}},
            \var{end} = \leasttwo{\project{x}{\var{end}}}{\var{aEnd}}
        ); \\
       \dbinsert{\var{tbl}}{\var{lRows}}; \\
       \dbinsert{\var{tbl}}{\var{rRows}} \\
        \text{where } \mkwd{pred} \defeq \restrict{x}{\{\ell_i\}_{i \in I}}{\vtrans{M_3}}
        \el
      } \spacerow \\
      \\
      \mkwd{startRows}(\var{tbl}, \var{pred}, \var{aStart}, \seq{\ell}) \defeq \\
      \quad
      \queryzero \\
        \qquad \forcomptwo{x}{\get{\var{tbl}}} \\
        \qqquad \whereone{\mkwd{pred} \wedge (\project{x}{\field{start}} <
            \var{aStart})
          \wedge (\project{x}{\field{end}} > \var{aStart}})
          \\
            \qqquad
            \bag{
                \etaexp
                    {x}
                    {\seq{\ell}}
                \recordplus
                \recordterm{
                    \fieldstart{\project{x}{\var{start}}},
                    \fieldend{\var{aStart}}
                    }
            } \spacerow \\
      \mkwd{endRows}(\var{tbl}, \var{pred}, \var{aEnd}, \seq{\ell}) \defeq \\
        \quad \queryzero \\
        \qquad \forcomptwo{x}{\get{\var{tbl}}} \\
        \qqquad \whereone{\mkwd{pred} \wedge (\project{x}{\field{start}} <
            \var{aEnd})
          \wedge (\project{x}{\field{end}} > \var{aEnd}}) \\
        \qqquad
            \bag{
                \etaexp
                    {x}
                    {\seq{\ell}}
                \recordplus
                \recordterm{
                    \fieldstart{\var{aEnd}},
                    \fieldend{\project{x}{\var{end}}}
                    }
            }
      \el
    \]
}
\caption{Translation from $\vlinq$ into $\linq$}
\label{fig:vlinq-translation-full}
\end{figure}

 \clearpage
\section{Illustration of sequenced delete and update behavior}\label{app:timelines}
Figure~\ref{fig:timelines} illustrates graphically the five different overlap relationships that can hold between the period of validity of an existing row to be updated or deleted (PV) and the period of applicability of a deletion or update.  The five cases involve when one interval is totally contained in the other (1, 3), when there is overlap but neither containment relationship holds (2,4) and when the intervals are disjoint (5).  In each case the effect of a deletion or update is described in the corresponding column of the table; generally the result is to replace the input row with zero, one, two or (exceptionally) three new rows.  When the PV and PA intervals are disjoint no action needs to be taken.

  \begin{figure*}[h]
  \begin{tabular}{|c|c|p{4cm}|p{4cm}|}
    \hline
    Case & Diagram & Delete behavior & Update behavior
    \\\hline
    1 & %
\begin{tikzpicture}[baseline=0pt]
\draw[thick, ->] (0,0) -- (\ImageWidth,0);

\foreach \x in {1,2,3,4,5,6,7}
\draw (\x cm,3pt) -- (\x cm,-3pt);

\draw [ultra thick ,decorate,decoration={brace,amplitude=5pt}] (2,0.25) -- +(4,0)
       node [black,midway,above=4pt, font=\footnotesize] {PV};
\draw [ultra thick,decorate,decoration={brace,amplitude=5pt}] (7,-0.25) --
+(-6,0)
       node [black,midway,font=\footnotesize, below=4pt] {PA};

\end{tikzpicture}
       & the entire row will be deleted
      & the entire row will be updated and adjusted to cover PA
      \\\hline
    2 & %
\begin{tikzpicture}[baseline=0pt]
\draw[thick, ->] (0,0) -- (\ImageWidth,0);

\foreach \x in {1,2,3,4,5,6,7}
\draw (\x cm,3pt) -- (\x cm,-3pt);

\draw [ultra thick ,decorate,decoration={brace,amplitude=5pt}] (2,0.25) -- +(4,0)
       node [black,midway,above=4pt, font=\footnotesize] {PV};
\draw [ultra thick,decorate,decoration={brace,amplitude=5pt}] (3,-0.25) -- +(-2,0)
       node [black,midway,font=\footnotesize, below=4pt] {PA};

\end{tikzpicture}
       & the overlapping portion will be deleted, leaving the rest of PV alone
      & insert new row with updated values covering PA and shorten existing row to only cover remainder of PV
      \\\hline
    3 & %
\begin{tikzpicture}[baseline=0pt]
\draw[thick, ->] (0,0) -- (\ImageWidth,0);

\foreach \x in {1,2,3,4,5,6,7}
\draw (\x cm,3pt) -- (\x cm,-3pt);

\draw [ultra thick ,decorate,decoration={brace,amplitude=5pt}] (2,0.25) -- +(4,0)
       node [black,midway,above=4pt, font=\footnotesize] {PV};
\draw [ultra thick,decorate,decoration={brace,amplitude=5pt}] (5,-0.25) -- +(-2,0)
       node [black,midway,font=\footnotesize, below=4pt] {PA};

\end{tikzpicture}
       & split row into two, one covering the part of PV before PA and the other covering the part of PV after PA
                                     & split row into three, two as in the case of deletion and a third providing the updated values for PA
                                       
    \\\hline
     4 & %
\begin{tikzpicture}[baseline=4pt]
\draw[thick, ->] (0,0) -- (\ImageWidth,0);

\foreach \x in {1,2,3,4,5,6,7}
\draw (\x cm,3pt) -- (\x cm,-3pt);

\draw [ultra thick ,decorate,decoration={brace,amplitude=5pt}] (2,0.25) -- +(4,0)
       node [black,midway,above=4pt, font=\footnotesize] {PV};
\draw [ultra thick,decorate,decoration={brace,amplitude=5pt}] (7,-0.25) --
    +(-2,0)
       node [black,midway,font=\footnotesize, below=4pt] {PA};

\end{tikzpicture}
       & symmetric with case 2
                                     & symmetric with case 2
                                       
    \\\hline
    5 & %
\begin{tikzpicture}[baseline=5pt]
\draw[thick, ->] (0,0) -- (\ImageWidth,0);

\foreach \x in {1,2,3,4,5,6,7}
\draw (\x cm,3pt) -- (\x cm,-3pt);

\draw [ultra thick ,decorate,decoration={brace,amplitude=5pt}] (2,0.25) -- +(4,0)
       node [black,midway,above=4pt, font=\footnotesize] {PV};
\draw [ultra thick,decorate,decoration={brace,amplitude=5pt}] (1.8,-0.25) -- +(-1.8,0)
       node [black,midway,font=\footnotesize, below=4pt] {PA};
\draw [ultra thick,decorate,decoration={brace,amplitude=5pt}] (8,-0.25) -- +(-1.8,0)
       node [black,midway,font=\footnotesize, below=4pt] {PA};

\end{tikzpicture}
       & no effect
                                     & no effect
                                       
    \\\hline
  \end{tabular}
  \caption{Timeline diagrams illustrating sequenced delete and update behavior}\label{fig:timelines}
\end{figure*}
 \clearpage
\section{Direct semantics and translations for $\vlinq$ current updates and
deletions}\label{sec:vlinq-current-updates}

In Section~\ref{sec:vlinq}, we showed that current updates and deletions in
$\vlinq$ are special cases of sequenced updates and deletions. For the sake
of completeness, we include the direct semantics and translations here.

{\footnotesize
\begin{mathpar}
  \inferrule
  [EV-Update]
  {
      \vtmevaltwo{\tmc}{\retpair{\tblvar}{\db'}} \\
    \db'' = \db'[ \tblvar \mapsto \flattenbag{\bag{\mkwd{upd}(v) \midspace v \in
    \db'(\tblvar)}}] \\
    \\
    {
      \mkwd{upd}(\dbrow{\var{data}}{\var{start}}{\var{end}}) =
      {
        \begin{cases}
           \bag
            {\dbrow
                {\recordwithtwo
                    {\var{data}}
                    {\seq{\ell = \vala}}
                }
                {\var{start}}
                {\var{end}}
            }
            &
            {
              \begin{array}[t]{l}
                \text{if } \puretmevaltwo{\tma \{ \var{data} / x \}}{\ttrue},
                \timevar \le \var{start},
                \text{ and }
                (\puretmevaltwo{\subst{\tmb_i}{\var{data}}{x}}{\vala_i})_i,
              \end{array}
            }  \\
            \bag{
                \dbrow{\var{data}}{\var{start}}{\timevar},
                \dbrow
                    {\recordwithtwo{\var{data}}{\seq{\ell = \vala}}}
                    {\timevar}
                    {\var{end})}
            }
            &
            {
              \begin{array}[t]{l}
                \text{if } \puretmevaltwo{\tma \{ \var{data} / x \}}{\ttrue},
                \var{start} < \timevar < \var{end},
                \text{ and }
                (\puretmevaltwo{\subst{\tmb_i}{\var{data}}{x}}{\vala_i})_i,
              \end{array}
            } \\
            \bag{\dbrow{\var{data}}{\var{start}}{\var{end}}} &
            {
              \begin{array}[t]{l}
                \text{if } \puretmevaltwo{M \{ \var{data} / x \}}{\ffalse}
                \text{ or }
                \timevar \ge \var{end}
              \end{array}
            }
        \end{cases}
      }
    }
  }
  { \vtmevaltwo
      {\dbupdate{x}{\tmc}{\tma}{(\seq{\ell = \tmb})}}
      {\retpair{\db''}{()}}
  }

  \inferrule
  [EV-Delete]
  {
      \vtmevaltwo{\tma}{\retpair{\tblvar}{\db'}} \\
    \db'' = \db'[ \tblvar \mapsto \flattenbag{\bag{\mkwd{del}(v) \midspace v \in
    \db'(\tblvar)}}] \\\\
    {
      \mkwd{del}(\dbrow{\var{data}}{\var{start}}{\var{end}}) =
      {
        \begin{cases}
           \bag
            {~}
            &
            {
              \begin{array}[t]{l}
                \text{if } \puretmevaltwo{\tmb \{ \var{data} / x \}}{\ttrue}
                \text{ and }
                \timevar \le \var{start}
              \end{array}
            }  \\
            \bag{
                \dbrow{\var{data}}{\var{start}}{\timevar}
            }
            &
            {
              \begin{array}[t]{l}
                \text{if } \puretmevaltwo{\tmb \{ \var{data} / x \}}{\ttrue}
                \text{ and }
                \var{start} < \timevar < \var{end}
              \end{array}
            } \\
            \bag{\dbrow{\var{data}}{\var{start}}{\var{end}}} &
            {
              \begin{array}[t]{l}
                \text{if } \puretmevaltwo{\tmb \{ \var{data} / x \}}{\ffalse}
                \text{ or }
                \timevar \ge \var{end}
              \end{array}
            }
        \end{cases}
      }
    }
  }
  { \vtmevaltwo
      {\dbdelete{x}{\tma}{\tmb}}
      {\retpair{\db''}{()}}
  }
\end{mathpar}
}

Rules \textsc{EV-Delete} and \textsc{EV-Update} are different to their $\tlinq$
counterparts, taking into account that a row may be valid further into the
future than the current time.

Specifically, rule \textsc{EV-Delete} has three cases: if the row being
inspected is entirely in the future (i.e., its $\var{start}$ field is greater
than or equal to the current time), then it is deleted from the database. If the
current time is greater than the start of the row, but less than the end of the
row, then the row is closed off at the current time. The row is not
modified if the current time is greater than the end time of the row, or the
predicate evaluates to $\ffalse$.
Current updates, described by \textsc{EV-Update}, follow much the same pattern,
but if the current time is between the start and end times of the row,
then the row is split: the the part up until the current time retains
the previous values, and the part from the current time until the end of the row
is set to the new values.

\[
        \bl
        \mkwd{currentAt}(x, \var{time}) \defeq
        \var{time} \ge \project{x}{\var{start}} \wedge
        \var{time} < \project{x}{\var{end}}
        \el
    \]

{\footnotesize
    \begin{mathpar}
        {\bl
\vtrans{\dbdeleteann{x}{M}{N}} = \\
    \quad
    {
      \blt
        \letintwo{\var{tbl}}{\vtrans{M}} \\
        \letintwo{\var{time}}{\now} \\
        \dbupdatetwo{x}{\var{tbl}} \\
        \quad \calcwd{where} \: {(\restrict{x}{\set{\ell_i}_i}{\vtrans{\tmb}} \wedge
            \project{x}{\var{start}} \le \var{time} \wedge
        \project{x}{\var{end}} > \var{time})} \\
        \quad \calcwd{set} \: (\fieldend{\var{time}}); \\
        \dbdeletetwo{x}{\var{tbl}} \\
        \quad \calcwd{where} \: {(\restrict{x}{\{\ell_i\}_i}{\vtrans{\tmb}} \wedge
        \project{x}{\var{start}} \ge \var{time})}
      \el
    }
    \el}

    {\bl
\vtrans{
    \dbupdateann
        [(\ell_i : \tya_i)_{i \in I}]
        {x}
        {\tmc}
        {\tma}
        {(\ell_j = \tmb_j)_{j\in J}}} = \\
    {
        \quad
      \blt
      \letintwo{\var{tbl}}{\vtrans{\tmc}} \\
      \letintwo{\var{time}}{\now} \\
      \letinone{\var{affected}} \\
      \quad \queryzero \\
      \qquad \forcomptwo{x}{\var{tbl}} \\
      \qqquad
        \calcwhereone
            {\restrict{x}{\{ \ell_i\}_{i \in I}}{\vtrans{M}} \wedge
              \currentat{x}{\var{time}}} \\
      \qqquad
      \bag{
            {\bl
                \recordterm{\ell_i = \project{x}{\ell_i}}_{i \in I
                \without J} \, \recordplus \\
                \recordterm
                    {\ell_j =
                        \restrict
                            {x}
                            {\{\ell_i\}_{i \in I}}
                            {\vtrans{\tmb_j}}}_{j \in J} \, \recordplus \\
                \recordterm{
                    \fieldstart{\var{time}},
                    \fieldend{\project{x}{\var{end}}}
                }
                \el
            }
      } \\
      \calcwd{in} \\
      \dbupdatetwo{x}{\var{tbl}} \\
      \quad \calcwd{where} \: \left( {
          \bl
            \restrict{x}{\{\ell_i\}_{i \in I}}{\vtrans{\tma}} \wedge
            (\project{x}{\field{start}} < \var{time}) \; \wedge \\
            (\project{x}{\var{end}} > \var{time}) \\
          \el
          } \right)
      \\
      \quad \calcwd{set} \: (\fieldend{\var{time}}); \\
      \dbupdatetwo{x}{\var{tbl}} \\
      \quad \calcwd{where} \: (\restrict{x}{\{ \ell_i\}_{i \in I}}{\vtrans{M}} \wedge \project{x}{\var{start}} \ge \var{time}) \\
      \quad \calcwd{set} \: \recordterm{(\ell_j = \restrict{x}{\{ \ell_i \}_{i
          \in I}}{\vtrans{\tmb_j}})_{j \in J}}; \\
      \dbinsert{\var{tbl}}{\var{affected}} \\
      \el
    }
    \el
}
\end{mathpar}
}

The meta-level $\mkwd{currentAt}$ function checks whether a row is valid at a
particular timestamp.

Current deletions are implemented using two $\linq$ operations:
an $\calcwd{update}$ operation, which adjusts the $\var{end}$ time of
all matching, currently-valid rows to the current time; a
$\calcwd{delete}$ operation, which entirely deletes all matching rows
which begin after the current time.

To implement current updates, we firstly calculate the affected bag of rows
(i.e., those rows which match the predicate and are valid at the current time),
with their updated values. We then do three operations:

\begin{itemize}
    \item An $\calcwd{update}$, which adjusts the $\var{end}$ time of all
        records which match the predicate, start before the current time, and
        end after the current time
    \item A second $\calcwd{update}$, which updates all rows matching the
        predicate which both start and end after the current timestamp
    \item Finally, an $\calcwd{insert}$, which inserts the contents of the
        $\var{affected}$ bag.
\end{itemize}

 \clearpage
\section{Proofs}\label{appendix:proofs}

\begin{lemma}[Pure comprehensions]\label{lem:pure-for}
    If $\vala = \bag{\vala_1, \ldots, \vala_n}$
    and we have
    a for comprehension $\forcomp{x}{\vala}{\tma}$ such that
    $\tseq{\cdot}{\forcomp{x}{\vala}{\tma}}{\tya}{\pure}$,
    and $\puretmevaltwo{\subst{\tma}{\valb_i}{x}}{\valb'_i}$ or
    $\readtmevaltwo{\subst{\tma}{\valb_i}{x}}{\valb'_i}$
    for each $i \in 1..n$,
    then
    $\puretmevaltwo
        {\forcomp{x}{\vala}{\tma}}
        {\bag{\valb'_1, \ldots, \valb'_n}}$.
\end{lemma}
\begin{proof}
    Follows from the comprehension evaluation rules.
\end{proof}

\begin{lemma}[Lifting of pure reduction]\label{lem:pure-lift}
    If $\tseq{\cdot}{\tma}{\tya}{\pure}$ and $\puretmevaltwo{\tma}{\vala}$ or
    $\readtmevaltwo{\tma}{\vala}$, then
    $\tmevaltwo{\tma}{\retpair{\vala}{\db}}$ for any $\db$.
\end{lemma}
\begin{proof}
    By induction on the derivation of $\puretmevaltwo{\tma}{\vala}$.
\end{proof}

\begin{lemma}[Correctness of $\tlinq$ translation (pure reduction)]\label{lem:tlinq-trans-pure}
    If $\puretmevaltwo{\tma}{\vala}$, then
    $\puretmevaltwo{\ttrans{\tma}}{\ttrans{\vala}}$.
\end{lemma}
\begin{proof}
    By induction on the derivation of $\puretmevaltwo{\tma}{\vala}$.
\end{proof}

\begin{lemma}\label{lem:tlinq-trans-restrict}
    If
    $\tseq
        {\tyenv, x : \recordty{\ell_i : \tya_i}_{i \in I}}
        {\tma}
        {\tyb}
        {\pure}$ and
    $\puretmevaltwo{\tma}{\vala}$,
    and $
        \tseq
            {\ttrans{\tyenv}, x : (\recordty{\ell_i : \ttrans{\tya_i}}_{i \in I})
                \recordplus \recordty{\ell_j : \tya'_j}_{j \in J}}
            {\ttrans{\tma}}
            {\ttrans{\tya}}
            {\pure}$,
    then
    $\puretmevaltwo[\flatten{\db'}, \timevar]{\restrict{x}{(\ell_i)_{i \in I}}{\ttrans{\tma}}}{\ttrans{\vala}}$.
\end{lemma}
\begin{proof}
    Follows from the definition of \mkwd{restrict} and Lemma~\ref{lem:tlinq-trans-pure}.
\end{proof}

\begin{lemma}[Correctness of $\vlinq$ translation (pure reduction)]\label{lem:vlinq-trans-pure}
    If $\puretmevaltwo{\tma}{\vala}$, then
    $\puretmevaltwo{\vtrans{\tma}}{\vtrans{\vala}}$.
\end{lemma}
\begin{proof}
    By induction on the derivation of $\puretmevaltwo{\tma}{\vala}$.
\end{proof}

\begin{lemma}\label{lem:vlinq-trans-restrict}
    If
    $\tseq
        {\tyenv, x : \recordty{\ell_i : \tya_i}_{i \in I}}
        {\tma}
        {\tyb}
        {\pure}$ and
    $\puretmevaltwo{\tma}{\vala}$,
    and $
        \tseq
            {\vtrans{\tyenv}, x : (\recordty{\ell_i : \vtrans{\tya_i}}_{i \in I})
                \recordplus \recordty{\ell_j : \tya'_j}_{j \in J}}
            {\vtrans{\tma}}
            {\vtrans{\tya}}
            {\pure}$,
    then
    $\puretmevaltwo[\flatten{\db'}, \timevar]{\restrict{x}{(\ell_i)_{i \in I}}{\vtrans{\tma}}}{\vtrans{\vala}}$.
\end{lemma}
\begin{proof}
    Follows from the definition of \mkwd{restrict} and Lemma~\ref{lem:vlinq-trans-pure}.
\end{proof}

\begin{lemma}\label{lem:vlinq-base-trans}
    If $\tseq{\tyenv}{\vala}{\tya}{\pure}$ where $\tya = \basety$ or $\tya =
    \recordty{\ell_i = \basety_i}_i$, then $\vtrans{\vala} = \vala$.
\end{lemma}
\begin{proof}
    By induction on the derivation of $\tseq{\tyenv}{\vala}{\tya}{\pure}$.
\end{proof}

\begin{proposition}[Type correctness of $\tlinq$ translation]
    \label{prop:tlinq-type-correct}
    If $\tseq{\tyenv}{\tma}{\tya}{\effs}$ in $\tlinq$, then
    $\tseq{\ttrans{\tyenv}}{\ttrans{\tma}}{\ttrans{\tya}}{\effs}$ in $\linq$.
\end{proposition}
\begin{proof}
    By induction on the derivation of $\tseq{\tyenv}{\tma}{\tya}{\effs}$.
\end{proof}

\begin{proposition}[Type correctness of $\vlinq$ translation]
    \label{prop:vlinq-type-correct}
    If $\tseq{\tyenv}{\tma}{\tya}{\effs}$ in $\vlinq$, then
    $\tseq{\vtrans{\tyenv}}{\vtrans{\tma}}{\vtrans{\tya}}{\effs}$ in $\linq$.
\end{proposition}
\begin{proof}
    By induction on the derivation of $\tseq{\tyenv}{\tma}{\tya}{\effs}$.
\end{proof}

\ttranscorrect*
\begin{proof}
    By induction on the derivation of $\ttmevaltwo{\tma}{\vala}$.
    We show the cases for the database-relevant terms.

    \begin{proofcase}{Accessor functions}
        Let us show the case for $\calcwd{data}$; the others are similar.

        Assumption:
        \begin{mathpar}
            \inferrule
            { \ttmevaltwo
                {\tma}
                {\retpair{\dbrow{\vala_1}{\vala_2}{\vala_3}}{\db'}}
            }
            { \ttmevaltwo
                {\data{\tma}}
                {\retpair{\vala_1}{\db'}}
            }
        \end{mathpar}

        Translation:
        \[
            \project{\ttrans{\tma}}{\var{data}}
        \]

        By the IH:
            $\tmevaltwo
                [\flatten{\db}, \timevar]
                {\ttrans{\tma}}
                { \recordterm
                    {\fielddata{\ttrans{\vala_1}},
                     \fieldstart{\ttrans{\vala_2}},
                     \fieldend{\ttrans{\vala_3}}
                    }
                }
            $

        Evaluating in $\linq$:

        \[
            \tmevaltwo
            [\flatten{\db}, \timevar]
            {\project
                {\recordterm{
                    \fielddata{\ttrans{\vala_1}},
                    \fieldstart{\ttrans{\vala_2}},
                    \fieldend{\ttrans{\vala_3}}}}
                {\var{data}}
            }
            {
                \retpair{\ttrans{\vala_1}}{\flatten{\db'}}
            }
        \]

        as required.
    \end{proofcase}

    \begin{proofcase}{get}
        Assumption:

        \begin{mathpar}
            \inferrule
            { \ttmevaltwo{\tma}{\retpair{\tblvar}{\db'}} }
            { \ttmevaltwo{\get{\tma}}{\retpair{\db'(\tblvar)}{\db'}} }
        \end{mathpar}

        Translation:
        \[
            \blt
           \forcomp{x}{\get{\tvaltrans{\tma}}} \\
            \quad \bag{\recordterm{
                \fielddata{(\ell_i = \project{x}{\ell_i})_i},
                \fieldstart{\project{x}{\field{start}}},
                \fieldend{\project{x}{\field{end}}}}}
                \el
        \]

        Suppose $\db(\tblvar) =
        \bag{
            \dbrow{
                \recordterm{\ell_i = \vala_{1_i}}_{i \in 1..m}
            }{\valb_{1_{s}}}{\valb_{1_{e}}},
            \ldots,
            \dbrow
            {
                \recordterm{\ell_i = \vala_{n_i}}_{i \in 1..m}
            }
            {\valb_{n_{s}}}
            {\valb_{n_{e}}}}$.

        By the IH,
            $\tmevaltwo
                [\flatten{\db}, \timevar]
                {\ttrans{\tma}}
                {\retpair{\tblvar}{\flatten{\db'}}}$

        Then,
        \[
            \blt
            \flatten{\db'}(\tablevar) = \\
            \quad
        \lbag \recordterm
            {\ell_1 = \ttrans{\vala_{1_1}},
             \ldots,
             \ell_m = \ttrans{\vala_{1_m}},
             \fieldstart{\ttrans{\valb_{1_{s}}}},
             \fieldend{\ttrans{\valb_{1_{e}}}}}, \\
        \qquad \ldots, \\
        \quad \recordterm
            {\ell_1 = \ttrans{\vala_{n_1}},
             \ldots,
             \ell_m = \ttrans{\vala_{n_m}},
             \fieldstart{\ttrans{\valb_{n_{s}}}},
             \fieldend{\ttrans{\valb_{n_{e}}}}} \rbag
             \el
        \]

        In $\linq$, we have:
        \[
            \tmevaltwo[\flatten{\db'}, \timevar]{
                \get{\tablevar}
            }
            {
                \bag{\recordterm
                    {\ell_1 = \ttrans{\vala_{1_1}},
                     \ldots,
                     \ell_m = \ttrans{\vala_{1_m}},
                     \fieldstart{\ttrans{\valb_{1_{s}}}},
                     \fieldend{\ttrans{\valb_{1_{e}}}}},
                \ldots,
                \recordterm
                    {\ell_1 = \ttrans{\vala_{n_1}},
                     \ldots,
                     \ell_m = \ttrans{\vala_{n_m}},
                     \fieldstart{\ttrans{\valb_{n_{s}}}},
                     \fieldend{\ttrans{\valb_{n_{e}}}}}}
            }
        \]

    By Lemmas~\ref{lem:pure-for} and~\ref{lem:pure-lift}
    {\small
    \[
        \tmevaltwo{
            \bl
            \forcomp{x}{\get{\tvaltrans{\tblvar}}} \\
                \quad \bag{\recordterm{
                    \fielddata{(\ell_i = \project{x}{\ell_i})_i},
                    \fieldstart{\project{x}{\field{start}}},
                    \fieldend{\project{x}{\field{end}}}}}
            \el
        }
        {
            \retpair{
                \blt
                \lbag \recordterm
                    {\fielddata{
                        \recordterm{
                            \ell_1 = \ttrans{\vala_{1_1}},
                            \ldots,
                            \ell_m = \ttrans{\vala_{1_m}}
                        }
                     },
                     \fieldstart{\ttrans{\valb_{1_{s}}}},
                     \fieldend{\ttrans{\valb_{1_{e}}}}}, \\
                \ldots,
                \\
                    \recordterm
                    {\fielddata{
                        \recordterm{
                            \ell_1 = \ttrans{\vala_{n_1}},
                            \ldots,
                            \ell_m = \ttrans{\vala_{n_m}}
                        }},
                     \fieldstart{\ttrans{\valb_{n_{s}}}},
                     \fieldend{\ttrans{\valb_{n_{e}}}}}
            \rbag
            }
            {
                \flatten{\db'}
            }
        }
        \el
    \]
    }

    as required.
    \end{proofcase}

    \begin{proofcase}{insert}
        \begin{mathpar}
            \inferrule
            {
                \ttmevaltwo{\tma}{\retpair{\tblvar}{\db'}} \\\\
                \puretmevaltwo{\tmb}{\vala} \\
                \mathit{vs} = \bag{\dbrow{v}{\timevar}{\forever} \mid v \in \vala} \\\\
                \db'' = \extendenv{\db'}{\tblvar}{\db'(\tblvar) \bagunion \mathit{vs}}
            }
            {
              \ttmevaltwo{\dbinsert{\tma}{\tmb}}{\retpair{()}{\db''}}
            }
        \end{mathpar}

        Translation:
        \[
            \ttmtrans{\dbinsertann{\tma}{\tmb}} = \\
              \quad
                {
                \blt
                  \effletone{\var{rows}} \\
                  \quad \forcomptwo{x}{\ttmtrans{\tmb}} \\
                  \qquad
                    \bag{
                        \recordplustwo
                            {\etaexp{x}{\seq{\ell}}}
                            {\recordterm{
                                \fieldstart{\now},
                                \fieldend{\forever}
                            }}
                    } \\
                    \calcwd{in} \\
                  \dbinsert{\ttmtrans{\tma}}{\var{rows}}
                \el
                }
    \]

        By the typing rules and evaluation rule, $\vala$ must be some bag
        $\bag{\recordterm{\ell_{1_i} = \valb_{1_i}}_i, \ldots,
            \recordterm{\ell_{n_i} = \valb_{n_i}}_i}$

        By Lemma~\ref{lem:tlinq-trans-pure}
        $\puretmevaltwo{\ttrans{\tmb}}{\ttrans{\vala}}$ with
        $\ttrans{\vala} = \bag{\recordterm{\ell_{1_i} = \ttrans{\valb_{1_i}}}_i, \ldots,
            \recordterm{\ell_{n_i} = \ttrans{\valb_{n_i}}}_i}$.

        By Lemmas~\ref{lem:pure-for} and~\ref{lem:tlinq-trans-pure}, the
        $\calcwd{for}$ comprehension evaluates to
        $\ttrans{\vala} =
        \bag{\recordterm{\ell_{1_i} = \ttrans{\valb_{1_i}},
            \fieldstart{\timevar}, \fieldend{\forever}}_i, \ldots,
            \recordterm{\ell_{n_i} = \ttrans{\valb_{n_i}},
            \fieldstart{\timevar}, \fieldend{\forever}}_i}$. Let us call this
            value
            $\var{rows}$.

        By the IH,
        $\tmevaltwo
            [\flatten{\db}, \timevar]
            {\ttrans{\tma}}{\retpair{\tblvar}{\flatten{\db'}}}$.

        Thus, we evaluate
        \[
            \tmevaltwo
                [\flatten{\db'}, \timevar]
                {\dbinsert
                    {\tblvar}
                    {\var{rows}}
                }
                {
                    \retpair
                        {()}
                        {\extendenv{\flatten{\db'}}{\tblvar}{
                                \flatten{\db'} \denotbagunion \var{rows}}}
                }
        \]

        noting that $\ttrans{\db''} = \flatten{\db'} \denotbagunion \var{rows}$,
        as required.
    \end{proofcase}

    \begin{proofcase}{delete}

        Assumption:
        \begin{mathpar}
            \inferrule
            {
                \ttmevaltwo{\tma}{\retpair{\tblvar}{\db'}} \\
                \db'' =
                    {\extendenv
                        {\db'}
                        {\tblvar}
                        { \bag{\mkwd{del}(d) \midspace d \in \db'(\tblvar)} }
                    }
                \\\\
                \mkwd{del}(\dbrow{\var{data}}{\var{start}}{\var{end}}) =
                {
                    \begin{cases}
                        \dbrow{\var{data}}{\var{start}}{\timevar}
                         & \text{if } \var{end} = \forever \text{
                            and } \puretmevaltwo{\subst{\tmb}{\var{data}}{x}}{\ttrue}
                            \\
                        \dbrow{\var{data}}{\var{start}}{\var{end}} & \text{otherwise}
                    \end{cases}
                }
            }
            {
                \ttmevaltwo
                    {\dbdelete{x}{\tma}{\tmb}}
                    {\retpair
                        {()}
                        {\db''}}
            }
        \end{mathpar}

        Translation:
        \[
            \bl
             \ttmtrans{\dbdelete{x}{\tma}{\tmb}} = \\
                \quad {
                  \blt
                  \dbupdatetwo{x}{\tvaltrans{\tma}} \\
                  \quad \calcwhereone{((\fun{x}{\ttmtrans{\tmb}}) \app \recordterm{\ell_i = \project{x}{\ell_i}}_i) \wedge \project{x}{\var{end}} = \forever)} \\
                  \quad \calcwd{set} \: (\fieldend{\now})
                  \el
                }
            \el
        \]

        By the IH,
            $\tmevaltwo
                [\flatten{\db}, \timevar]
                {\ttrans{\tma}}
                {\retpair{\tblvar}{\flatten{\db'}}}$

        Suppose $\db'(\tblvar) =
        \bag{
            \dbrow{
                \recordterm{\ell_i = \vala_{1_i}}_{i \in 1..m}
            }{\valb_{1_{s}}}{\valb_{1_{e}}},
            \ldots,
            \dbrow
            {
                \recordterm{\ell_i = \vala_{n_i}}_{i \in 1..m}
            }
            {\valb_{n_{s}}}
            {\valb_{n_{e}}}}$.

        Then,
        \[
            \blt
            \flatten{\db'}(\tablevar) = \\
            \quad
        \lbag \recordterm
            {\ell_1 = \ttrans{\vala_{1_1}},
             \ldots,
             \ell_m = \ttrans{\vala_{1_m}},
             \fieldstart{\ttrans{\valb_{1_{s}}}},
             \fieldend{\ttrans{\valb_{1_{e}}}}}, \\
        \qquad \ldots, \\
        \quad \recordterm
            {\ell_1 = \ttrans{\vala_{n_1}},
             \ldots,
             \ell_m = \ttrans{\vala_{n_m}},
             \fieldstart{\ttrans{\valb_{n_{s}}}},
             \fieldend{\ttrans{\valb_{n_{e}}}}} \rbag
             \el
        \]

        Recall the definition of $\calcwd{update}$ in $\linq$:

        \begin{mathpar}
            \inferrule
            {
                \tmevaltwo{\tmc}{\retpair{\tblvar}{\db_1}} \\
                        \db_2 =
                            {\extendenv
                                {\db_1}
                                {\tblvar}
                                { \bag{\mkwd{upd}(v) \mid v \in \db_1(\tblvar)}}
                            }
                        \\\\
                        {
                        \mkwd{upd}(v) =
                        {
                            \begin{cases}
                                \recordwithtwo
                                    {v}
                                    {\seq{\ell = \valb}}
                                    & \text{ if }
                                    \puretmevaltwo{\tma \{ v / x \}}{\ttrue}
                                    \text{ and }
                                    (\puretmevaltwo{\subst{\tmb_i}{v}{x}}{\valb_i})_i \\
                                v & \text{ if } \puretmevaltwo{\tma \{ v / x \}}{\ffalse}
                            \end{cases}
                        }
                        }
            }
            {
                \tmevaltwo
                    {\dbupdate{x}{\tmc}{\tma}{(\seq{\ell = \tmb})}}
                    {\retpair{()}{\db_2}}
            }
        \end{mathpar}

        Let us reason on each row individually.
        Say we have a row
        $\dbrow{\recordterm{\ell_i = \vala_i}}{\valb_{\var{start}}}{\valb_{\var{end}}}$
        in $\db'$.
        Thus, in $\flatten{\db'}$
        that row will be
        $\recordterm{\ell_i = \ttrans{\vala_i}}_i \recordplus
        \recordterm{
             \fieldstart{\ttrans{\valb_{\var{start}}}},
             \fieldend{\ttrans{\valb_{\var{end}}}}}$; let us call this $x$.

        Now note that
        $\puretmevaltwo
            {\recordterm{\ell_i = \project{x}{\ell_i}}_i}
            {\recordterm{\ell_i = \ttrans{\vala_i}}_i}
        $
        which is equal to the translation of the data component of the record in
        the transaction time database.

        Therefore by Lemma~\ref{lem:tlinq-trans-pure}, we have that
        $\puretmevaltwo{\ttrans{\tmb} \{ \recordterm{\ell_i =
        \ttrans{\vala_i}}_i / x \}}{\ttrue}$
        whenever
        $\puretmevaltwo{\tmb \{ \var{data} / x \}}{\ttrue}$.
        Since $\ttrans{\forever} = \forever$ we have that $\var{end} =
        \forever$ if $\project{x}{\var{end}} = \forever$. Thus the
        $\calcwd{update}$ affects the same rows.

        We have that $\puretmevaltwo{\now}{\timevar}$ and thus
        $\recordwith{\recordterm{\ell_i = \ttrans{\vala_i}}_i
            \recordplus
            \recordterm{
                 \fieldstart{\ttrans{\valb_{\var{start}}}},
                 \fieldend{\ttrans{\valb_{\var{end}}}}}
        }{\var{end}}{\timevar}
        =
        \flatten{\dbrow{\recordterm{\ell_i = \ttrans{\vala_i}}_i}{\valb_{\var{start}}}{\timevar}}
        $
        for each affected record, as required.
    \end{proofcase}

    \begin{proofcase}{update}
        Assumption:

        \begin{mathpar}
            \inferrule
            {
               \ttmevaltwo{\tmc}{\retpair{\tblvar}{\db'}} \\
               \db'' = \extendenv{\db'}{t}{\flatten{\bag{\mkwd{u}(v) \mid v \in \db'(t)}}} \\
               \mkwd{u}(\dbrow{\mathit{data}}{\mathit{start}}{\mathit{end}}) =
                    {
                    \begin{cases}
                        {\blt
                           \bag{\dbrow
                               {\var{data}}
                               {\var{start}}
                               {\timevar},
                               \dbrow
                               {
                                   \recordwithtwo
                                   {\mathit{data}}
                                   {\seq{\ell = \vala'}}
                               }
                               {\timevar}
                               {\forever}
                           }
                           \el}
                       &
                               {
                                 \blt
                                     \text{if }
                                     \puretmevaltwo{\subst{\tma}{\var{data}}{x}}{\ttrue}, \\
                                     \mathit{end} = \forever, \\
                                     \quad (\puretmevaltwo{\subst{\tmb_i}{\var{data}}{x}}{\vala'_i})_i
                                 \el
                               }
                      \\
                      {\blt
                          \bag{\dbrow{\mathit{data}}{\mathit{start}}{\mathit{end}}}
                      \el}
                       & { \blt \text{otherwise} \el }
                   \end{cases}
                   }
            }
            { \ttmevaltwo
                {\dbupdate{x}{\tmc}{\tma}{(\seq{\ell = \tmb})}}
                {\retpair{()}{\db''}}
            }
        \end{mathpar}

        Translation:
  \[
    \bl
  \ttmtrans{
    \dbupdateann
        [\recordterm{\ell_i : \tya_i}_{i \in I}]
        {x}
        {\tmc}
        {\tma}
        {(\ell = \tmb_j)_{j \in J}}} = \\
  \quad
  {
    \blt
      \letintwo{\var{tbl}}{\ttrans{\tmc}} \\
      \letinone{\var{affected}}\\
      \quad \forcomptwo{x}{\get{\var{tbl}}} \\
      \qquad
        \whereone
            {(\restrict{x}{\seq{\ell}}{\ttmtrans{\tma}} \wedge \iscurrent{x})} \\
      \qquad \bag{
          \recordtermlr{
              {\bl
                  \fielddata{
                        \recordterm{\ell_i = \project{x}{\ell_i}}_{i \in I \without J}
                        \;\recordplus \\
                        \quad
                        \recordterm
                            {\ell_j =
                                \restrict
                                    {x}
                                    {(\ell_i)_{i \in I}}
                                    {\ttrans{\tmb_j}}}
                  }, \\
                  \fieldstart{\now},
                  \fieldend{\forever}
                  \el
              }
           }
        }
      \\
      \calcwd{in}
      \\
      \dbupdatetwo{x}{\var{tbl}} \\
      \quad \calcwd{where} \: (\restrict{x}{\seq{\ell}}{\ttmtrans{\tma}} \wedge \iscurrent{x}) \\
      \quad \calcwd{set} \: \recordterm{\fieldend{\now}}; \\
      \dbinsert{\var{tbl}}{\var{affected}}
    \el
  }
  \el
  \]

        By the IH,
        $\tmevaltwo{\ttrans{\tmc}}{\retpair{\tblvar}{\flatten{\db'}}}$.

        Suppose
        \[
            \bl
            \db(\tblvar') =
            \bag{
                \dbrow{
                    \recordterm{\ell_i = \vala_{1_i}}_{i \in 1..l}
                }{\valb_{1_{s}}}{\valb_{1_{e}}},
                \ldots,
                \dbrow{
                    \recordterm{\ell_i = \vala_{{m}_i}}_{i \in 1..l}
                }{\valb_{m_{s}}}{\valb_{m_{e}}}, \\
                \dbrow{
                    \recordterm{\ell_i = \vala_{{m + 1}_i}}_{i \in 1..l}
                }{\valb_{{m + 1}_{s}}}{\valb_{{m + 1}_{e}}},
                \ldots
                \dbrow
                {
                    \recordterm{\ell_i = \vala_{n_i}}_{i \in 1..l}
                }
                {\valb_{n_{s}}}
                {\valb_{n_{e}}}}
            \el
        \]
    Then
        \[
            \bl
            \flatten{\db'}(\tblvar) =
            \bag{
                \recordterm{\ell_i = \ttrans{\vala_{1_i}}}_{i \in 1..l}
                \recordplus
                \recordterm{\fieldstart{\ttrans{\valb_{1_{s}}}},
                \fieldend{\ttrans{\valb_{1_{e}}}} },
                \ldots,
                \recordterm{\ell_i = \ttrans{\vala_{m_i}}}_{i \in 1..l}
                \recordplus
                \recordterm{\fieldstart{\ttrans{\valb_{m_{s}}}},
                \fieldend{\ttrans{\valb_{m_{e}}}} },\\
                \recordterm{\ell_i = \ttrans{\vala_{{m + 1}_i}}}_{i \in 1..l}
                \recordplus
                \recordterm{\fieldstart{\ttrans{\valb_{{m + 1}_{s}}}},
                \fieldend{\ttrans{\valb_{{m + 1}_{e}}}} },
                \ldots
                \recordterm{\ell_i = \ttrans{\vala_{{n}_i}}}_{i \in 1..l}
                \recordplus
                \recordterm{\fieldstart{\ttrans{\valb_{{n}_{s}}}},
                \fieldend{\ttrans{\valb_{{n}_{e}}}} }
            }
                \el
        \]

        Suppose that the guard returns $\ttrue$ for records $1..m$, i.e.,
            $\puretmevaltwo{\tma \{ (\ell_i = \vala_{1_i})_{i \in 1..l} \}}{\ttrue}$, \ldots,
            $\puretmevaltwo{\tma \{ (\ell_i = \vala_{m_i})_{i \in 1..l} \}}{\ttrue}$
            with each $\valb_{i_{e}} = \forever$.

        Suppose that the guard returns false for the remainder.

        By
        Proposition~\ref{prop:tlinq-type-correct} and
        Lemma~\ref{lem:tlinq-trans-restrict} we have that
        \[
            \bl
                \puretmevaltwo{\ttrans{\tma} \{ (\ell_i = \ttrans{\vala_{1_i}})_{i \in 1..l} \}}{\ttrue}, \ldots,
                \puretmevaltwo{\ttrans{\tma} \{ (\ell_i = \ttrans{\vala_{m_i}})_{i \in 1..l} \}}{\ttrue}.
            \el
        \]
        and $\valb_{1_{e}} = \forever, \ldots, \valb_{m_{e}} = \forever$

        Thus we have by Lemma~\ref{lem:pure-for},
        Proposition~\ref{prop:tlinq-type-correct}, and
        and Lemma~\ref{lem:tlinq-trans-restrict} that
        \[
            \bl
                \var{affected} = \\
                \recordterm{\ell_i = \ttrans{\vala'_{1_i}}}_{i \in 1..l}
                \recordplus
                \recordterm{\fieldstart{\timevar}, \fieldend{\forever}},
                \ldots,
                \recordterm{\ell_i = \ttrans{\vala'_{m_i}}}_{i \in 1..l}
                \recordplus
                \recordterm{\fieldstart{\timevar}, \fieldend{\forever}}
            \el
        \]

        Note that
            $\var{affected} =
                \flatten{\bag{
                    \dbrow
                    {
                        \recordterm{\ell_i = \vala'_{1_i}}_{i \in 1..l}
                    }
                    { \timevar }
                    { \forever },
                    \ldots,
                    \dbrow
                    {
                        \recordterm{\ell_i = \vala'_{m_i}}_{i \in 1..l}
                    }
                    { \timevar }
                    { \forever }
                }}$.

        By applying the $\calcwd{update}$ operation (with appropriate uses of
        Proposition~\ref{prop:tlinq-type-correct} and
        Lemma~\ref{lem:tlinq-trans-restrict}), we have some
        $\db''$ with table $\tblvar$ containing:

        \[
            \bl
            \db'' =
            \bag{
                \recordterm{\ell_i = \ttrans{\vala_{1_i}}}_{i \in 1..l}
                \recordplus
                \recordterm{\fieldstart{\ttrans{\valb_{1_{s}}}},
                \fieldend{\timevar} },
                \ldots,
                \recordterm{\ell_i = \ttrans{\vala_{m_i}}}_{i \in 1..l}
                \recordplus
                \recordterm{\fieldstart{\ttrans{\valb_{m_{s}}}},
                \fieldend{\timevar} },\\
                \recordterm{\ell_i = \ttrans{\vala_{{m + 1}_i}}}_{i \in 1..l}
                \recordplus
                \recordterm{\fieldstart{\ttrans{\valb_{{m + 1}_{s}}}},
                \fieldend{\ttrans{\valb_{{m + 1}_{e}}}} },
                \ldots
                \recordterm{\ell_i = \ttrans{\vala_{{n}_i}}}_{i \in 1..l}
                \recordplus
                \recordterm{\fieldstart{\ttrans{\valb_{{n}_{s}}}},
                \fieldend{\ttrans{\valb_{{n}_{e}}}} }
            }
                \el
        \]

        and after performing the $\calcwd{insert}$ of $\var{affected}$, the table stands

        \[
            \bl
            \db''(t) =
            \bag{
                \recordterm{\ell_i = \ttrans{\vala_{1_i}}}_{i \in 1..l}
                \recordplus
                \recordterm{\fieldstart{\ttrans{\valb_{1_{s}}}},
                \fieldend{\timevar} },
                \ldots,
                \recordterm{\ell_i = \ttrans{\vala_{m_i}}}_{i \in 1..l}
                \recordplus
                \recordterm{\fieldstart{\ttrans{\valb_{m_{s}}}},
                \fieldend{\timevar} },\\
                \recordterm{\ell_i = \ttrans{\vala_{{m + 1}_i}}}_{i \in 1..l}
                \recordplus
                \recordterm{\fieldstart{\ttrans{\valb_{{m + 1}_{s}}}},
                \fieldend{\ttrans{\valb_{{m + 1}_{e}}}} },
                \ldots
                \recordterm{\ell_i = \ttrans{\vala_{{n}_i}}}_{i \in 1..l}
                \recordplus
                \recordterm{\fieldstart{\ttrans{\valb_{{n}_{s}}}},
                \fieldend{\ttrans{\valb_{{n}_{e}}}} }, \\
                \recordterm{\ell_i = \ttrans{\vala'_{1_i}}}_{i \in 1..l}
                \recordplus
                \recordterm{\fieldstart{\timevar}, \fieldend{\forever}},
                \ldots,
                \recordterm{\ell_i = \ttrans{\vala'_{m_i}}}_{i \in 1..l}
                \recordplus
                \recordterm{\fieldstart{\timevar}, \fieldend{\forever}},
            }
                \el
        \]

        which is equal to the database obtained by performing the $\tlinq$ update, as required.
    \end{proofcase}
\end{proof}

\vtranscorrect*
\begin{proof}
    By induction on the derivation of
    $\vtmevaltwo{\tma}{\retpair{\vala_{\mkwd{V}}}{\db_{\mkwd{V}}}}$.

    \begin{proofcase}{EV-SeqInsert}

        Assumption:

        \begin{mathpar}
            \inferrule
              {
                  \vtmevaltwo{\tma}{\retpair{\tblvar}{\db_1}} \\
                  \puretmevaltwo{\tmb}{\bag{\seq{\vala}}} \\\\
                  \forall \dbrow{\var{data}}{\var{start}}{\var{end}} \in \seq{\vala}. \var{start}
                  < \var{end} \\\\
                  \db_2 = \extendenv{\db_1}{\tblvar}{\db_1(\tblvar) \denotbagunion
                  \bag{\seq{\vala}}}
              }
              {
                \vtmevaltwo{\dbinsertseq{\tma}{\tmb}}{\retpair{()}{\db_2}}
              }
        \end{mathpar}

        Translation:
        \[
            \bl
                \vtrans{\dbinsertseq{\tma}{\tmb}} \\
                \quad \letintwo{\var{tbl}}{\vtrans{\tma}} \\
                \quad \letinone{\var{rows}} \\
                \qquad \forcomptwo{x}{\vtrans{\tmb}} \\
                \qqquad \bag{
                            \recordplustwo
                                {\etaexp{x}{\seq{\ell}}}
                                {\recordterm{
                                    \fieldstart{\project{x}{\var{start}}},
                                    \fieldend{\project{x}{\var{end}}}
                                }}
                        } \\
                \quad \calcwd{in} \\
                \quad \dbinsert{\var{tbl}}{\var{rows}}
            \el
        \]

        By the IH,
        $\tmevaltwo{\vtrans{\tma}}{\retpair{\tblvar}{\flatten{\db'}}}$.

        By Lemma~\ref{lem:vlinq-trans-pure},
        $\puretmevaltwo
            [\flatten{\db'}, \timevar]
            {\vtrans{\tmb}}
            {\bag{\seq{\vtrans{\vala}}}}$; since each $\vala_i$ is
        a value of type $\validtimety{\tya}$, we have that each $\vala_i$ must
        be of the form
        $\dbrow
            {\valb_{\var{data}}}
            {\valb_{\var{start}}}
            {\valb_{\var{end}}}$.

        Thus we have that by the definition of the translation on database rows
        and Lemma~\ref{lem:vlinq-base-trans}, each $\vtrans{\vala_i}$ must be of
        the form:

        $\recordterm
            {\fielddata{\valb_{\var{data}}},
             \fieldstart{\valb_\var{start}},
             \fieldend{\valb_{\var{end}}}}$

        By Lemma~\ref{lem:pure-for} the $\calcwd{for}$ comprehension evaluates to
        a bag of the form
        $\recordterm
            {\etaexp{x}{\seq{\ell}}}
         \recordplus
         \recordterm
            {\fieldstart{\valb_\var{start}}, {\fieldend{\valb_{\var{end}}}}}$.

        Note that this matches the definition of $\flatten{\vala_i}$.

        Therefore by \textsc{E-Insert}:

        \begin{mathpar}
            \inferrule
            {
                \tmevaltwo{\vtrans{\tma}}{\retpair{\tblvar}{\flatten{\db_1}}} \\
                \puretmevaltwo{\vtrans{\tmb}}{\bag{\seq{\flatten{\vala}}}}
            }
            {
                \tmevaltwo
                    {\dbinsert{\tma}{\tmb}}
                    {\retpair
                        {()}
                        {\extendenv{\flatten{\db_1}}{\tblvar}{\flatten{\db_1}(\tblvar)
                            \denotbagunion \seq{\flatten{\vala}}}} }
            }
        \end{mathpar}

        as required.
    \end{proofcase}

    \begin{proofcase}{EV-NonseqDelete}
        Assumption:
        \begin{mathpar}
            \inferrule
            {
                \tmevaltwo{\tma}{\retpair{\tblvar}{\db_1}}
                \\
                \db_2 =
                    {\extendenv
                        {\db_1}
                        {\tblvar}
                        { \bag{d \in \db(\tblvar) \mid
                            \puretmevaltwo{\tmb \{d / x \}}{\ffalse}} }
                    }
            }
            {
                \tmevaltwo
                    { \dbdeletenonseq{x}{\tma}{\tmb} }
                    {\retpair{()}{\db_2}}
            }
        \end{mathpar}

        Translation:

        \[
            \bl
            \vtrans{\dbdeletenonseqann{x}{\tma}{\tmb}} \\
            \quad
            {
            \blt
                \dbdelete{x}{\vtrans{\tma}}{\lift{x}{\vtrans{\tmb}}}
            \el
            } \\
            \text{where } \lift{x}{f} \defeq \\
                \quad
                (\fun{x}{f}) \app
                    \recordterm{
                        \fielddata{\etaexp{x}{\{\ell_i\}_{i}}},
                        \fieldstart{\project{x}{\var{start}}},
                        \fieldend{\project{x}{\var{end}}}}
            \el
        \]

        We want to show:
        \begin{mathpar}
            \inferrule
            {
                \tmevaltwo{\tma}{\retpair{\tblvar}{\flatten{\db_1}}}
                \\
                \db_2 =
                    {\extendenv
                        {\flatten{\db_1}}
                        {\tblvar}
                        { \bag{ d \in \flatten{\db_1}(\tblvar) \mid
                            \puretmevaltwo{\lift{x}{\vtrans{\tmb} \{ \vtrans{d}
                        / x \}}}{\ffalse}} }
                    }
            }
            {
                \tmevaltwo
                    { \dbdeletenonseq{x}{\vtrans{\tma}}{\lift{x}{\vtrans{\tmb}}} }
                    {\retpair{()}{\flatten{\db_2}}}
            }
        \end{mathpar}

        By the IH we have that
        $\tmevaltwo
            [\flatten{\db}, \timevar]
            {\tma}
            {
                \retpair
                    {\tblvar}
                    {\flatten{\db_1}}
            }$.

        By \textsc{TV-NonseqDelete}, we have that
        $\tseq{\tyenv, x : \validtimety{\tya}}{\tmb}{\boolty}{\pure}$.

        Since $\puretmevaltwo{\tmb \{ d / x \}}{\vala_{\var{res}}}$
        where $d$ is some row $\dbrow{\recordterm{\ell_i =
        \vala_i}_i}{\valb_{\var{start}}}{\valb_{\var{end}}}$,
        it follows by Lemma~\ref{lem:vlinq-trans-pure} that
        $\puretmevaltwo{\vtrans{\tmb} \{ \vtrans{d} / x
            \}}{\vtrans{\vala_{\var{res}}}}$.

        By Lemma~\ref{lem:vlinq-base-trans},
        $\vtrans{\vala_{\var{res}}} = \vala_{\var{res}}$.

        Consider the definition of $\mkwd{lift}$. We can show:

        \begin{mathpar}
            \inferrule
            {
                \vtrans{\tmb}
                \puretmevaltwo{\fun{x}{\vtrans{\tmb}}}{\fun{x}{\vtrans{\tmb}}}
                \\
                \puretmevaltwo
                {
                    \recordterm{
                        \fielddata{\etaexp{x}{\{\ell_i\}_{i}}},
                        \fieldstart{\project{x}{\var{start}}},
                        \fieldend{\project{x}{\var{end}}}}
                }
                {
                    \recordterm{
                        \fielddata{\recordterm{\ell_i = \vala_i}_i},
                        \fieldstart{\valb_{\var{start}}},
                        \fieldend{\valb_{\var{end}}}
                    }
                }
                \\
                \puretmevaltwo{\vtrans{\tmb} \{ \vtrans{d} / x
                    \}}{\vala_{\var{res}}}
            }
            {
                \puretmevaltwo
                {
                    (\fun{x}{\vtrans{\tmb}}) \app
                        \recordterm{
                            \fielddata{\etaexp{x}{\{\ell_i\}_{i}}},
                            \fieldstart{\project{x}{\var{start}}},
                            \fieldend{\project{x}{\var{end}}}}
                }
                {
                    \vala_{\var{res}}
                }
            }
        \end{mathpar}
        since
        $
            \recordterm{
                \fielddata{\recordterm{\ell_i = \vala_i}_i},
                \fieldstart{\valb_{\var{start}}},
                \fieldend{\valb_{\var{end}}}
            }
            =
            \vtrans{d}
        $.

        Therefore we know that the lifted translated predicate will evaluate to
        the same value as the source predicate in $\vlinq$. It follows that the
        $\calcwd{delete}$ operation will affect the same rows, resulting in the
        same (flattened) database, as required.
   \end{proofcase}

    \begin{proofcase}{EV-NonseqUpdate}
        Follows the same reasoning as \textsc{EV-NonseqDelete}.
    \end{proofcase}

    \begin{proofcase}{EV-SeqUpdate}
        \newcommand{\vstart}{\vala_{\var{start}}}
        \newcommand{\vend}{\vala_{\var{end}}}
        As before, it suffices to reason about each case of the update in turn.

        By the IH, we have that
        $\tmevaltwo{\vtrans{\tmc}}{\retpair{\tblvar}{\flatten{\db'}}}$.

        By Lemma~\ref{lem:vlinq-trans-pure}, we have that:
        \begin{itemize}
            \item $\puretmevaltwo{\tma_1}{\vtrans{\vstart}}$
            \item $\puretmevaltwo{\tma_2}{\vtrans{\vend}}$
        \end{itemize}

        and by Lemma~\ref{lem:vlinq-base-trans},
        $\vtrans{\vala_{\var{start}}} = \vala_{\var{start}}$ and
        $\vtrans{\vala_{\var{end}}} = \vala_{\var{end}}$.

        Therefore let $\var{aStart} = \vstart$ and $\var{aEnd} = \vend$.

        Since the two $\calcwd{insert}$ statements occur after the
        $\calcwd{update}$ statement, the records in $\var{startRows}$ and
        $\var{endRows}$ will not affect the $\calcwd{update}$.
        It therefore suffices to prove the more general statement by considering
        a single row at a time.

        Therefore, we now reason by cases on $\mkwd{upd}$ with a database of the
        form:
            $\db' = [\tblvar \mapsto
            \bag{\dbrow{\varv}{\var{start}}{\var{end}}}]$.

        Since we know databases are well formed, we know that for each row,
        $\var{start} < \var{end}$.
        By the premises, we also know that $\vstart < \vend$.

        \begin{subcase}{Case 1:
                $\puretmevaltwo{\subst{\tma_3}{\var{v}}{x}}{\ttrue}$
                and $\vala_{\var{start}} \le \var{start}$
                and $\vala_{\var{end}} \ge \var{end} $}

                Here we have that $\vstart \le \var{start} < \var{end} \le
                \vend$.

                Since $\vstart \le \var{start}$, we have that
                $\project{x}{\var{start}} \not< \vstart$. Therefore,
                $\var{startRows} = \bag{~}$.

                Since $\vend \ge \var{end}$, we have that
                $\project{x}{\var{end}} \not< \vend$. Therefore,
                $\var{endRows} = \bag{~}$.

                The $\calcwd{update}$ predicate will match since by
                Proposition~\ref{prop:vlinq-type-correct} and
                Lemma~\ref{lem:vlinq-trans-restrict},
                $\puretmevaltwo
                    {\restrict{x}{\{\ell_i\}_{i \in I}}{\vtrans{\tma_3}}}
                    {\ttrue}$.

                By
                Proposition~\ref{prop:vlinq-type-correct}, and
                Lemmas~\ref{lem:vlinq-trans-restrict}
                and~\ref{lem:vlinq-base-trans},
                $\puretmevaltwo
                    {\restrict{x}{\{\ell_i\}_{i \in I}}{\vtrans{\tmb_j}}}
                    {\valb_j}$ for all $j \in J$.

                In this case:
                \begin{itemize}
                    \item $\greatest{\project{x}{\var{start}}, \vstart} =
                        \project{x}{\var{start}}$
                    \item $\least{\project{x}{\var{end}}, \vend} =
                        \project{x}{\var{end}}$
                \end{itemize}

                Thus after the update and two null inserts we have:

                $\db'' = [\tblvar \mapsto
                \bag{
                    \recordterm{\ell_i = \vala_i}_{i \in I \without J}
                    \recordplus
                    \recordterm{\ell_j = \valb_j}_{j \in J}
                    \recordplus
                    \recordterm{\fieldstart{\var{start}}, \fieldend{\var{end}}}
                 }]$.

                 which is equal to
                 $
                    [\tblvar \mapsto
                    \flatten{
                        \dbrow{\recordwithtwo{\varv}{\seq{\ell = \valb}}}{\var{start}}{\var{end}}
                    }
                 $
                 as required.
        \end{subcase}

        \begin{subcase}{Case 2:
            $\puretmevaltwo{\subst{\tma_3}{\var{v}}{x}}{\ttrue}$
            and $\vala_{\var{start}} \le \var{start}$ and
            $\vala_{\var{end}} < \var{end}$}

            In this case we have that $\vstart \le \var{start} \le \vend <
            \var{end}$.

            Since $\project{x}{\var{start}} < \vstart$, the $\var{startRows}$
            predicate does not hold and so $\var{startRows} = \bag{~}$.

            However, $\project{x}{\var{start}} < \vend$ and
            $\project{x}{\var{end}} > \vend$, so
            $\bag{
                    \etaexp{x}{\{\ell_j\}_{j \in J}}
                    \recordplus
                    \recordterm{\fieldstart{\vend},
                        \fieldend{\project{x}{\var{end}}}}
            } \puretmeval
                    \bag{\varv
                    \recordplus
                    \recordterm{\fieldstart{\vend}, \fieldend{\var{end}}}} =
                    \var{endRows}$.

            Again the $\calcwd{update}$ predicate matches; this time we have:

            \begin{itemize}
                \item $\greatest{\var{start}, \vstart} = \var{start}$
                \item $\least{\var{end}, \vend} = \vend$
            \end{itemize}

            and after the $\calcwd{update}$ and $\calcwd{insert}$ our database is:

            \[
                [ \tblvar \mapsto
                    \bag{
                        (\varv \recordplus
                         \recordterm{\fieldstart{\vend}, \fieldend{\var{end}}}),
                            ((\ell_i = \vala_i)_{i \in I \without J}
                            \recordplus
                            \recordterm{\ell_j = \valb_j}{j \in J}
                            \recordplus
                            (\fieldstart{\var{start}}, \fieldend{\vend}))
                    }
                ]
            \]

            which is equal to

            \[
                [ \tblvar \mapsto
                    \flatten{\bag{\dbrow
                            {\recordwithtwo{\varv}{(\ell_j = \valb_j)_{j \in J}}}
                            {\var{start}}
                            {\vend},
                         \dbrow{\varv}{\fieldstart{\vend}}{\fieldend{\var{end}}}
                     }}
                ]
            \]

            as required.
        \end{subcase}

        \begin{subcase}{Case 3:
            $\puretmevaltwo{\subst{\tma_3}{\var{v}}{x}}{\ttrue}$
            and $\vala_{\var{start}} > \var{start}$  and
            $\vala_{\var{end}} < \var{end}$}

            Here we have that $\var{start} < \vstart < \vend < \var{end}$.

            By similar reasoning to the above case (we omit references to
            specific lemmas, which are as above), we have that:

            \begin{itemize}
                \item The $\mkwd{where}$ clause of $\var{startRows}$ evaluates
                    to $\ttrue$ as both $\project{x}{\var{start}} < \vstart$ and
                    $\project{x}{\var{end}} > \vstart$. Thus $\var{startRows} =
                    \bag{\varv \recordplus \recordterm{\fieldstart{\var{start}},
                        \fieldend{\vstart}}}$
                \item The $\mkwd{where}$ clause of $\var{endRows}$ evaluates
                    to $\ttrue$ as both $\project{x}{\var{start}} < \vend$ and
                    $\project{x}{\var{end}} > \vend$. Thus $\var{endRows} =
                    \bag{\varv \recordplus
                        \recordterm{\fieldstart{\vend}, \fieldend{\var{end}}}}$.
            \end{itemize}

            Similarly, the $\calcwd{update}$ will apply since
            $\project{x}{\var{start}} < \vend$ and
            $\project{x}{\var{end}} > \vstart$.

            We have in this case that
            $\greatest{\var{start}, \vstart} = \vstart$ and
            $\least{\var{end}}{\vend} = \vend$.

            Thus after the update and two inserts, our final database is equal to:

            \[
                [\tblvar \mapsto \flatten{\bag{
                        \dbrow{\varv}{\var{start}}{\vstart},
                        \dbrow{\valb}{\vstart}{\vend},
                        \dbrow{\varv}{\vend}{\var{end}}}
                        }]
            \]

            as required.
        \end{subcase}

        \begin{subcase}{Case 4:
            $\puretmevaltwo{\subst{\tma_3}{\var{v}}{x}}{\ttrue}$
            and
            $\vala_{\var{start}} > \var{start}$
            and
            $\vala_{\var{end}} \ge \var{end}$}

            Similar to Case 2, except there is a record produced in
            $\var{startRows}$ instead of $\var{endRows}$.
        \end{subcase}

        \begin{subcase}{Case 5: other cases}

            This `catch all' case boils down to three sub-subcases:

            \begin{subsubcase}{$\puretmevaltwo{\tma_3}{\ffalse}$}
                In this case, the $\mkwd{where}$ clause of the two queries would
                evaluate to $\ffalse$, so we would have
                $\var{startRows} = \var{endRows} = \bag{~}$. Furthermore, for
                the same reason, the
                $\calcwd{update}$ would not apply. Therefore, the database would
                be unaltered.
            \end{subsubcase}
            \begin{subsubcase}{$\vend < \var{start}$}
                Here we have that $\vstart < \vend \le \var{start} < \var{end}$.

                In this case, the $\mkwd{where}$ clause for $\var{startRows}$
                would evaluate to $\ffalse$ as $\project{x}{\var{start}} \not<
                \vstart$, and the $\mkwd{where}$ clause for $\var{endRows}$
                would evaluate to $\ffalse$ since $\project{x}{\var{start}}
                \not< \vend$.
                Thus, $\var{startRows} = \var{endRows} = \bag{~}$.

                The $\calcwd{update}$ would not apply since
                $\project{x}{\var{start}} \not< \vend$.

                Therefore, the database will be unaltered.
            \end{subsubcase}
            \begin{subsubcase}{$\vstart > \var{end}$}
                Here we have that $\var{start} < \var{end} \le \vstart < \vend$.

                Again, the $\mkwd{where}$ clause of $\var{startRows}$ would
                evaluate to $\ffalse$ since $\project{x}{\var{end}} \not>
                \vstart$,
                and the $\mkwd{where}$ clause of $\var{endRows}$ would evaluate
                to $\ffalse$ since $\project{x}{\var{end}} \not> \vend$.
                Thus, $\var{startRows} = \var{endRows} = \bag{~}$.

                The update would not apply since $\project{x}{\var{end}} \not>
                \vstart$.

                Therefore, the database will be unaltered.
            \end{subsubcase}
        \end{subcase}
    \end{proofcase}

    \begin{proofcase}{EV-SeqDelete}
        Follows the same reasoning as \textsc{EV-SeqUpdate}. The main difference
        is that the final $\calcwd{update}$ is replaced with a
        $\calcwd{delete}$, so where the translation results in an updated
        record in \textsc{EV-SeqUpdate}, the translation results in the absence
        of a row here.
    \end{proofcase}

\end{proof}
 \clearpage
\section{Links examples}\label{appendix:code}

New count values are added to the table with a \emph{sequenced
insert}, using the upload time as the start time.

\begin{lstlisting}
fun insertNewData (new) {
  vt_insert sequenced covid_data
    values (subcat, weekdate, count)
      [withValidity(
        (subcat = new.subcat, weekdate = new.weekdate,
         count = new.count),
      upload_time, forever)]
}
\end{lstlisting}

Accepted updates are added using sequenced updates.  Here
\lstinline{accepted_updates} is a list of updated values that have
been approved by the user. For each element of the list, a sequenced
update is made.

\begin{lstlisting}
fun updateData (accepted_updates) {
  for (x <- accepted_updates)
    [update sequenced (y <-v- covid_data)
      between (x.time_added, forever)
      where (x.subcat==y.subcat && x.weekdate==y.weekdate)
    set (count = x.new_value) ]
}
\end{lstlisting}

For update provenance queries of individual counts,
a self join is computed over the \lstinline{subcategory} and
\lstinline{week} fields of the valid time table to provide a nested
result table where each count is associated with a list
of count values and their associated start and end time
information. This is a nonsequenced query because the time period
information is explicitly added to the result table.

\begin{lstlisting}
query nested {for (x <- vtCurrent(covid_data))
  [(subcat = x.subcat, weekdate = x.weekdate,
    count = x.count, mods =
      for (y <-v- covid_data)
        where (x.weekdate==vtData(y).weekdate &&
               x.subcat==vtData(y).subcat &&
          [(vtf=vtFrom(y),vtt=vtTo(y),
            count=vtData(y).count)])]}
\end{lstlisting}

\end{document}